\documentclass[sigconf]{acmart}
\AtBeginDocument{%
  }

\setcopyright{acmlicensed}
\copyrightyear{2026}
\acmYear{2026}
\acmDOI{XXXXXXX.XXXXXXX}
\acmConference[SIGMOD '27]{the ACM SIGMOD/PODS International Conference on Management of Data}{June 13--19,
  2027}{Huntington Beach, CA}
\acmISBN{978-1-4503-XXXX-X/2018/06}

\usepackage[linesnumbered,ruled,vlined]{algorithm2e}
\usepackage{amsmath}
\usepackage{amsthm}
\usepackage{dsfont}
\usepackage[inline]{enumitem}
\usepackage{makecell}
\usepackage{MnSymbol}
\usepackage{multirow}
\usepackage{subcaption}
\usepackage{xcolor}

\makeatletter
\newcommand{\bparagraph}{%
  \@startsection{paragraph}{4}{\z@}%
    {0.3em}
    {-0.3em}%
    {\normalfont\normalsize\bfseries}
}
\makeatother

\newtheorem{definition}{Definition}
\newtheorem{lemma}{Lemma}

\SetKwComment{Comment}{$\triangleright$\ }{}

\SetCommentSty{mycommfont}

\DeclareMathOperator*{\argmin}{arg\,min}

\begin{document}

\title[Fast Label-Filtering Approximate Nearest Neighbor Search via Progressive Label Set Stratification]{Fast Label-Filtering Approximate Nearest Neighbor Search via Progressive Label Set Stratification}

\author{Ziqi Wang}
\orcid{0009-0003-2027-3524}
\affiliation{%
    \department{State Key Laboratory for Novel Software Technology}
    \institution{Nanjing University \country{China}}
}
\email{ziqiw.nju@gmail.com}

\author{Jingzhe Zhang}
\orcid{0009-0001-7062-9096}
\affiliation{%
    \department{State Key Laboratory for Novel Software Technology}
    \institution{Nanjing University \country{China}} 
}
\email{jzzhang.nju@gmail.com}

\author{Shuo Shen}
\orcid{0009-0004-4035-7570}
\affiliation{%
    \department{State Key Laboratory for Novel Software Technology}
    \institution{Nanjing University \country{China}} 
}
\email{sshen.nju@gmail.com}

\author{Wei Hu}
\orcid{0000-0003-3635-6335}
\authornote{Corresponding author}
\affiliation{
    \department{State Key Laboratory for Novel Software Technology}
    \department{National Institute of Healthcare Data Science}
    \institution{Nanjing University \country{China}}
}
\email{whu@nju.edu.cn}


\begin{abstract}
Approximate nearest neighbor search (ANNS) retrieves the most similar vectors to a query vector in high-dimensional space.
Label-filtering ANNS (LFANNS) extends ANNS with a label filter that the labels of base vectors must satisfy a set relation (e.g., equality, containment, or overlap) with the query labels.
Existing LFANNS indices suffer from inconsistent performance across different filter types and degraded scalability under varying label scale and distribution.
In this paper, we define label-stratified similarity graph (LSSG), where edges connect neighboring vectors whose label sets fall within stratified similarity thresholds.
To implement LSSG efficiently, we design an incremental insertion algorithm to prune redundant edges in both vector and label spaces, and leverage a MinHash structure to ensure scalability for large-scale labels.
We analyze stepwise probabilities under explicit label models and explain why stricter label tiers reduce ineffective in-filtering expansions.
Benchmark experiments show that LSSG achieves ideal optimality for equality queries, and 1.06x--92.9x and 1.08x--84.1x faster than the best competing index for containment and overlap, respectively, in query speed with identical accuracy and 0.35x index size.
\end{abstract}

\ccsdesc[500]{Information systems~Nearest-neighbor search}
\ccsdesc[500]{Information systems~Specialized information retrieval}
\keywords{approximate nearest neighbor search, graph-based index, vector database, label set similarity}


\received{17 April 2026}
\received[revised]{20 August 2026}
\received[accepted]{12 September 2026}

\maketitle

\section{Introduction}
\label{sec:intro}

Approximate nearest neighbor search (ANNS) is a fundamental operation for semantic query answering, used in vector database management systems (VDBMS)~\cite{milvus-vdb,adbv-vdb,pase-vdb,faiss,vdb-survey}, search engines~\cite{spann,spfresh,jdvisual,infinity}, retrieval-augmented generation (RAG) applications~\cite{rag-survey1,rag-survey2,alayaDB}, etc.
Enormous high-dimensional vectors are generated from texts, images, and audios by modern embedding models~\cite{clip,siglip,glove}.
Given a query vector, rather than scanning all base vectors to obtain the exact $k$ nearest neighbors, ANNS indices enable to retrieve most of these neighbors with significantly lower latency.
Particularly, graph-based ANNS indices~\cite{hnsw,nsg,nssg,filtered-diskann,hscg,almg} are popular for their fast query speed and high accuracy.

With growing demands for customized queries, label-filtering ANNS (LFANNS) has become prominent for semantic-constrained vector search beyond basic similarity retrieval~\cite{nhq-nips23,unifilter,acorn,infinity}.
Label set predicates are commonly organized around three canonical \emph{positive} semantics~\cite{ung,rwalks}.
For example, in a medical image retrieval system~\cite{eg-imgcontent}, a physician may query a CT image labeled ``lung'' and ``nodule''.
An \emph{equality} filter strictly matches images labeled only ``lung'' and ``nodule'', excluding cases with additional findings or comorbidities.
This filter is useful for scenarios requiring precise and controlled subsets~\cite{eg-imgretrieval}.
A \emph{containment} filter retrieves images whose label sets include both query terms (e.g., also annotated with ``CT'', ``solid'', or ``cancer''), ensuring a comprehensive coverage of relevant cases.
An \emph{overlap} filter returns images sharing at least one query label (e.g., lung-related images without nodules, or nodule images from other regions), supporting association analysis.

There are three mainstream solutions for LFANNS~\cite{filterbench_sigmod}.
\begin{enumerate*}
    \item \emph{Pre-filtering} scans all data objects and only computes vector similarity between the query vector and the filtered subset.
    Since no vector index is constructed for the subset, the linear scan becomes inefficient when the subset is large (e.g., tens of thousands of vectors).
    
    \item \emph{Post-filtering} operates on an ANN index constructed over the entire vector dataset.
    It first retrieves the most similar vectors using a retrieval size larger than $k$, and then retains $k$ vectors whose label sets satisfy the filter.
    It can be ineffective when the filter is highly selective and fewer than $k$ vectors are returned.
    
    \item \emph{In-filtering} performs filter checks during index scanning.
    It is commonly adopted by graph-based indices~\cite{nsg,hnsw,filtered-diskann}, with filter checks executed when selecting next-hop candidates.
    Similar to post-filtering, graph traversal may become less connected when few neighbors satisfy the filter and may terminate early at local optima~\cite{acorn}.
\end{enumerate*}
To overcome the drawbacks of pre-filtering and post-filtering, as well as the early termination of in-filtering graph traversal, novel LFANNS indices~\cite{nhq-nips23,filtered-diskann,caps,ung,rwalks,tfanns,eli,unifilter,curator} are proposed with label partition or label enhancement to improve query performance.
However, two critical challenges in LFANNS remain largely unsolved.

\emph{Challenge 1: Inconsistent support for diverse query semantics.}
It is difficult to design a unified index due to fundamentally different label set relations: equality queries require exact matches, containment queries require super/sub-sets, and overlap queries require partial intersections.
An LFANNS index that is effective for one filter type often does not transfer well to others.
For example, indices based on set containment~\cite{eli} can support equality filters only by maintaining isolated auxiliary structures for each unique label set, which leads to severe index fragmentation and space overhead.
Likewise, overlap queries are often handled by iterative containment searches, whose cost grows quickly as the query label set expands~\cite{eli,ung}.
As far as we know, no existing index provides a unified solution to ensure consistent efficiency and theoretical guarantees across all semantics.

\begin{figure}[t]
    \centering
    \includegraphics[width=.95\linewidth]{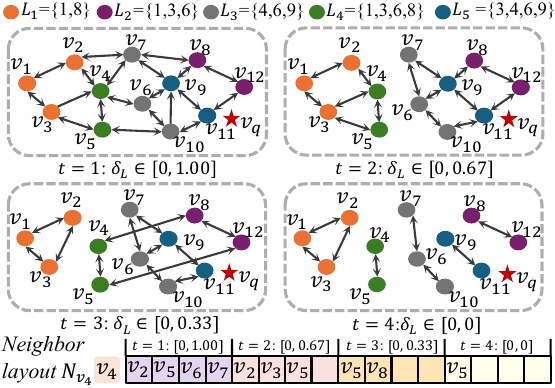}
    \caption{Label-stratified similarity graph (LSSG)}
    \Description{Label-stratified similarity graph (LSSG)}
    \label{fig:LSSG_structure}
\end{figure}

\emph{Challenge 2: Vulnerable to varying label characteristics.}
Another key challenge in multi-label indexing is to remain robust because the number of distinct label sets grows exponentially as the label space expands.
When this number grows to millions, existing indices often face a critical scalability bottleneck: either index size increases dramatically~\cite{eli,rwalks}, or query latency becomes prohibitively high~\cite{ung,tfanns}.
The problem is further aggravated by label distribution.
Many existing indices implicitly assume that labels follow skewed distributions (e.g., power-law)~\cite{tfanns,ung,eli,rwalks,unifilter}, where a large number of data objects share a small set of common labels.
This assumption enables efficient label sharing and therefore more compact index structures.
However, when the distribution goes toward uniform, such sharing opportunities diminish.
As a result, indices that rely on label sharing~\cite{eli} tend to become fragmented or require extensive duplication to remain effective~\cite{ung}.
This creates a rigid trade-off between space consumption and query efficiency, limiting robustness under real-world settings.

We propose \emph{Label-Stratified Similarity Graph} (LSSG), a multi-tier proximity graph that supports equality, containment, and overlap filters within a single index.
The key idea is to interpret label set relations through a \emph{distance} view: different query semantics correspond to searching among vertices with different degrees of label set similarity.
Accordingly, LSSG organizes all vertices into tiers stratified by label set distance.
Each tier enforces a label distance bound, and edges are created only between vertices whose label sets satisfy that bound.
Upper tiers use tighter bounds, producing nested similarity ranges (e.g., $[0,1]\supset[0,0.67]\supset[0,0.33]\supset[0,0]$ in Fig.~\ref{fig:LSSG_structure}).
The bottom tier ($[0,1]$) is label-agnostic and preserves global navigability, while the top tier ($[0,0]$) connects only identical label sets for exact matching. 
Intermediate tiers progressively tighten label consistency for containment and overlap.
Within each tier, vector similarity picks the final neighbors among label-eligible candidates, enabling progressive exploration at query.

Let us consider a containment query with label set $\{3,6\}$ and query vector $v_q$ (red star in Fig.~\ref{fig:LSSG_structure}).
Starting from a random \emph{valid} entry (e.g., $v_4$ whose label set $\{1,3,6,8\}$ contains $\{3,6\}$), a bottom-tier traversal behaves like a vanilla proximity graph. 
It greedily follows vector-similarity edges and may move toward $v_5$, but later stalls because the most promising next hops (e.g., $v_6,v_7,v_{10}$) are \emph{label-incompatible} and cannot be expanded.
The search is therefore trapped in a local minima over the filtered subset.
LSSG mitigates this by \emph{escalating to a stricter tier} when progress stalls.
Higher tiers contain more label-consistent edges, which can reveal an in-filter bridge that is absent (or far less likely) in the label-agnostic tier.
For example, LSSG can take $v_5\rightarrow v_{12}$ at tier 3 and reach the true nearest neighbor via $v_{12}\rightarrow v_{11}$.
Overall, lower tiers preserve global reachability, while higher tiers provide label-consistent shortcuts to escape local minima in selective regions.

In summary, our main contributions are outlined below:
\begin{itemize}[left= 0pt]
    \item
    We propose LSSG, a unified multi-tier proximity graph for equality, containment, and overlap filters.
    The progressive label set stratification unifies the semantics of the three filters in a single graph through tier-ordered in-filtering traversal.
    
    \item
    We introduce a stepwise metric for in-filtering traversal and analyze filter-valid expansion probabilities.
    Equality is preserved exactly at the top tier, while containment and overlap admit lower bounds under explicit label models.
    These bounds strengthen monotonically with label set similarity, explaining why intermediate tiers improve filtered navigation.
    We also derive a conditional logarithmic bound on expected query cost with respect to the filtered subset size.
    
    \item We conduct experiments on eight real-world datasets.
    LSSG improves index size by 16\% and indexing time by 22\% over the most build-efficient baseline~\cite{ung}.
    At matched accuracy, it is up to 1.06x--92.9x (avg.\ 19.3x) faster for containment and up to 1.08x--84.1x (avg.\ 15.5x) faster for overlap than the strongest competitor~\cite{eli}, while using 0.35x index size and 0.87x indexing time.
    Evaluation also validates scalability to hundreds of thousands of labels and millions of label sets, robustness to label distribution drift, effective tier utilization, and competitive performance against popular ANNS systems.
    \emph{Source code, experiment data, and appendix are available at \url{https://github.com/nju-websoft/LSSG}}.
\end{itemize}

\section{Preliminaries}
\subsection{Problem Formulation} \label{sec:problem}

\begin{definition}[ANNS]
Given a $d$-dimensional vector dataset $\mathcal{V}$ measured by a distance function $\delta_v$, for a query vector $v_q\in \mathbb{R}^d$, ANNS seeks to find a result set $\mathcal{P}\subset \mathcal{V}$ of $k$ neighbors ($k\ll |\mathcal{V}|$) to minimize $\sum_{i=1}^{k}\delta_v(v_i,v_q), \forall v_i\in \mathcal{P}$.
\end{definition}

The main distinction between ANNS and exact $k$ nearest neighbor search is that non-closest vectors are allowed in trade for higher efficiency (e.g., queries per second (QPS)).
The approximation accuracy is usually measured by recall@$k$, defined by $\frac{|\mathcal{P}\cap\mathcal{T}|}{k}$, where $\mathcal{T}$ denotes the set of $k$ exact nearest neighbors.

\begin{table}
\centering
\caption{Frequently-used notations}
\label{tab:notation}
{\small
\setlength{\tabcolsep}{3pt}
\begin{tabular}{cp{6.5cm}}
\toprule
    \textbf{Notations} & \textbf{Descriptions} \\
\midrule
    $\mathcal{D} =\{\mathcal{V},\mathcal{L}\}$ & Hybrid dataset $\mathcal{D}$ consisting of vector dataset $\mathcal{V}$ and label dataset $\mathcal{L}$, with mapping function $l:\mathcal{V}\rightarrow\mathcal{L}$\\
    $\delta_v(v_1,v_2)$    & Distance between vectors $v_1$ and $v_2$ \\
    $\delta_L(L_1, L_2)$   & Distance between label sets $L_1$ and $L_2$ \\
    $f$           & Label filter. $f_e$: equality; $f_c$: containment; $f_o$: overlap\\
    $N_v^t$       & Neighbor set of $v$ at tier $t$\\
    $c_i^t,n_i^t$ & $i^\text{th}$ candidate ($c$) and neighbor ($n$) in tier~$t$\\
    $T$           & Number of tiers\\
    $\omega_c$    & Index parameter: beam width in insertion\\   
    $m$           & Index parameter: maximum out-degree\\
    $\phi$        & Index parameter: IVF scanning budget\\
    $\tau,\beta$  & Index parameter: MinHash signature length and bands\\
\bottomrule
\end{tabular}}
\end{table}

\begin{definition}[Label Set and Filter Type] \label{def:labelfilter}
Given a universe of labels $A=\{a_1, a_2, \ldots, a_{|A|}\}$, where $a_j$ is a basic type (e.g., integer, text), a label set dataset is defined by $\mathcal{L}=\{L_1,L_2, \ldots,L_{|\mathcal{L}|}\}, \forall L_i \subseteq A$.
Given a query label set $L_q\subseteq A$, there are three label filter ($f$) types:
\begin{itemize}
\item Equality filter: $f_e(L_i\,|\,L_q)=\mathds{I}(L_q=L_i)$, 
\item Containment filter: $f_c(L_i\,|\,L_q)=\mathds{I}(L_q \subseteq L_i)$, 
\item Overlap filter: $f_o(L_i\,|\,L_q)=\mathds{I}(L_q\cap L_i \neq \emptyset)$,
\end{itemize}
where $\mathds{I}(\cdot)$ is an indicator that equals 1 if true and 0 otherwise.
\end{definition}

Label-filtering ANNS is to find the most similar vectors whose corresponding label sets satisfy a given label filter.

\begin{definition}[LFANNS]
Given a hybrid dataset $\mathcal{D}=\{\mathcal{V}, \mathcal{L}\}$, $l(\cdot):\mathcal{V}\rightarrow \mathcal{L}$ denotes a mapping function for the corresponding label set of a vector in $\mathcal{L}$. 
Given a query vector $v_q\in \mathbb{R}^d$ and a query label set $L_q\subseteq A$, LFANNS aims to find a result set $\mathcal{R}\subset \mathcal{V}$ of $k$ neighbors ($k\ll|\mathcal{V}|$) to minimize $\sum_{i=1}^k \delta_v(v_i, v_q)$, $\forall v_i \in \mathcal{R} \wedge f(l(v_i)\,|\,L_q)=1$, where $f \in \{f_e,f_c,f_o\}$.
\end{definition}

Like ANNS, the accuracy of LFANNS is measured by $\text{recall}@k$ as well.
The \emph{filter ratio} is defined as the number of passed items over the size of the dataset, formally $\frac{\sum_{v_i\in \mathcal{V}}f(l(v_i)\,|\,L_q)}{|\mathcal{V}|}$.
Query workloads with high selectivity typically exhibit a small filter ratio.

Frequently used notations are present in Table~\ref{tab:notation}.

\begin{table}
    \centering
    \caption{Comparison of LFANNS indices}
    {\small
    \setlength{\tabcolsep}{0pt}
    \begin{tabular*}{\columnwidth}{@{\hspace{3pt}\extracolsep{\fill}}lccccccc@{\hspace{3pt}}}
    \toprule
       \textbf{Indices} & \textbf{Mech.} & \textbf{Sem.} & \textbf{Robt.} & \textbf{Guar.} & \textbf{Vers.} & \textbf{Fidl.} & \textbf{Cmpl.} \\
    \midrule
       NHQ & Part. & $f_e$ & $\times$ & $\times$ & $\times$ & $\times$ & $\times$\\
       Vamana & Enhc. & $f_c$ & $\times$ & $\times$ & $\times$ & $\times$ & $\times$\\
       CAPS & Part. & $f_e$ & $\times$ & $\times$ & $\times$ & $\times$ & $\checkmark$\\
       UniFilter & Part. & $f_c,f_o$ & $\checkmark$ & $\checkmark$ & $\times$ & $\times$ & $\checkmark$ \\
       Curator & Part. & $f_c,f_o$ & $\checkmark$ & $\times$ & $\times$ & $\times$ & $\checkmark$\\
       Packing & Enhc. & $f_c$ & $\checkmark$ & $\times$ & $\checkmark$ & $\times$ & $\times$\\
       RWalks & Enhc. & $f_c$ & $\times$ & $\checkmark$ & $\checkmark$ & $\times$ & $\times$\\
       UNG & Enhc. & $f_e,f_c$ & $\times$ & $\times$ & $\checkmark$ & $\checkmark$ & $\checkmark$\\
       ELI & Part. & $f_c,f_o$ & $\times$ & $\checkmark$ & $\checkmark$ & $\times$ & $\times$\\
       \midrule
       LSSG & Enhc. & $f_e,f_c,f_o$ & $\checkmark$ & $\checkmark$ & $\checkmark$ & $\checkmark$ & $\checkmark$\\
    \bottomrule
    \end{tabular*}}

    \vspace{2pt}
    \parbox{\linewidth}{
    \small
    \textbf{Notes.}
    Mech.: mechanism (\emph{Part.}: label partition; \emph{Enhc.}: label enhancement).
    Sem.: query \emph{semantics} supported consistently by design
    (cf.\ Challenge~1).
    Robt.: robustness (cf.\ Challenge~2).
    Guar.: whether theoretical guarantees are provided.
    Vers./Fidl./Cmpl.: the three challenges in \cite{ung}:
    \emph{Versatility}: base and query label sets contain an arbitrary number of labels;
    \emph{Fidelity}: comparisons are only made to vectors whose label sets pass the filter;
    \emph{Completeness}: the index returns sufficient $k$ vectors.
    }
    \label{tab:lfann-comp}
\end{table}

\subsection{Related Work}

\bparagraph{ANNS and Filtering ANNS.}
There are mainly four categories of ANNS indices: tree-based~\cite{tree-based}, LSH-based~\cite{lsh,lsh-lcc,db-lsh2,det-lsh}, clustering-based~\cite{cracking-ivf,ada-ivf,pq,rbq,exrbq}, and graph-based~\cite{hnsw,nsg,nssg,filtered-diskann,hscg,almg}.
Graph-based indices are popular for fast and accurate query performance~\cite{graph-survey,almg}, and their efficiency mainly relies on the monotonic search property of relative neighborhood graphs (RNG)~\cite{nsg,hnsw,graph-survey}, where edge pruning preserves this property by eliminating redundant connections based on neighborhood dominance~\cite{cng,tau-mg}.

Filtering ANNS includes four categories:
\begin{enumerate*}
    \item Vector databases and some ANNS libraries~\cite{vbase-vdb,pase-vdb,milvus-vdb,faiss,sieve,hqi} employ cost estimation to select between pre-filtering and post-filtering strategies.
    \item Predicate-agnostic indices~\cite{acorn,bench-agnostic,navix} support arbitrary filters without requiring attribute information during index construction.
    \item Range-filtering ANNS indices~\cite{serf,irange,dsg,digra,wow} are specifically designed for numerical attributes with filters that denote numeric ranges.
    \item LFANNS indices~\cite{nhq-nips23,filtered-diskann,caps,ung,rwalks,tfanns,eli,unifilter,curator} achieve high efficiency for filtering data objects whose label sets contain query-specified labels.
\end{enumerate*}
Specialized indices for range and label filters show superiority over the first two categories~\cite{irange,eli}.
This paper focuses on LFANNS and discusses recent work as follows.
 
\bparagraph{LFANNS} indices can be divided into two categories.
\emph{Label-partition indices} group vectors with identical labels or label sets, and groups are arranged by a tree or an inverted file (IVF) structure.
Representatives include CAPS~\cite{caps}, ELI~\cite{eli}, UniFilter~\cite{unifilter}, and Curator~\cite{curator}.
This category is effective when filters are simple and label sets are small (e.g., base/query label sets contain fewer than five labels on average), since the index can focus on frequent labels~\cite{caps,curator} or manageable label combinations~\cite{unifilter,eli}.
However, as label sets grow, the number of possible combinations explodes, which either limits support for complex queries~\cite{caps,unifilter,curator} or incurs substantial space overhead~\cite{eli}.
Consequently, when a query does not match any pre-built (shared) index, it must fall back to a less compatible index (typically much larger than the true filtered subset), which can significantly reduce efficiency.

\emph{Label-enhancement indices} leverage label-enhanced edges to connect two vertices in a graph index via fused distance or label sharing relationships.
Representatives are NHQ~\cite{nhq-nips23}, Vamana (filtered and stitched variants)~\cite{filtered-diskann}, RWalks~\cite{rwalks}, Packing~\cite{tfanns}, and UNG~\cite{ung}.
This category can handle label sets with many labels, but existing algorithms perform well only at one end of the selectivity spectrum.
NHQ and RWalks are efficient when the filter is non-selective (i.e., when many vectors can pass the filter).
They combine the vector distance and the label set distance using a user-defined weight.
When the filter is selective, the graph traversal becomes less connected.
Vamana, Packing, and UNG are powerful on query workloads with high selectivity.
The graph is segmented into groups with identical labels or label sets, and enhanced by cross-group edges that connect only to a limited number of label sets satisfying overlap or containment relationships. 
Consequently, the traversal becomes less connected for non-selective queries~\cite{ung,unifilter}.

As summarized in Table~\ref{tab:lfann-comp}, LSSG is the only index satisfying all listed criteria. 
Its performance consistency across all three semantics is empirically validated in Sec.~\ref{sec:eval}.

\section{Label Stratified Similarity Graph}
This section follows the LSSG workflow.
We first define the index topology, then construct it through similar label set selection and multi-tier insertion, and introduce tier-ordered in-filtering search.

\subsection{Graph Topology Formalization} \label{sec:lssg}

\begin{definition}[Label-Stratified Similarity Graph, LSSG] \label{def:lssg}
    Let $G=(V,E)$ be a label-stratified similarity graph, where $V$ and $E$ are sets of vertices and directed edges, respectively.
    Each vertex in $V$ denotes a vector, and $E$ is partitioned into $T$ subsets.
    The $t^\text{th}$ ($1\leq t\leq T$) subset and $V$ compose a subgraph of $G$ in tier $t$: $G_t=(V,E_t)$.
    An edge $v_i\rightarrow v_j \in E_t$ is subject to the following two constraints:
    \begin{enumerate}[left= 0pt, label=C\arabic*.]
        \item $\forall v_k \in V\backslash\{v_i,v_j\}, \delta_v(v_i,v_j) < \delta_v(v_i,v_k) \,\vee\, \delta_v(v_i,v_j) < \delta_v(v_j, v_k)$;
        \item $\delta_L(l(v_i),l(v_j)) \leq \theta_t$ and $\forall t\in(1,T), \forall v_k \in N_u^t \backslash \{v_i\} \wedge \delta_v(v_i,v_k)$ $> \delta_v(v_i,v_j), \delta_L(l(v_j), l(v_k)) \ge \theta_t$, where $\theta_t = 1-\frac{t-1}{T-1}$ and $N_{v_i}^t$ is the neighbor of $v_i$ in $G_t$.
    \end{enumerate}
\end{definition}

Unlike traditional ANNS graphs relying solely on vector distance $\delta_v$, LSSG couples vector proximity with label set distance $\delta_L$.
We use Euclidean distance for $\delta_v$ and Jaccard distance for $\delta_L$, which are widely adopted in LFANNS \cite{ung,eli,nhq-nips23}. 
Other choices can also be incorporated \cite{hnsw,nsg,hscg,psp}.
The triangle inequality of Jaccard distance can benefit effective label-based pruning.
Constraint~C1 follows the standard RNG-style diversification rule used in graph-based ANNS indices like NSG~\cite{nsg} and HNSW~\cite{hnsw}.
Applied within each tier, it removes redundant neighbors in vector space and preserves directional diversity.
Constraint~C2 is the core of LSSG’s tiered design and has two roles:
(1) \emph{Stratification}: the condition $\delta_L(l(v_i),l(v_j))\leq\theta_t$ makes upper tiers (larger $t$, smaller $\theta_t$) increasingly strict, so edges gradually concentrate among vertices with more similar (and ultimately identical) label sets.
(2) \emph{Diversification}: when multiple candidates appear in a similar direction in vector space, C2 prefers neighbors that are label-wise diverse.
Once a closer neighbor $v_j$ is selected, a farther neighbor $v_k$ is retained only if it is sufficiently different from $v_j$ in label space (i.e., $\delta_L(l(v_j),l(v_k))\ge\theta_t$), preventing intermediate tiers from being dominated by a single label cluster and improving reachability across different label combinations.

Fig.~\ref{fig:LSSG_structure} illustrates the resulting hierarchy.
Tier~1 ($\theta_1=1$) serves as the backbone. 
Since all label set distances are $\leq 1$, C2 is effectively relaxed and the graph behaves like a standard RNG-based ANNS structure, ensuring global navigability.
In intermediate tiers ($1<t<T$), edges are selectively admitted or removed according to $\theta_t$, which steers traversal toward regions that better match a query’s label requirements.
For example, in tier~2 ($\theta_2=0.67$), an edge $v_6 \rightarrow v_9$ can be kept because $\delta_L(L_3, L_5)$ falls in the threshold, while a spatially close candidate $v_5$ is excluded if $\delta_L(L_3, L_4) > 0.67$.
Tier~$T$ ($\theta_T=0$) connects only vertices with identical label sets, forming exact-match subgraphs that remain reachable through lower tiers.

\subsection{Similar Label Set Selection}
\label{sec:labelsel}

\begin{algorithm}[t]
\SetNoFillComment
\caption{SelectSimilarLabelSets}
\label{alg:sel_labelsets}
\SetKw{KwBreak}{break}
\SetKwInOut{HyperIn}{Index parameters}
\SetKwProg{myproc}{procedure}{}{}
\KwIn{$L_x$: inserted label set}
\KwOut{$\mathcal{L}_C$: similar label set dataset}
\HyperIn{$\phi,\tau,\beta$ (see Table~\ref{tab:notation})}
\lIf{using IVF scanning}{\Return $\mathrm{IVFUnion}(L_x, \phi)$}
\lElse{\Return $\mathrm{MinHashProbe(\mathcal{L}_x, \tau, \beta)}$}
\BlankLine
\myproc{$\mathrm{IVFUnion}(L_x, \phi)$}{
    $\mathcal{L}_C \leftarrow$ empty bitmap s.t. $|\mathcal{L}_C|=|\mathcal{L}_I|$\;
    $L'_x\leftarrow$ sort $L_x$ by ascending label frequency\;
    \ForEach{label $a\in L'_x$}{
        \lIf(\Comment*[f]{$\veedot$: bitwise OR}){$|\mathcal{L}_C|< \phi$}{$\mathcal{L}_C\leftarrow\mathcal{L}_C \veedot\mathrm{IVF}_a$}
    }
    \Return $\mathcal{L}_C$\;
}

\BlankLine
\myproc{$\mathrm{MinHashProbe}(L_x, \tau, \beta)$}{
    $\mathcal{L}_C\leftarrow\{\},r\leftarrow \frac{\tau}{\beta}$\Comment*{calculate MinHash rows per band}
    \For{band $B_i \leftarrow B_1,\ldots,B_\beta$}{
        \Comment{$B_i$: band is a hash table; hash: FNV-1a hashing~\cite{fnv}}
        $\mathcal{L}_C \leftarrow \mathcal{L}_C \cup B_i(\mathrm{hash}(S_x[r\times(i-1), r\times i]))$\;
    }
    \Return $\mathcal{L}_C$\;
}
\end{algorithm}

To build LSSG, we need to frequently identify candidate label sets that are likely to fall in a tier-specific threshold (i.e., $\delta_L(\cdot,\cdot)\le \theta_t$), after which we compute $\delta_L$ to verify and filter them.
A naive approach is to precompute and store $\delta_L$ for all label set pairs, requiring $\mathrm{O}(1)$ lookup but $\mathrm{O}(|\mathcal{L}|^2)$ space and preprocessing.
This is infeasible at scale.
We use two candidate-selection strategies to avoid quadratic storage and compute exact $\delta_L$ only for a bounded candidate set.

\bparagraph{Progressive IVF scanning}
reduces memory overhead by indexing label-to-labelset membership.
For each label $a$, we maintain a bitmap (or posting list) $\mathrm{IVF}_a$ indicating which label sets contain $a$.
Given an inserted label set $L_x$, $\mathrm{IVFUnion}$ in Alg.~\ref{alg:sel_labelsets} initializes an empty bitmap (Line~4), sorts labels in $L_x$ by ascending label frequency (Line~5), and progressively merges their bitmaps via bitwise OR (Lines~6--7) until a strict candidate budget $\phi$ (in terms of the number of 1-bits) is reached.
Processing low-frequency labels first improves selectivity by activating fewer label sets, thereby reducing the candidate space and subsequent $\delta_L$ computations.

\bparagraph{Progressive MinHash probing.}
While IVF scanning is accurate and space-efficient, its probing cost can grow with the label universe size $|A|$ and the number of label sets $|\mathcal{L}|$ (especially under high-frequency labels), and it treats all retrieved candidates uniformly, even though candidates with higher Jaccard similarity are more likely to satisfy $\delta_L\le \theta_t$.
To address both issues, we use the standard MinHash LSH banding technique~\cite{minhash} to prioritize label sets with high Jaccard similarity.
MinHash probing is insensitive to $|A|$ during lookup and tends to retrieve more similar label sets with higher probability.
We target label sets within a tier-specific distance threshold $\theta_t$ (equivalently, Jaccard similarity $J\ge 1-\theta_t$), with a collision probability characterized by Theorem~\ref{theorem:minhash}.

\begin{theorem} \label{theorem:minhash}
Let $\mathcal{H}=\{h_1,h_2,\dots, h_\tau\}$ be $\tau$ independent hash functions used to form MinHash signatures.
For a label set $L\subseteq A$, its signature is $s_L=\big(h_1^{\min}(L), h_2^{\min}(L), \dots, h_\tau^{\min}(L)\big)$, where $h_i^{\min}(L)=\min_{a\in L}h_i(a)$.
Partition $s_L$ into $\beta$ bands, each containing $r=\frac{\tau}{\beta}$ rows.
For two label sets $L_1,L_2$, if $J(L_1,L_2)\ge 1-\theta_t$, the probability that they match in at least one band satisfies:
\begin{align} \label{eq:p_minhash_colli}
    P_{\mathrm{collide}} \ge 1-\left(1-\left(1-\theta_t\right)^r\right)^\beta.
\end{align}

Equality holds at the similarity boundary $J(L_1,L_2)=1-\theta_t$.
\end{theorem}
The proof is given in Appendix~\ref{app:proof_th31}.
Since $P_{\mathrm{collide}}$ increases with Jaccard similarity (and thus decreases with $\theta_t$), more similar label sets are more likely to be retrieved.
This property aligns with LSSG's tiering: higher tiers impose smaller $\theta_t$ and stricter label constraints, so candidate selection increasingly prioritizes highly similar label sets.
In practice, $(\tau,\beta)$ controls the recall--efficiency trade-off.
A larger signature improves recall but also increases probing cost.

Based on this guarantee, $\mathrm{MinHashProbe}$ (Lines~11--12 in Alg.~\ref{alg:sel_labelsets}) gathers candidate label sets $\mathcal{L}_C$ by probing $\beta$ LSH tables using the band hashes of $L_x$'s signature. 
We then compute $\delta_L$ only on $\mathcal{L}_C$ for the later selection of which candidates satisfy $\delta_L\le \theta_t$.

\begin{figure}
    \centering
    \includegraphics[width=\linewidth]{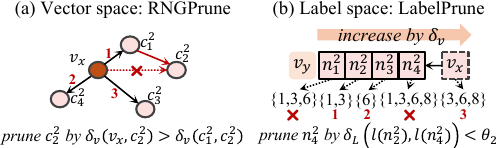}
    \caption{Dual-space pruning diversification for tier~2}
    \Description{Dual-space pruning diversification for tier~2}
    \label{fig:dualpruning}
\end{figure}
\subsection{Multi-tiered Insertion} \label{sec:method_insertion}

\begin{algorithm}[t]
\SetNoFillComment
\caption{Insert}
\label{alg:insert}
\SetKw{KwBreak}{break}
\SetKwInOut{HyperIn}{Index parameters}
\KwIn{$v_x$: vector to insert; $L_x$: label set of $v_x$}
\KwOut{$I$: index after $v_x$ is inserted}
\HyperIn{$m,\omega_c,T$ (see Table~\ref{tab:notation})}
\Comment{Phase 1: label set registration}
\If{$L_x$ $\notin$ inserted label set dataset $\mathcal{L}_I$}{
    insert $L_x$ into label IVF bitmaps of $\mathcal{L}_I$\;
    calculate signature $S_x$ and insert it into MinHash LSH\;
}
\Comment{Phase 2: similar label set selection}
$\mathcal{L}_C \leftarrow \mathrm{SelectSimilarLabelSets}(L_x)$\;
$ \mathcal{J}\leftarrow\{\delta_L(L_x,L_y)\,|\,L_y \in \mathcal{L}_C\}$\;
$C\leftarrow\{\}$\Comment*{pool for cross-tier sharing candidates}
\For{tier $t \leftarrow 1,\ldots,T$}{
    \Comment{Phase 3: filtered candidate selection}
    $\mathcal{L}_t \leftarrow \{ L_y \,|\, \delta_L(L_x, L_y) \in \mathcal{J} \cap [0, \theta_t] \}$\Comment*{$\theta_t = 1 - \frac{t-1}{T-1}$}
    $C\leftarrow \{c\,|\,c \in C \wedge l(c)\in \mathcal{L}_t\}$\Comment*{retain similar label sets}
    \lIf{$|C|<m$}{$C\leftarrow C\cup\mathrm{BeamSearch}(v_x, \mathcal{L}_t, \omega_c)$}
    \Comment{Phase 4: dual-space pruning and edge update}
    $N_{v_x}^t\leftarrow \mathrm{RNGPrune}(C, \frac{m}{2})$\;
    \ForEach{vector $v_y \in N_{v_x}^t$}{
        $N_{v_y}^t \leftarrow N_{v_y}^t \cup \{v_x\}$\Comment*{skip Line~14 if $|N_{v_y}^t| \le m$}
        $N_{v_y}^t \leftarrow \mathrm{RNGPrune}(N_{v_y}^t,m) \cap \mathrm{LabelPrune}(N_{v_y}^t, \theta_t)$\;
    }
}
\Return $I$\;
\end{algorithm}

Alg.~\ref{alg:insert} inserts a vector $v_x$ with label set $L_x$ in four phases.

\bparagraph{Phase~1 (label set registration).}
We first check whether $L_x$ has appeared before.
If $L_x$ is new, we register it in the label-to-labelset inverted structures (e.g., IVF bitmaps).
When MinHash is enabled, we additionally compute the MinHash signature of $L_x$ and insert it into the corresponding LSH band tables.

\bparagraph{Phase~2 (similar label set selection).}
We then invoke Alg.~\ref{alg:sel_labelsets} to retrieve a candidate set of label sets likely to be close to $L_x$.
For each retrieved label set $L$, we compute $\delta_L(L_x,L)$ and cache the results in a map $\mathcal{J}$.
It is later reused to filter candidates by tier thresholds $\theta_t$ without recomputing $\delta_L$.

\bparagraph{Phase~3 (vector-space candidate collection).}
For each tier $t$, we collect candidate vertices by combining vector proximity with the tier-specific label constraint.
Specifically, we restrict candidates to vertices whose label sets satisfy $\delta_L(L_x, l(v))\in [0,\theta_t]$, and perform beam search (Alg.~\ref{alg:beamsearch}) to obtain a vector-close candidate pool.
When the candidate pool from Phase~2 is already sufficiently large, we optionally skip beam search to accelerate indexing, a common strategy in hierarchical graph indices~\cite{wow,dsg}.

\bparagraph{Phase~4 (dual-space pruning and edge update).}
Finally, we construct outgoing edges for $v_x$ by enforcing the two constraints in Def.~\ref{def:lssg}.
We select up to $\frac{m}{2}$ neighbors from the filtered candidates and keep the remaining capacity for future insertions~\cite{hnsw}, following an RNG-style pruning strategy ($\mathrm{RNGPrune}$).
Fig.\ref{fig:dualpruning}(a) iteratively considers candidates in ascending $\delta_v$ order ($c_1^2\rightarrow c_4^2 \rightarrow c_3^2 \rightarrow c_2^2$) and removes geometrically redundant edges ($c_2^2$ in this example).

For each picked neighbor $v_y$, we also attempt to insert the back-edge into $v_y$'s neighbor list.
If $v_y$'s list is full, we apply $\mathrm{RNGPrune}$ to reduce vector-space redundancy. 
If the remaining neighbors still satisfy the RNG constraint, we further apply $\mathrm{LabelPrune}$ to increase label set diversity.
This dual-space pruning reserves capacity for neighbors with more diverse label sets, and reduces local graph density, which speeds up traversal convergence and, in turn, lowers the total indexing cost of candidate preparation in Phase~3.

Fig.~\ref{fig:dualpruning}(b) shows $\mathrm{LabelPrune}$ in tier~2, where $\theta_2 = 0.67$ (rounded from $\theta_2=\frac{2}{3}$, hereafter).
Consider diversifying the neighbor list of a selected neighbor $v_y$ with $l(v_y)=\{1,3,6\}$, and candidate neighbors
$n_1^2:\{1,3,6\}$, $n_2^2:\{1,3\}$, $n_3^2:\{6\}$, $n_4^2:\{1,3,6,8\}$, as well as the incoming vertex $v_x$ with $l(v_x)=\{3,6,8\}$.
$\mathrm{LabelPrune}$ processes candidates in ascending $\delta_v$ order and compares each candidate only against the \emph{already kept} neighbors.
A candidate is pruned if (1) its label set is identical to $l(v_y)$, or (2) its label set distance to any kept neighbor is smaller than the tier threshold $\theta_2$.
Specifically, $n_1^2$ is pruned by (1) due to $l(n_1^2)=l(v_y)$.
$n_2^2$ is kept as the first accepted neighbor.
$n_3^2$ is also kept since $\delta_L(l(n_2^2),l(n_3^2))=1 \not< 0.67$.
$n_4^2$ is pruned by (2) because it is label-redundant w.r.t.\ the kept neighbor $n_2^2:\delta_L(l(n_2^2),l(n_4^2))=1-\frac{2}{4}=0.5 < 0.67$.
Finally, $v_x$ is kept since it is sufficiently different from the kept neighbors:
$\delta_L(l(n_2^2),l(v_x))=0.75>0.67$ and
$\delta_L(l(n_3^2),l(v_x))= 0.67 \not< 0.67$.

\subsection{In-Filtering Traversal and LFANNS Query}

\begin{algorithm}[t]
\SetNoFillComment
\caption{BeamSearch}
\label{alg:beamsearch}
\SetKw{KwBreak}{break}
\SetKwInOut{HyperIn}{Index parameters}
\KwIn{$v_q$: query vector; $\mathcal{L}_q$: selected label set dataset; $\omega$: beam width}
\KwOut{$W=\{v_i\,|\, l(v_i) \in \mathcal{L}_q\}$: neighbor set}
\HyperIn{$m,T$ (see Table~\ref{tab:notation})}
$\mathcal{V}_r\leftarrow$ selected vectors s.t. $\forall v_r\in\mathcal{V}_r, l(v_r)\in \mathcal{L}_q$\;
$C \leftarrow \mathcal{V}_r, W\leftarrow \mathcal{V}_r$\Comment*{min-heap $C$ and max-heap $W$ sorted by $\delta_v$}
\While{$C\neq \emptyset$}{
    $c\leftarrow$ pop the nearest vertex to $v_q$ from $C$\;
    \lIf{$\delta_v(c,v_q)> \max_{w\in W} \delta_v(w,v_q)$}{\KwBreak}
    \For{tier $t \leftarrow 1, \ldots, T$}{
        $U\leftarrow U\cup\{u\,|\,u\in N_c^t \wedge l(u)\in \mathcal{L}_q\}$\;
        \If{number of unvisited vertices in $U > m$}{
            keep the first $m$ unvisited vertices and \KwBreak\;
        }
    }
    \ForEach{unvisited neighbor $u\in U$}{
        \If{$|W| < \omega$ or $\delta_v(u,v_q) < \max_{w\in W} \delta_v(w,v_q)$}{
            mark $u$ as visited and push it to $C,W$\;
        }
        \lWhile{$|W|>\omega$}{pop $W$}
    }
}
\Return $W$\;
\end{algorithm}

Alg.~\ref{alg:beamsearch} performs beam search on the \emph{filtered} subgraph induced by the label set scope $\mathcal{L}_q$ (i.e., only vertices with $l(v)\in\mathcal{L}_q$ are eligible).
It starts by selecting a small set of entry vertices $\mathcal{V}_r$ from the eligible vertices (Line~1), and initializes both the candidate min-heap $C$ and the result max-heap $W$ with $\mathcal{V}_r$ (Line~2).
Selecting entry vertices across label sets helps preserve \emph{completeness}. 
Even when cross-label connectivity is sparse in lower tiers, the search can still start from multiple valid label set regions rather than being trapped in a single region.
This issue is inevitable if the index is not built on the fully filtered graph for every possible query, yet such a practice is impractical for LFANNS as to be discussed in Sec.~\ref{sec:analysis}.

The traversal iteratively expands candidates (Lines~3--13).
At each iteration, it pops the current closest vertex $c$ from $C$ (Line~4) and terminates early when $c$ is already farther than the worst element in $W$ (Line~5), a standard beam-search stopping criterion.
For next-hop expansion, it scans neighbors of $c$ tier by tier and only collects those whose labels satisfy $l(u)\in\mathcal{L}_q$ (Line~7).
This in-filtering expansion enforces \emph{fidelity} by preventing the search from spending steps on unqualified vertices.
To align the expansion cost with the graph topology, Lines~8--9 cap the number of newly considered (unvisited) candidates to at most $m$, avoiding excessive $\delta_v$ evaluations that are unlikely to improve the current beam.
Lines~10--13 maintain the two heaps so that the search continues to prioritize vertices closer to $v_q$ while bounding $W$ by the beam width $\omega$.

Alg.~\ref{alg:searchknn} shows the end-to-end LFANNS query procedure on LSSG.
It first transforms the query filter $f$ and query label set $L_q$ into a label set scope $\mathcal{L}_q$ (Lines~1--3), which specifies the admissible label sets for the query.
For equality queries ($f_e$), $\mathcal{L}_q$ is simply $\{L_q\}$.
For containment ($f_c$) and overlap ($f_o$), $\mathcal{L}_q$ is obtained via bitmap intersection or union over the inverted structures induced by the labels in $L_q$.
This design supports arbitrary query label sets without assuming a fixed label cardinality, thereby satisfying \emph{versatility}.
Given $\mathcal{L}_q$, the algorithm runs in-filtering beam search (Line~4) and finally retains the closest $k$ vectors in $\mathcal{R}$ (Line~5).

\section{Stepwise Analysis of Filter-Valid Expansion}
\label{sec:analysis}

\subsection{Stepwise Oracle Proxy}
The key effect of LSSG is to leverage tiered label stratification to increase the probability that an expansion from a filter-valid vertex reaches another filter-valid vertex.
This reduces ineffective in-filtering expansions compared with a label-agnostic graph, while avoiding the infeasible cost of building a separate graph for every possible query filter.
We use the \emph{oracle label-filtered RNG} as a conceptual reference for filter-valid navigation.

\begin{definition}[Oracle Label-Filtered RNG] \label{def:oracle}
    Given a query label set $L_q$ and a label filter $f$ on a hybrid dataset $\mathcal{D}=\{\mathcal{V}, \mathcal{L}\}$, an oracle label-filtered RNG, denoted by $O_{f,L_q}$, is built on $\mathcal{V}_{f,L_q}\subseteq \mathcal{V}$, where $\forall v_i\in\mathcal{V}_{f,L_q}, v_j\in\mathcal{V}-\mathcal{V}_{f,L_q}, f(l(v_i)\,|\,L_q)=1, f(l(v_j)\,|\,L_q)=0$.
    Edges in $O_{f,L_q}$ satisfy C1 in Def.~\ref{def:lssg}.
\end{definition}

\begin{algorithm}[t]
\SetNoFillComment
\caption{LFANNS}
\label{alg:searchknn}
\SetKw{KwBreak}{break}
\KwIn{$v_q$: query vector; $L_q$: query label set; $f$: query filter; $\omega_s$: search beam width; $k$: number of neighbors}
\KwOut{$\mathcal{R}$: $k$ filtered approximate nearest neighbors}
\lIf{$f=f_e$}{
    $\mathcal{L}_q\leftarrow \{L_q\}$
}
\lElseIf{$f=f_c$}{
    $\mathcal{L}_q\leftarrow \mathrm{IVFIntersetion}(L_q,|L_q|)$
}
\lElseIf{$f=f_o$}{
    $\mathcal{L}_q\leftarrow \mathrm{IVFUnion}(L_q,|L_q|)$
}
$\mathcal{R} \leftarrow \mathrm{BeamSearch}(v_q,\mathcal{L}_q, \omega_s)$\;
\lWhile{$|\mathcal{R}| > k$}{pop $R$}
\Return $\mathcal{R}$\;
\end{algorithm}

Constructing $O_{f,L_q}$ for all possible $(f,L_q)$ is infeasible because the query label set is not known in advance.
LSSG therefore keeps a single multi-tier graph and approximates the oracle behavior only at the level of one expansion step.
Let $c$ be the current hop selected by best-first search (Alg.~\ref{alg:beamsearch}, Line~4), and let $u'\in N_c^t$ be a neighbor of $c$ at tier $t$.
Denote $L_1=l(c)$ and $L_2=l(u')$.
By LSSG construction, an edge in tier $t$ satisfies $\delta_L(L_1,L_2)=1-J(L_1,L_2)\le \theta_t = 1-\epsilon_t$, equivalently $J(L_1,L_2)\ge \epsilon_t$, where $J(L_1,L_2)=\frac{|L_1\cap L_2|}{|L_1\cup L_2|}$ and $\epsilon_t\in[0,1]$.
The stepwise filter-valid expansion probability is
\begin{align}
\resizebox{.91\columnwidth}{!}{$
    \Pr(u'\in U)
    =\Pr\big(f(L_2\,|\,L_q)=1 \,\big|\, f(L_1\,|\,L_q)=1 \wedge J(L_1,L_2)\ge \epsilon_t\big),
$}
\end{align}
where $L_q$ is the query label set and $f$ is the filter in Alg.~\ref{alg:searchknn}.
For the oracle graph $O_{f,L_q}$ (Def.~\ref{def:oracle}), all vertices satisfy $f(\cdot\,|\,L_q)=1$ by construction, hence $\Pr(u'\in U)=1$ trivially.
For LSSG, it measures the local probability of an expansion inside the filtered vertex set, and is also a stepwise proxy for filtered navigation quality.
For a fixed tier $t$, we write $\epsilon$ instead of $\epsilon_t$ for simplicity.
Considering the three filter types in Def.~\ref{def:labelfilter}, $\Pr(u'\in U)$ has corresponding forms:
\begin{definition}[Target Probabilities for Three Filters]
For label sets $L_1,L_2$ and query label set $L_q$, define:
\begin{itemize}
    \item Equality: $P_e=\Pr\big(L_q=L_2\,|\,L_q=L_1 \wedge J(L_1,L_2)\ge \epsilon\big)$;
    \item Containment: $P_c=\Pr\big(L_q\subseteq L_2\,|\,L_q\subseteq L_1 \wedge J(L_1,L_2)\ge\epsilon\big)$;
    \item Overlap: $P_o=\Pr\big(L_q\cap L_2 \neq \emptyset \,|\, L_q \cap L_1 \neq \emptyset \wedge J(L_1,L_2)\ge\epsilon\big)$.
\end{itemize}
\end{definition}

For tier~$T$ of LSSG, $\theta_T=0$ and thus $J(L_1,L_2)\ge \epsilon_T=1$, which implies $L_1=L_2$.
Hence $P_e=1$ in tier~$T$, which means once an equality-valid vertex is reached, a top-tier expansion preserves equality validity at the label level.
The containment and overlap cases require additional model assumptions, stated explicitly below.

\subsection{Containment Target Probability}

We analyze containment under an independent Bernoulli model.
All proofs are in Appendices~\ref{app:proof_containment} and~\ref{app:proof_overlap}.

\begin{definition}[Independent Bernoulli inclusion model]\label{def:label_iid}
For each label $a_i\in A$, let $p_i\in[0,1]$ be its inclusion probability, sorted so that $p_i\ge p_{i+1}$.
Labels are included independently in each generated label set, i.e., $\Pr(a_i\in L)=p_i$.
We assume $L_1$, $L_2$, and $L_q$ are generated independently from this model unless stated otherwise.
\end{definition}

We first give a distribution-agnostic lower bound for $P_c$.

\begin{theorem}[Lower Bound of $P_c$] \label{theorem:pc_gen}
    For $\forall L_1, L_2 \in \mathcal{L}$ that satisfy $J(L_1, L_2) \ge \epsilon$, the worst-case lower bound of $P_c$ is
    \begin{align} \label{eq:worst_pc}
        P_c\ge (1-p_{\max})^{(1-\epsilon)|A|},\quad p_{\max} = \max_{a_i\in A}p_i.
    \end{align}
    
    Further, $\mu=\mathbb{E}[|L_1|] = \sum_{a_i\in A}p_i, \sigma^2=\mathbb{V}[|L_1|]=\sum_{p_i\in A}p_i(1-p_i)$,
    \begin{align} \label{eq:prob_pc}
        P_c \ge (1-p_{\max})^{(1-\epsilon)(\mu+\frac{\sigma}{\sqrt{\xi}})},\quad \text{with probability }1-\xi.
    \end{align}
\end{theorem}

Theorem~\ref{theorem:pc_gen} is most informative for uniform or near-uniform label distributions.
For example, suppose $\mu=5$ and $|A|=100$, i.e., each label set contains five labels on average and the universe has 100 labels.
Under a uniform distribution, $p_i=\frac{\mu}{|A|}=0.05$.
For $\epsilon=0.8$, the worst-case bound in Eq.~\eqref{eq:worst_pc} gives $P_c \ge 0.95^{0.2\times 100}=0.358$, which is always valid but conservative.
With confidence $1-\xi=95\%$, Eq.~\eqref{eq:prob_pc} yields $P_c \ge 0.95^{0.2\times\left(5+\sqrt{\frac{4.75}{0.05}}\right)}=0.859 > 0.358$.

In real-world systems, label distributions are often power-law (Zipf-like)~\cite{zipfbook,zipfreview,ung,eli}.
A small number of head labels occur frequently, while most tail labels are rare.
In such skewed settings, the distribution-agnostic bound in Theorem~\ref{theorem:pc_gen} becomes overly conservative because $p_{\max}$ is dominated by extreme head labels.
However, this looseness does not reflect actual intersection dynamics.
Under a high Jaccard similarity constraint, a frequent head label is unlikely to appear in the difference set $D=L_1\backslash L_2$ because missing a high-frequency label in $L_2$ would strongly penalize Jaccard similarity.
Thus, the difference set tends to be dominated by labels with relatively small-to-moderate inclusion probabilities, leading to the sharper bound below.

\begin{theorem}[Lower Bound of $P_c$ under Skewed Distribution] \label{theorem:pc_skew}
Let label inclusion probabilities follow the Zipf distribution, i.e., $p_i=\frac{\mu i^{-s}}{H_{|A|,s}}$ for $i=1,2,\dots,|A|$ and $s\ge 0$, where $H_{|A|,s}=\sum_{j=1}^{|A|} j^{-s}$ is the generalized harmonic number and $\mu=\mathbb{E}[|L|]\in(0,H_{|A|,s}]$ is the expected label set size. Define $r=\max(2,\argmin_i \, \frac{H_{i,s}}{H_{|A|,s}} \ge \rho)$ for some $\rho\in(0,1)$, and let $B=\frac{2\mu(1-\epsilon)}{1+\epsilon}$.
For any $q\ge 0$, define
$
\Psi_q(x):=
\begin{cases}
\dfrac{(r+x-1)^{1-q}-(r-1)^{1-q}}{1-q}, & q\neq 1,\\\quad
\ln\dfrac{r+x-1}{r-1}, & q=1,
\end{cases}
\quad x\ge 0.
$
Then, the containment target probability satisfies
\begin{align}
P_c \ge
\exp\left(
-\frac{\mu}{H_{|A|,s}}\Psi_s(B)
-\frac{\mu^2}{H_{|A|,s}^2}\Psi_{2s}(B)
\right)
\bigl(1-\mu\rho(1-\epsilon)\bigr).
\end{align}
\end{theorem}

\begin{figure}
    \centering
    \includegraphics[width=\linewidth]{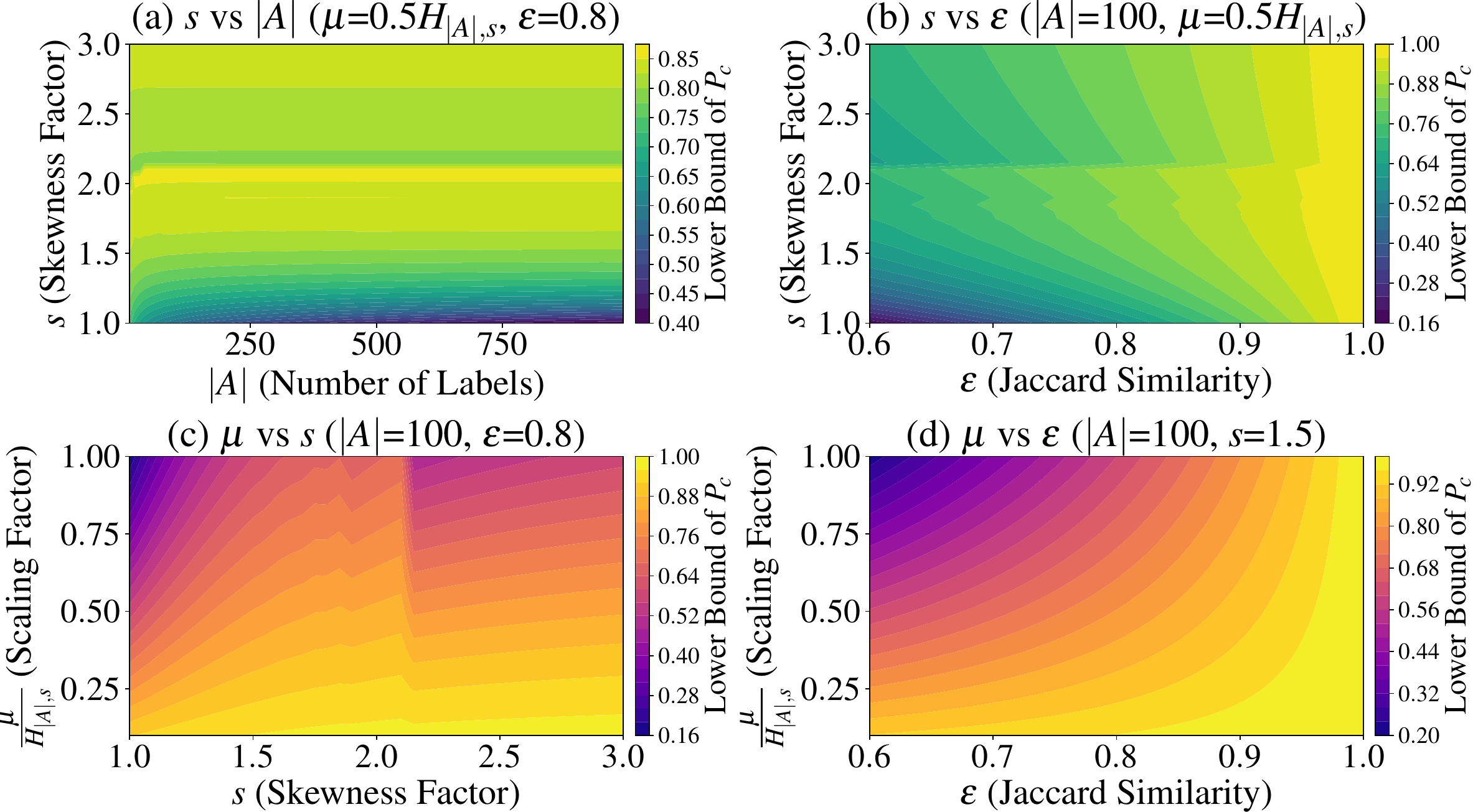}
    \caption{Simulation of Theorem~\ref{theorem:pc_skew}}
    \Description{Simulation of Theorem~\ref{theorem:pc_skew}}
    \label{fig:pc_lower}
\end{figure}

Theorem~\ref{theorem:pc_skew} complements the conservative bound in Theorem \ref{theorem:pc_gen} by giving a sharper characterization under skewed label distributions.
Fig.~\ref{fig:pc_lower} shows three consistent trends: the bound improves with stronger skew and larger Jaccard threshold $\epsilon$, but decreases as the expected label set size $\mu$ grows.
For example, at $|A|=100$, the bound increases from $P_c\ge 0.56$ when $s=1$ to $P_c\ge 0.84$ when $s=2$, and from $P_c\ge 0.60$ at $\epsilon=0.6$ to $P_c\ge 0.79$ at $\epsilon=0.8$.
These trends support the same navigation intuition: stricter tiers shrink the missing-label set $D=L_1\backslash L_2$, thereby increasing the probability of a containment-valid next hop expansion.

\subsection{Overlap Target Probability}

For overlap queries, the probability of \emph{non}-overlap is maximized when the query label set is as small as possible.
\begin{theorem}[Lower Bound of $P_o$] \label{theorem:po}
Assuming that query label sets of the same size are equally likely, for any $L_1,L_2\in\mathcal{L}$ satisfying $J(L_1,L_2)\ge\epsilon$, the overlap target probability satisfies $P_o\ge\epsilon$.
\end{theorem}

Theorem~\ref{theorem:po} shows that $P_o$ increases with label set similarity.

\bparagraph{Summary.}
The containment and overlap analyses lead to the same design principle.
As the tier threshold becomes stricter, an expansion from a filter-valid vertex is more likely to remain filter-valid.
For containment, stricter tiers reduce the possible missing-label set $L_1\backslash L_2$.
For overlap, the lower bound scales with the shared-label ratio induced by $J(L_1,L_2)\ge\epsilon_t$.
Thus, lower tiers preserve global vector-space reachability, while higher tiers provide label-consistent shortcuts.
This is the core mechanism by which LSSG approaches oracle filtered navigation without constructing a separate graph for each query filter.

\section{Asymptotic Complexity}
\label{sec:complexity}
\subsection{Index Size}
LSSG is composed of graph topology and auxiliary label selection structures.
Each vertex stores at most $m$ neighbors per tier, giving $\mathrm{O}(m|\mathcal{V}|T)$ graph space.
This bound is further tightened by the diversification strategy (Fig.~\ref{fig:dualpruning}), which prunes redundant edges based on label and vector proximity.
IVF bitmaps require $\mathrm{O}(|A||\mathcal{L}|)$ bits, while MinHash signatures and LSH buckets require $\mathrm{O}((\tau+\beta)|\mathcal{L}|)$ space.
Relative to a standard proximity graph, the topology incurs the tier factor $T$ and the additional label selection structures.

\subsection{Indexing Time}

For an inserted vector $v_x$ with label set $L_x$, let $K_x$ be the number of label set candidates returned by IVF or MinHash, $\bar{\ell}_x$ be their average cardinality, and $B_x$ be the number of processed IVF bitmap words.
The cost of candidate selection and Jaccard verification is
\begin{align}
\resizebox{.91\columnwidth}{!}{$
C_{\mathrm{label}}(x)
= \mathrm{O}\left(K_x(|L_x|+\bar{\ell}_x)\right) + 
\begin{cases}
\mathrm{O}(B_x), & \text{IVF},\\
\mathrm{O}(\tau|L_x|+\beta+K_x), & \text{MinHash}.
\end{cases}
$}
\end{align}

Let $H_{x,t}$ denote the actual number of vertices expanded at tier $t$.
Each expansion scans at most $mT$ neighbors, computes at most $m$ vector distances, and performs at most $m$ queue updates.
Including outgoing and reverse-edge pruning, the deterministic insertion cost $C_{\mathrm{ins}}(x)$ is
\begin{align}
C_{\mathrm{ins}}(x)
=&\
C_{\mathrm{label}}(x)
+
\mathrm{O}\left(
\sum_{t=1}^{T}
H_{x,t}(mT+md+m\log\omega_c)
\right)
\notag\\
&+
\mathrm{O}\left(
T\bigl(m\omega_c d+m^3(d+\bar{\ell}_x)\bigr)
\right).
\end{align}

The indexing time mainly depends on the expansion counts $H_{x,t}$.
For a tier-induced search graph satisfying the monotonicity conditions of~\cite{nsg}, the expected search path length is logarithmic in the graph size.
With fixed beam width, this gives
\begin{align}
\resizebox{.91\columnwidth}{!}{$
\mathbb{E}[H_{x,t}]=\mathrm{O}\left(\log|\mathcal{V}_{x,t}|\right),
\quad
\mathcal{V}_{x,t}=\{v\in\mathcal{V}:\delta_L(l(v),L_x)\le\theta_t\}.
$}
\end{align}

Therefore, for fixed parameters and bounded label set selection cost, LSSG retains the expected $\mathrm{O}(|\mathcal{V}|\log|\mathcal{V}|)$ construction time of standard proximity graph indexing, with additional costs from label selection, pruning, and tiered traversal.

\subsection{Query Time}

For a query $(v_q,L_q,f)$, Alg.~\ref{alg:searchknn} first constructs the filtered label set scope $\mathcal{L}_q$.
The scope construction cost is
\begin{align}
C_{\mathrm{scope}}(f,L_q)
=
\begin{cases}
\mathrm{O}(1), & f=f_e,\\
\mathrm{O}(\sum_{a\in L_q}|\mathrm{IVF}_a|), & f\in\{f_c,f_o\}.
\end{cases}
\end{align}

Let $H_{f,q}$ be the actual number of expanded vertices and define $\eta_q=mT+md+m\log\omega_s$.
The deterministic query cost is
\begin{align}
C_{\mathrm{query}}(f,q)
=
C_{\mathrm{scope}}(f,L_q)
+
\mathrm{O}(H_{f,q}\eta_q).
\end{align}

This bound holds for any realized graph and query trace.
Here, $mT$ is the tier-scanning overhead, while stricter tiers increase the chance that scanned neighbors remain label-qualified.

Let $\mathcal{V}_{f,L_q}=\{v\in\mathcal{V}:l(v)\in\mathcal{L}_q\}$ be the filtered vertex scope.
By the monotonic-search result used above, its oracle filtered graph requires $\mathrm{O}(\log|\mathcal{V}_{f,L_q}|)$ expected progress steps.
Let $\alpha_f\in\{P_e,P_c,P_o\}$ be the filter-valid expansion lower bound from Sec.~\ref{sec:analysis}.
Under the stepwise proxy in Sec.~\ref{sec:analysis}, one filter-valid expansion requires $\mathrm{O}(\alpha_f^{-1})$ attempted expansions in expectation.
Combining this factor with the oracle progress length gives the conditional estimate $\mathbb{E}[H_{f,q}]=\mathrm{O}\left(\alpha_f^{-1}\log|\mathcal{V}_{f,L_q}|\right)$.
Therefore,
\begin{align}
\mathbb{E}[C_{\mathrm{query}}(f,q)]=C_{\mathrm{scope}}(f,L_q)+
\mathrm{O}\left(\alpha_f^{-1}\eta_q\log|\mathcal{V}_{f,L_q}|\right),
\end{align}
which serves as the conditional search-effort estimate.
For fixed $T,m,\omega_s$ and $\alpha_f=\Omega(1)$, the expected traversal cost is $\mathrm{O}(\log|\mathcal{V}_{f,L_q}|)$.
Compared with traversal on a standard proximity graph, LSSG pays the tier scanning overhead but reduces ineffective expansions.

\section{Evaluation}
\label{sec:eval}
We evaluate LSSG by addressing the following research questions:
\begin{enumerate}[left=0pt, label=\textbf{RQ\arabic*.}]
\item How does LSSG perform in construction?
Can it achieve smaller index size and faster indexing speed than state-of-the-art LFANNS indices?
Is it robust to variations in vector volume and label cardinality?
\item How does LSSG perform in LFANNS queries?
Does it consistently deliver high recall and low latency across query semantics, evolving label distributions, and diverse label characteristics, thereby addressing the two challenges in Sec.~\ref{sec:intro}?
\item How effective are LSSG's key design components?
What are the comparative strengths of its label set selection algorithms?
How do tier usage and filter-valid expansion explain the query benefits by stratification?
How does LSSG scale with increasing vector and label scales?
\end{enumerate}

\subsection{Experiment Setup}

\bparagraph{Environment.}
We use a Ubuntu 20.04 LTS server with an Intel Xeon Gold 6326 CPU @ 2.90 GHz and 240 GB memory.
The CPU has 16 physical cores, 512 KB L1 cache, 65,536 KB L2 cache, and 16 MB L3 cache.
All indices are implemented in C++ and compiled by g++ 9.4.0 with -Ofast and -avx512f compile options enabled.

\bparagraph{Datasets.}
Table~\ref{tab:data} summarizes the statistics of eight real-world benchmark datasets with Zipf-like label distributions.
Datasets are evaluated under three semantics, each using $|Q|$ queries (in total $|Q|\times 3$).
For each query semantics, label sets are generated independently.
Equality queries sample existing base label sets.
Containment and overlap queries sample from the dataset label universe.
All queries are retained only if they admit at least $k$ valid answers.
Query-label statistics are reported in Appendix~\ref{app:data_detail}.
\begin{itemize}[left=0pt]
    \item \textbf{SIFT} and \textbf{GIST}~\cite{siftgisturl} are two classic vector datasets.
    Following \cite{acorn, ung, nhq-nips23, unifilter}, we generate synthetic integer labels.
    
    \item \textbf{TripClick}~\cite{tripclickurl} encodes click logs from a health web search engine \cite{tripclick} using the DPR model~\cite{dpr}.
    Following \cite{acorn,ung}, we use 28 clinical labels, with remaining labels categorized as ``others''.

    \item \textbf{LAION}~\cite{laion400murl} contains four million image-caption pairs, with images embedded by the CLIP model \cite{clip}.
    Following \cite{acorn}, we pick a one-million subset with 30 caption keywords as labels.
    
    \item \textbf{YTB-Video} and \textbf{YTB-Audio}~\cite{youtube8murl} share identical topic labels.
    Following \cite{ungfilterbench}, feature vectors of one million video frames and five million audio samples are used for benchmarking.

    \item \textbf{Wikipedia}~\cite{wikieurl} has five million English paragraph embeddings, encoded via all-MiniLM-L6-v2~\cite{minilm}.
    Labels are from Wikidata~P31 (\texttt{instance of}) and P279 (\texttt{subclass of}) properties.
    
    \item \textbf{YFCC-1M}~\cite{yfccurl} is from an image embedding dataset with 10 million vectors \cite{nipscomp2023}.
    Following~\cite{ungfilterbench}, main experiments use the first one million subset and YFCC-10M serves for scalability test.
\end{itemize}

We compute Spearman's $\rho$ between label Jaccard distance $\delta_L$ and vector L2 distance $\delta_v$ by label set size quartile.
Seven datasets except Wikipedia have negligible overall correlation ($|\rho|\leq0.021$) and quartile-wise values $<0.05$.
Wikipedia shows a modest positive correlation ($\rho=0.15$ overall, $0.038$--$0.176$ by quartile).
So, the evaluated workloads generally exhibit weak label--vector correlation, which LSSG does not assume.
Complete results are in Appendix~\ref{app:data_detail}.
\begin{table}[t!]
    \centering
    \caption{Statistical data of the benchmark datasets}
    \label{tab:data}
    \small
    \setlength{\tabcolsep}{3pt}
    \begin{tabular}{lcrrrcr}
        \toprule
        \textbf{Datasets} & \textbf{Size} & \textbf{Dim} & $|A|$ & $|\mathcal{L}|$ & \makecell{Avg/Med/P90\\label set size} & $|Q|$ \\
        \midrule
        SIFT      & 1,000,000 & 128   & 12      & 2,385   & 2.3 / 2 / 4  & 10,000 \\
        GIST      & 1,000,000 & 960   & 12      & 2,385   & 2.3 / 2 / 4  & 1,000  \\
        TripClick & 1,055,976 & 768   & 29      & 7,734   & 1.1 / 1 / 2  & 1,000  \\
        LAION     & 1,000,448 & 512   & 30      & 62,066  & 2.9 / 3 / 5  & 1,000  \\
        YTB-Video & 1,000,000 & 1,024 & 3,862   & 160,018 & 3.0 / 3 / 5  & 200    \\
        YTB-Audio & 5,000,000 & 128   & 3,862   & 498,335 & 3.0 / 3 / 5  & 1,000  \\
        Wikipedia & 5,000,000 & 384 & 237,417 & 184,298 & 5.9 / 6 / 8 & 1,000\\
        YFCC-1M   & 1,000,000 & 192   & 181,931 & 636,356 & \ \ 8.8 / 6 / 16 & 1,000  \\
        \bottomrule
    \end{tabular}
\end{table}

\bparagraph{Competitors.}
We compare LSSG to ten competing LFANNS baselines using their official open-source implementations.
For each baseline, we use the recommended building parameters, and sweep query parameters to obtain Pareto-optimal QPS--recall curves.
\begin{itemize}[left=0pt]
    \item \textbf{LSSG-IVF} and \textbf{LSSG-MinHash} differ in similar label set selection.
    Unless stated otherwise, they both use $T=9$, $m=16, \omega_c=128$, with $\phi=50{,}000$ for IVF and $(\tau,\beta)=(64,16)$ for MinHash.
    
    \item \textbf{ACORN}~\cite{acorn} is a predicate-agnostic index.
    It uses bitmaps to indicate which vector IDs can pass the filter.
    As suggested, we set its index parameters to $M=16, M_\beta=32, L=256, \gamma=30$.
    
    \item \textbf{Packing} is the best strategy in \cite{tfanns}.
    Following their settings, we set $M=96, \mathrm{efCons}=256$ for the optimal query performance.
    
    \item \textbf{RWalks}~\cite{rwalks} uses label vector and hybrid distance to enhance HNSW indices.
    We use the recommended parameters in their paper, $M=32, \mathrm{efCons}=200,\tau=0$, and $h=0.1$.
    
    \item \textbf{UNG}~\cite{ung} builds separate Vamana indices for all label sets and adds cross-group edges to ensure connectivity.
    We follow them and set $M=32,L=100$ with $\delta=6$ cross-group edges.
    
    \item \textbf{ELI}~\cite{eli} builds multiple small indices with identical parameters.
    We choose $M=16, \mathrm{efCons}=200$ with elastic factor set to 0.2 as reported the most competitive in the paper.

    \item \textbf{FAISS}~\cite{faiss} and \textbf{VSAG}~\cite{vsag} serve as ANNS libraries and \textbf{Milvus} \cite{milvus-vdb} and \textbf{PGVector}~\cite{pgvector} serve as VDBMS with native filtering.
    HNSW ($M=32, \mathrm{efCons}=200$) and k-means IVF ($\mathrm{n\_lists}=10{,}000$) are evaluated.
    
    \item \textbf{Pre-filtering} only computes distance between the query vector and filtered vectors to obtain exact nearest neighbors.
\end{itemize}

Note that UniFilter~\cite{unifilter} is excluded due to code unavailability, and its reported performance is comparable to UNG.
Curator~\cite{curator} requires filters to be known a priori with at most two labels, conflicting with our query-time filter setup for arbitrary labels.
Early LFANNS studies, including Filtered-Vamana \cite{filtered-diskann}, Stitched-Vamana \cite{filtered-diskann}, NHQ \cite{nhq-nips23}, and CAPS \cite{caps}, are omitted for inferior performance to existing baselines like UNG \cite{ung} and ELI \cite{eli}.

\begin{table}[t!]
    \centering
    \caption{Indexing performance}
    \label{tab:index_time_size}
    {\small
    \setlength{\tabcolsep}{3pt}
    \begin{tabular}{lrcccrrrr}
        \toprule
        \multicolumn{9}{c}{\textbf{Index size (MB)}} \\
        \toprule
        \textbf{Indices} & SIFT & GIST & Trip & LAION & Video & Audio & Wiki & YFCC \\
        \midrule
        ACORN & 346 & 346 & 366 & 347 & 346 & 3,416 & 3,397 & 346 \\
        Packing & 310 & 306 & 311 & 439 & 497 & \textbf{570} & 2,447 & 2,826 \\
        RWalks & 413 & 413 & 578 & 552 & 29,785 & - & - & -\\
        UNG & \textbf{228} & 410 & 318 & \textbf{271} & \underline{152} & \underline{954} & 1,044 & 433 \\
        ELI & 419 & 440 & 351 & 745 & 21,292 & 92,693 & - & - \\
        \midrule
        LSSG-IVF & 311 & \underline{204} & \textbf{199} & \textbf{271} & \textbf{136} & 1,069 & \textbf{602} & \textbf{239} \\
        $\star$-MinHash & \underline{296} & \textbf{192} & \underline{200} & \underline{277} & \underline{142} & 1,065 & \underline{606} & \underline{303} \\
        \toprule
        \multicolumn{9}{c}{\textbf{Indexing time (s)}} \\
        \toprule
        \textbf{Indices} & SIFT & GIST & Trip & LAION & Video & Audio & Wiki & YFCC \\
        \midrule
        ACORN & 223 & 881 & 754 & 500 & 411 & 1,313 & 3,630 & 277 \\
        Packing & 176 & 358 & 471 & 574 & 419 & 2,828 & 2,387 & \underline{145} \\
        RWalks & 155 & 631 & 595 & 306 & 760 & - & - & - \\
        UNG & \textbf{76} & 429 & 311 & 189 & \textbf{211} & 1,422 & 973 & 309 \\
        ELI & \underline{98} & 364 & 287 & 257 & 253 & \underline{807} & - & - \\
        \midrule
        LSSG-IVF & 133 & \underline{241} & \underline{270} & \underline{181} & 287 & 946 & \underline{487} & 379 \\
        $\star$-MinHash & 122 & \textbf{224} & \textbf{233} & \textbf{179} & \underline{243} & \textbf{721} & \textbf{426} & \textbf{111} \\
        \bottomrule
    \end{tabular}}
\end{table}

\begin{figure*}
    \centering
    \includegraphics[width=\linewidth]{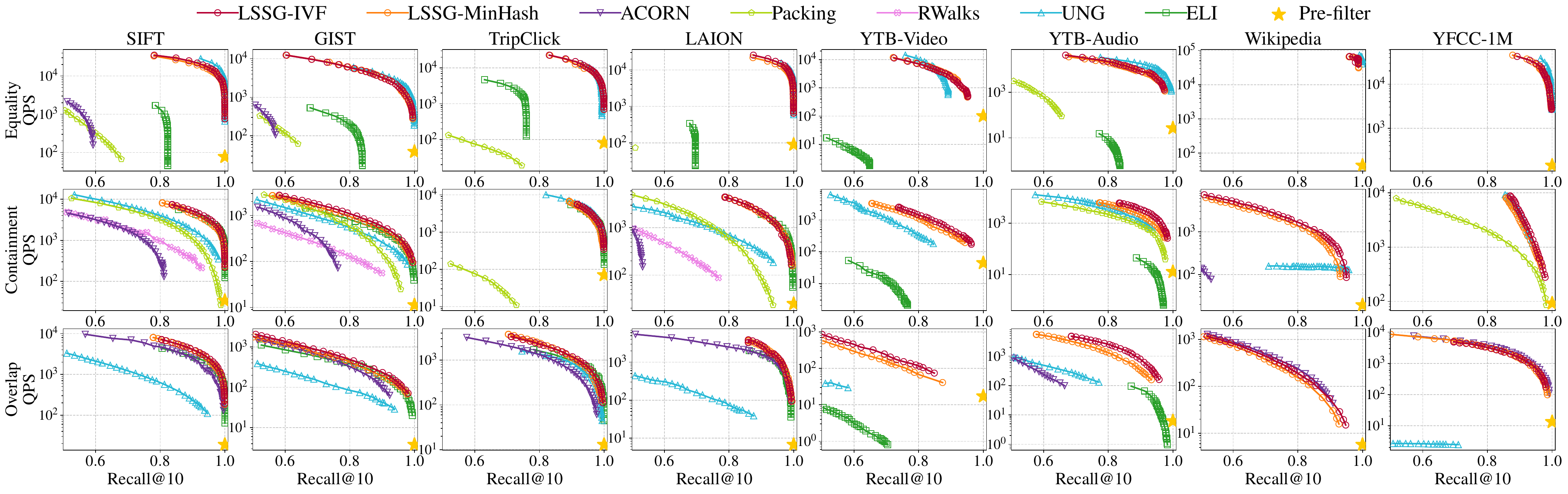}
    \caption{Queries Per Second (QPS) of LFANNS indices}
    \Description{Queries Per Second (QPS) of LFANNS indices}
    \label{fig:qps_recall}
\end{figure*}

\begin{figure}
    \centering
    \includegraphics[width=\linewidth]{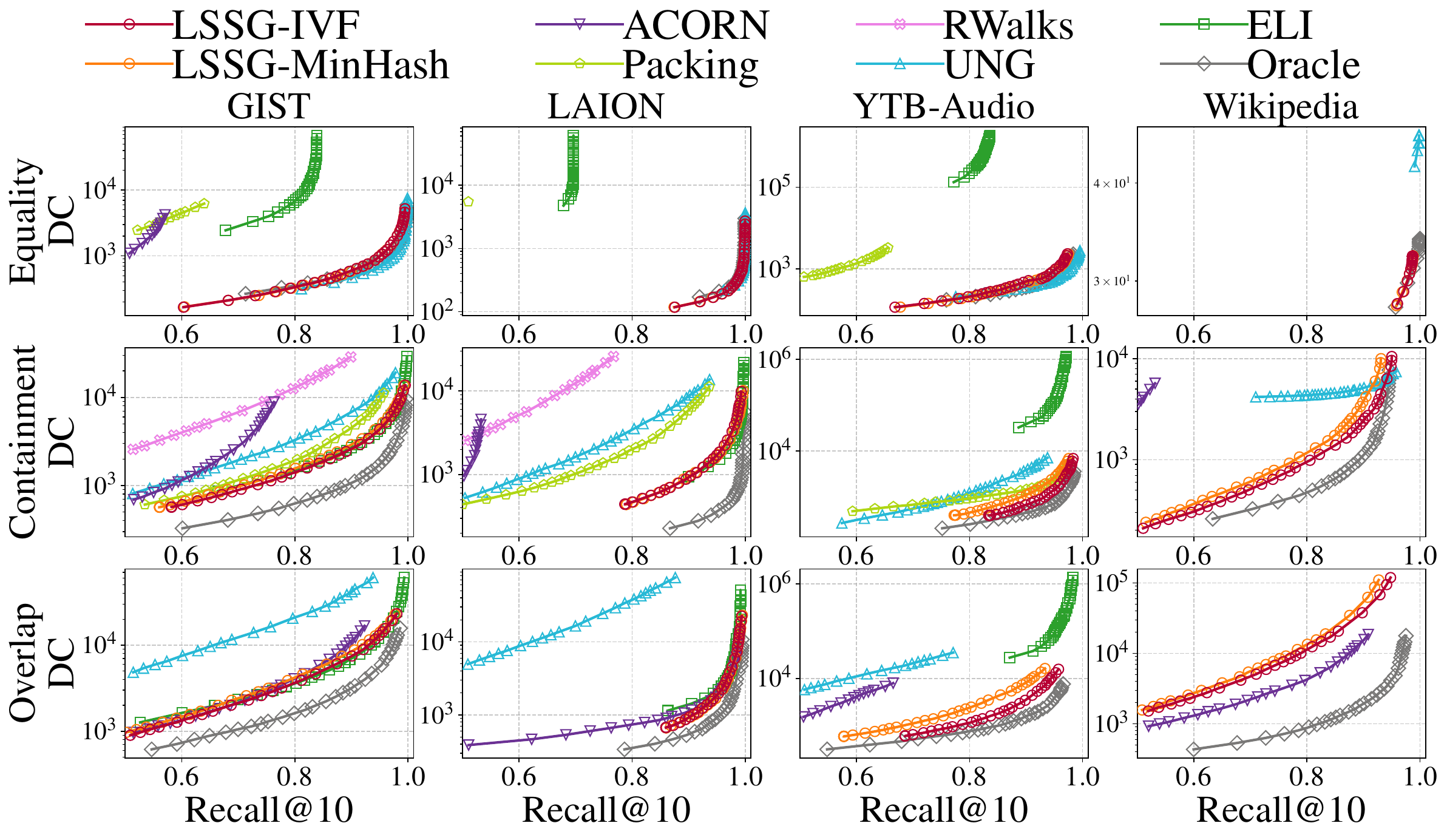}
    \caption{Distance Computations (DC) of LFANNS indices}
    \Description{Distance Computations (DC) of LFANNS indices}
    \label{fig:dc_recall}
\end{figure}

\subsection{RQ1: Indexing Analysis}
Table~\ref{tab:index_time_size} lists the indexing performance under 16-thread parallelism.
\begin{enumerate}[wide]
    \item \emph{Index size.}
    LSSG-IVF and LSSG-MinHash are the smallest (or among the most compact) indices across datasets.
    Averaged over datasets, compared with the most space-efficient baseline UNG, LSSG-IVF and LSSG-MinHash use 0.82x and 0.84x of UNG's size, respectively.
    Specifically, MinHash adds signature and LSH-bucket storage, but it often probes far fewer label sets than IVF scanning, shrinking the insertion candidate set $U$ and producing a sparser final graph that offsets the overhead. 
    So, LSSG-MinHash is smaller than LSSG-IVF on SIFT, GIST, and YTB-Audio.

    \item \emph{Indexing time.}
    LSSG-MinHash achieves the best or near-best indexing time on all datasets, averaging 0.78x UNG's and 0.87x ELI's time (the two fastest baselines).
    LSSG-IVF is also comparable to UNG and ELI.
    This is consistent with Sec.~\ref{sec:method_insertion}, where LSSG combines fast similar label set selection (IVF/MinHash) with dual-space pruning to reduce extensive vector-distance computations.

    \item \emph{Sensitivity to label characteristics.}
    LSSG's indexing cost remains stable as $|A|$ increases, indicating it is dominated by local graph operations rather than combinatorial label enumeration.
    ACORN and UNG show similar insensitivity, but several baselines degrade. 
    Packing is fast on YFCC-1M but less competitive when $|A|$ is small (e.g., LAION). 
    RWalks is highly sensitive to $|A|$ and fails on YTB-Audio, Wikipedia, and YFCC-1M due to attribute-vector enhancement overhead.
    ELI is fast on small-to-medium label spaces but spends large space overhead to meet the elastic factor (about twice LSSG on the first four datasets, $\sim$20x on YTB-Video, and $\sim$90x on YTB-Audio), and fails on Wikipedia and YFCC-1M as index sharing scales exponentially with labels per set (e.g., 40 labels $\Rightarrow 2^{40}$ combinations).
    Overall, LSSG is less affected by label cardinality than prior LFANNS indices while consistently outperforming them.
\end{enumerate}

\subsection{RQ2: LFANNS Query Performance}

\bparagraph{Overall comparison.}
Fig.~\ref{fig:qps_recall} shows QPS--Recall@10 curves in one thread (below 0.5 recall are omitted), following~\cite{ung,tfanns,eli,acorn}.
ELI and Packing lack \emph{native} support for equality. 
To avoid adding extra index components, we evaluate them by reducing equality to a strengthened containment check (i.e., candidates must contain the query labels and have the same label set cardinality).

\begin{figure}
    \centering
    \includegraphics[width=\linewidth]{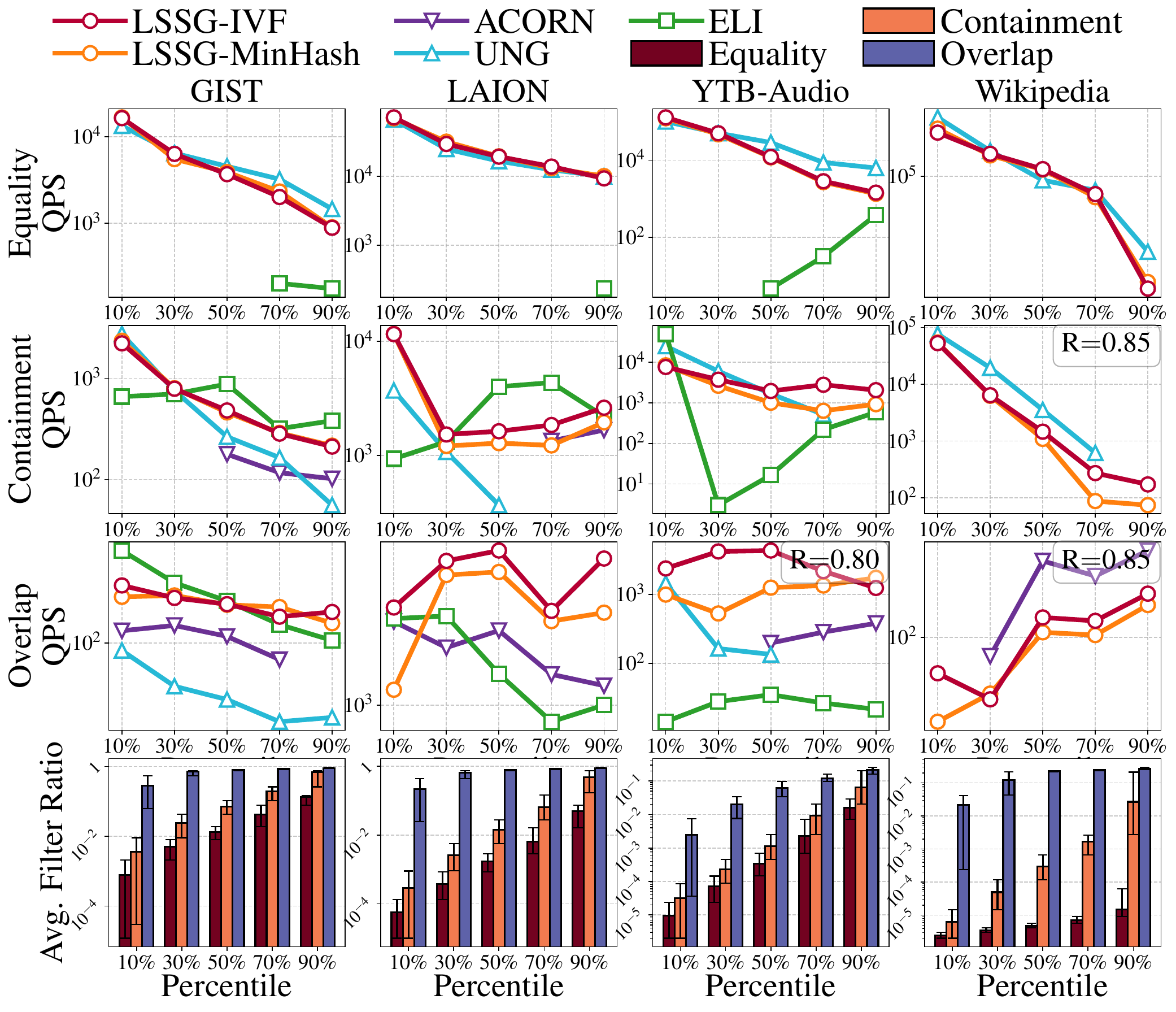}
    \caption{QPS at 0.9 recall across different selectivity percentiles. Recall is adjusted and labeled in the top-right corner when no baseline attains 0.90 recall. The bottom row shows the average filter ratio for each percentile.}
    \Description{QPS at 0.9 recall across different selectivity percentiles. Recall is adjusted and labeled in the top-right corner when no baseline attains 0.90 recall. The bottom row shows the average filter ratio for each percentile.}
    \label{fig:qps_sel}
\end{figure}

\begin{enumerate}[wide]
    \item \emph{Equality.}
    LSSG and UNG show the best performance.
    UNG is sometimes slightly more accurate (e.g., on GIST, YTB-Audio, and Wikipedia) for its denser \emph{static} Vamana underlay with looser RNG pruning.
    This gap reflects the underlay choice (static Vamana vs.\ incremental HNSW), rather than the label-filtering mechanism.
    LSSG achieves a comparable QPS--Recall frontier while supporting \emph{native incremental updates} through HNSW-style insertion.

    \item \emph{Containment.}
    LSSG provides the strongest overall performance.
    Among the baselines, UNG is most competitive, while LSSG-MinHash achieves 2.95x faster on average at matched recall.
    The second best ELI is strong on SIFT--LAION, where LSSG-MinHash achieves 1.10x faster on average.
    However, ELI cannot reach 0.9 recall on YTB-Video, and on YTB-Audio, LSSG-MinHash achieves 92.9x faster at 0.9 recall.
    Furthermore, ELI cannot be built on Wikipedia or YFCC-1M.
    Packing is competitive on SIFT, GIST, YTB-Audio, and YFCC-1M, but is less consistent overall, while RWalks and ACORN show limited performance.

    \item \emph{Overlap.}
    LSSG and ACORN provide the strongest overall performance.
    On Wikipedia, LSSG-IVF reaches 0.948 recall, exceeding ACORN's peak recall of 0.909.
    ELI is competitive on SIFT--LAION, where LSSG-MinHash achieves 1.34x faster on average, but degrades on larger-label workloads.
    On YTB-Video, LSSG-MinHash achieves 84.1x faster at 0.7 recall.
    UNG is generally weaker.
    On SIFT--TripClick, where it reaches 0.9 recall, LSSG-MinHash achieves 7.3x faster on average, with the gap increasing to 99.8x on YFCC-1M at matched recall.
    Packing and RWalks do not reach 0.5 recall for overlap, with peak recall of only 0.41 and 0.32 on SIFT, respectively.

    \item \emph{IVF vs.\ MinHash in LSSG.} LSSG-IVF can be slightly better on a few workloads (notably containment/overlap on YTB-Audio, Wikipedia, and YFCC-1M).
    MinHash probing prioritizes high similarity label sets and may miss low similarity yet admissible ones, while IVF scanning evaluates more Jaccard distances and thus covers a broader set of allowed label sets---an effect that matters more with many labels and for overlap. 
    Despite this trade-off, LSSG-MinHash builds faster and retains strong overall query throughput.
\end{enumerate}

\bparagraph{Approximation to the oracle label-filtered RNG.}
Fig.~\ref{fig:dc_recall} compares distance computations (DC) with an oracle label-filtered RNG (Def.~\ref{def:oracle}), implemented as HNSW on the fully filtered subset.
For equality, LSSG and UNG stay closest to the oracle on GIST, LAION, and YTB-Audio.
On Wikipedia, UNG is within about 10 DC of the oracle.
For containment, LSSG is consistently closest, while Packing is occasionally competitive and UNG, ACORN, and RWalks generally require more DC or attain lower recall.
ELI degrades on YTB-Audio and cannot be built on Wikipedia.
For overlap, LSSG remains closest on GIST, LAION, and YTB-Audio, with ACORN closer only on Wikipedia.
Overall, LSSG most consistently approximates the oracle, empirically supporting the connection between the stepwise advantage analyzed in Sec.~\ref{sec:analysis} and reduced search effort.
Appendix~\ref{app:all_exps} reports results on the remaining datasets.

\begin{figure}
    \centering
    \includegraphics[width=\linewidth]{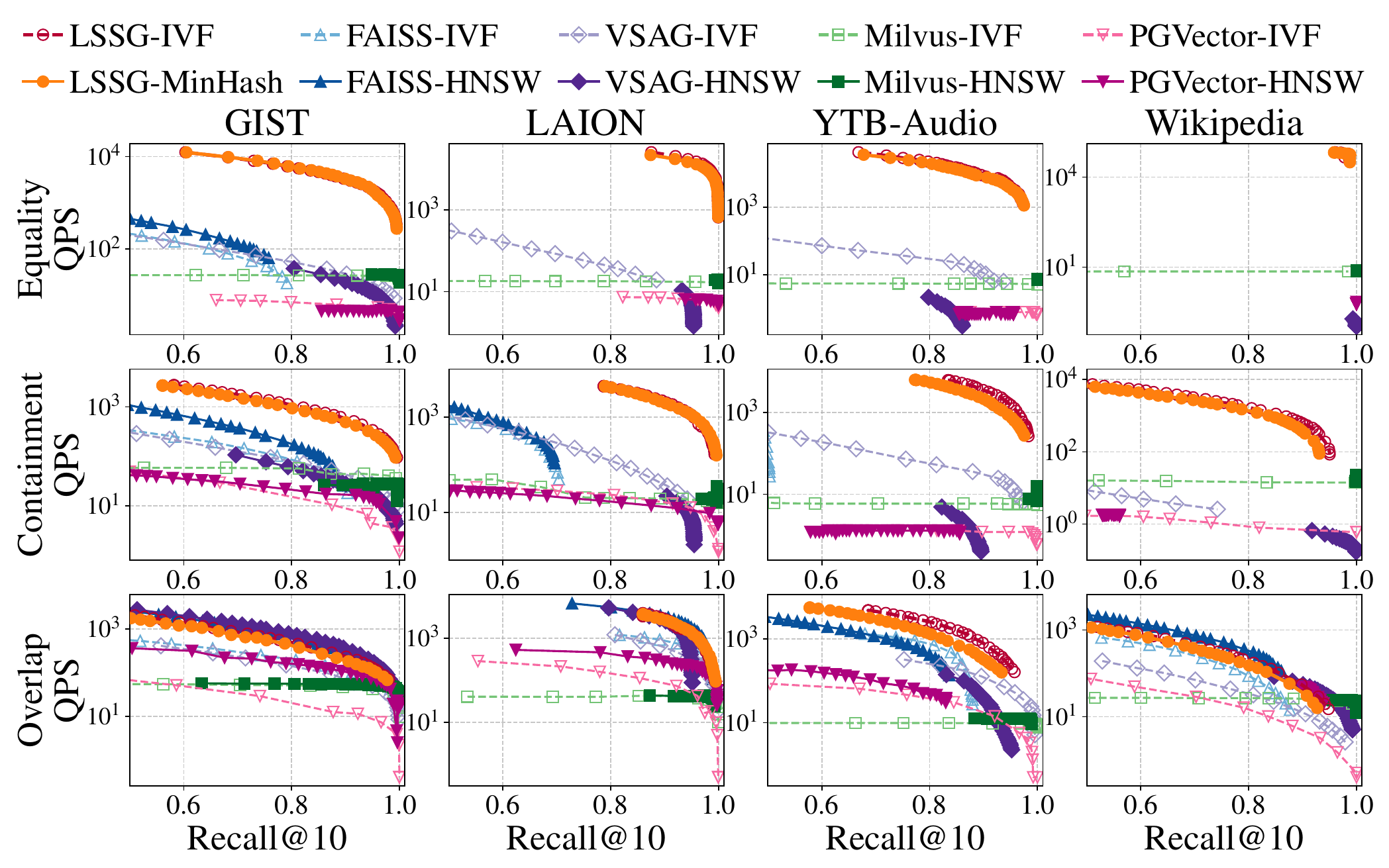}   
    \caption{Comparison with ANNS libraries \& vector databases}
    \Description{Comparison with ANNS libraries \& vector databases}
    \label{fig:vdb-compare}
\end{figure}

\begin{figure}
    \centering
    \includegraphics[width=0.99\linewidth]{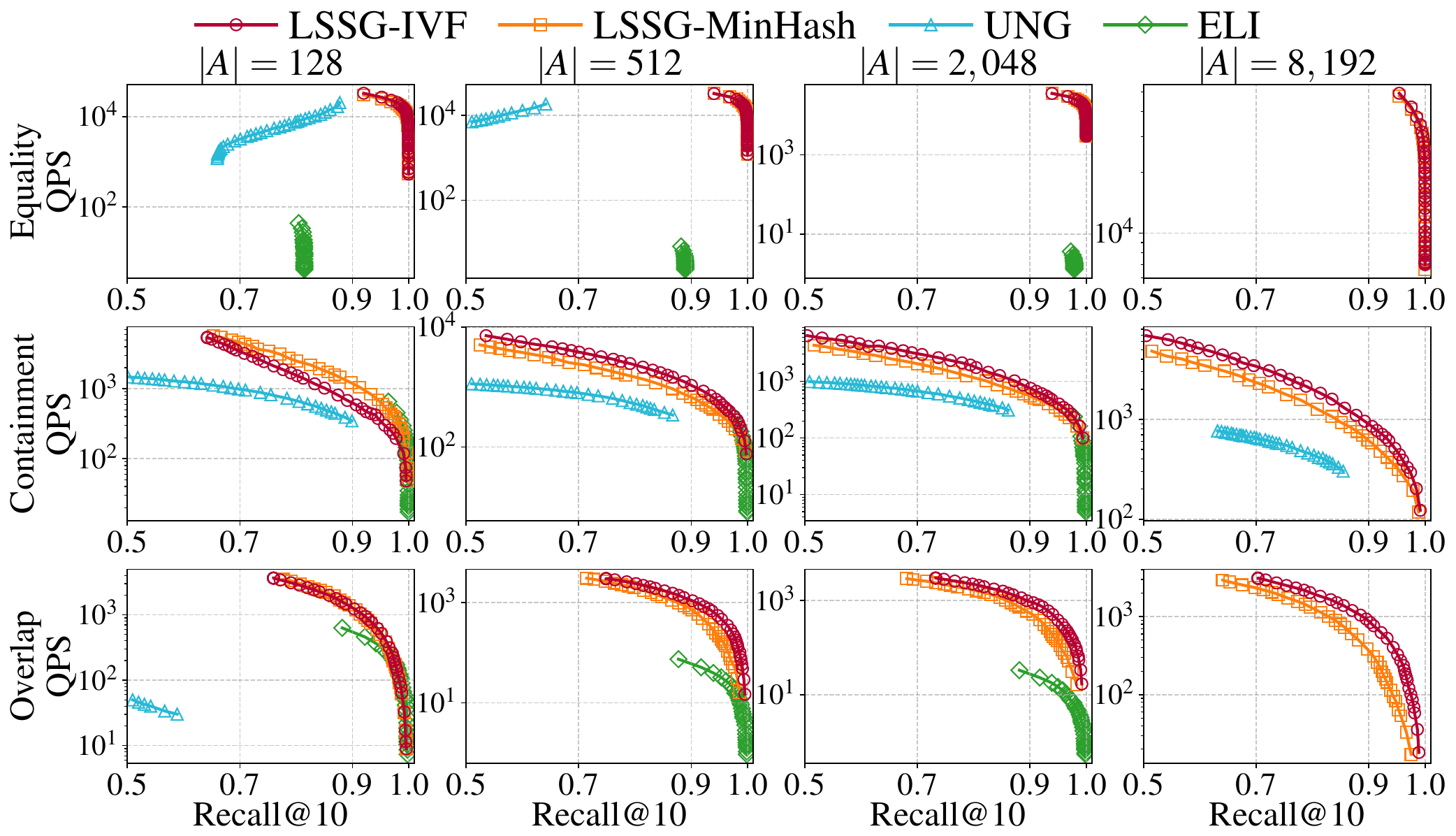}
    \caption{Query performance w.r.t. varying label scale $|A|$}
    \Description{Query performance w.r.t. varying label scale $|A|$}
    \label{fig:varying_labelnum}
\end{figure}

\bparagraph{Performance w.r.t. different selectivity.}
In Fig.~\ref{fig:qps_sel}, we sort queries by selectivity and split each workload into five equal-sized percentiles.
LSSG is the most stable method across the entire selectivity spectrum.
Averaged over datasets, for containment, LSSG-MinHash achieves 3.87x, 0.93x, 1.51x, 3.08x, and 8.48x QPS of UNG and 8.38x, 281x, 43.1x, 57.6x, and 5.29x of ELI at identical recall, grouped by percentiles from 10\%, 30\%, 50\%, 70\% to 90\%.
For overlap, the improvements over UNG are 7.1x, 39.2x, 31.3x, 47.1x, and 52.7x, and over ELI are 80.2x, 77.3x, 23.5x, 15.7x, and 6.5x.
See Appendix~\ref{app:all_exps} for results on the remaining datasets.
\begin{enumerate}[wide]
    \item \emph{High selectivity (10\%).}
    With small filtered subsets, quickly restricting search space can be effective.
    Thus, UNG is competitive for containment, whose cross-group edges connect a small number of relevant label sets, explaining the smaller 3.87x average containment gap.
    For overlap, ELI is strongly penalized by its multi-round containment reduction, yielding a large gap (80.2x).

    \item \emph{Medium selectivity (30\%--70\%).}
    UNG is competitive around 30\%, where LSSG-MinHash achieves 0.93x QPS on average, but UNG degrades as selectivity decreases because its per-label set sub-indices become increasingly fragmented, and cross-group edges are less effective for traversing across many relevant label sets.
    ELI is most unstable for containment in this regime (LSSG/ELI peaks at 281x at 30\% and stays at 43.1x--57.6x from 50\% to 70\%).
    Its heuristic may select a shared index that is much larger than the true filtered subset, so many visited candidates are filtered out.
    LSSG remains stable by tiered navigation to keep progress within label-consistent regions and completing the search with a single in-filter traversal.

    \item \emph{Low selectivity (90\%).}
    With large filtered subsets (common for overlap), label-agnostic traversal is less penalized, making ACORN more competitive, especially on Wikipedia.
    UNG becomes particularly weak for overlap and the gap grows to 52.7x, reflecting limited cross-label connectivity.
    ELI's overlap gap narrows to 6.5x, suggesting multi-round overhead is less harmful when each round is only weakly selective.
    LSSG remains faster due to one-pass traversal.
    
\end{enumerate}

\begin{figure}
    \centering
    \includegraphics[width=\linewidth]{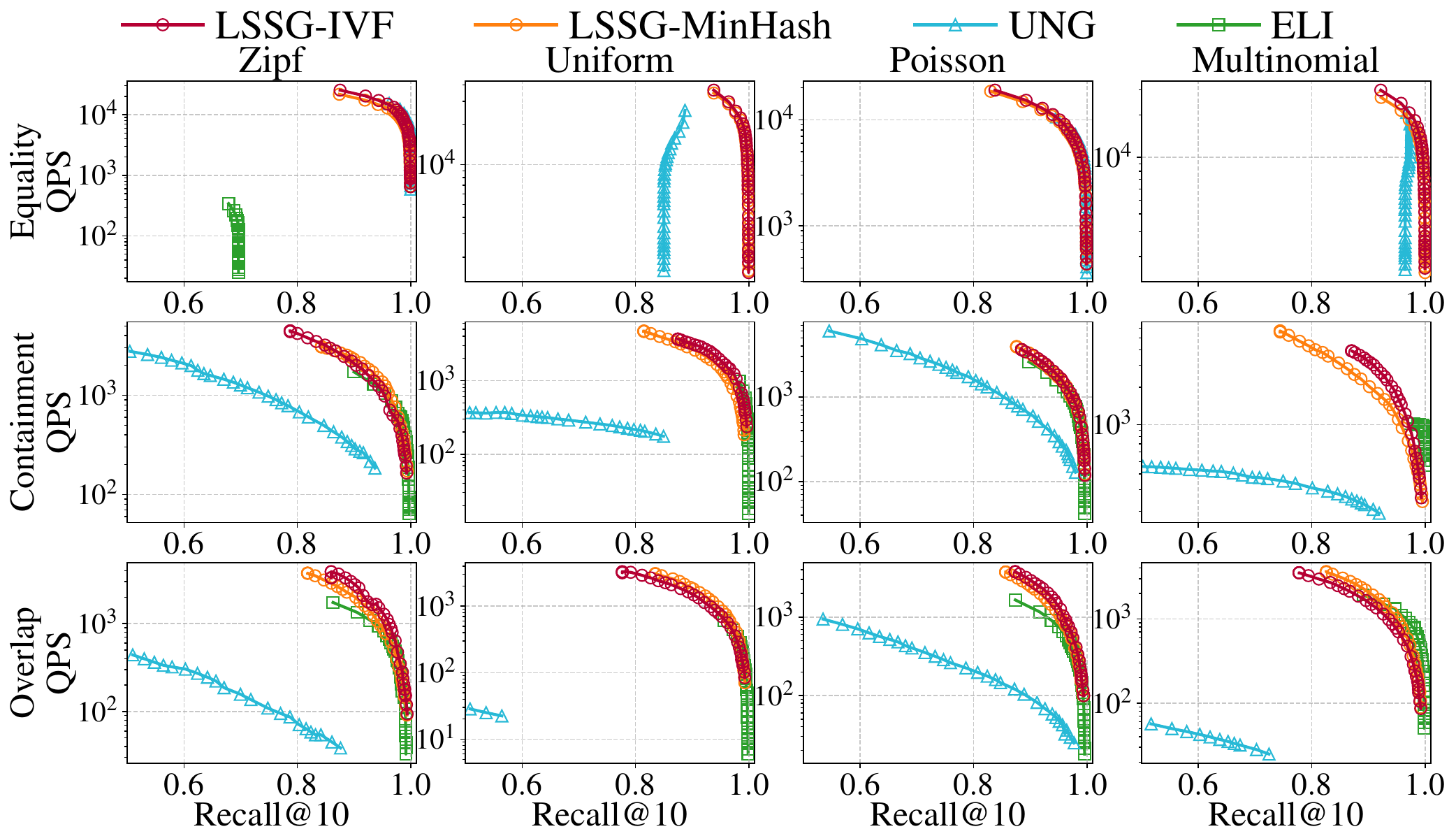}
    \caption{Performance w.r.t. different label distributions}
    \Description{Performance w.r.t. different label distributions}
    \label{fig:distribution}
\end{figure}

\begin{figure}
    \centering
    \includegraphics[width=\linewidth]{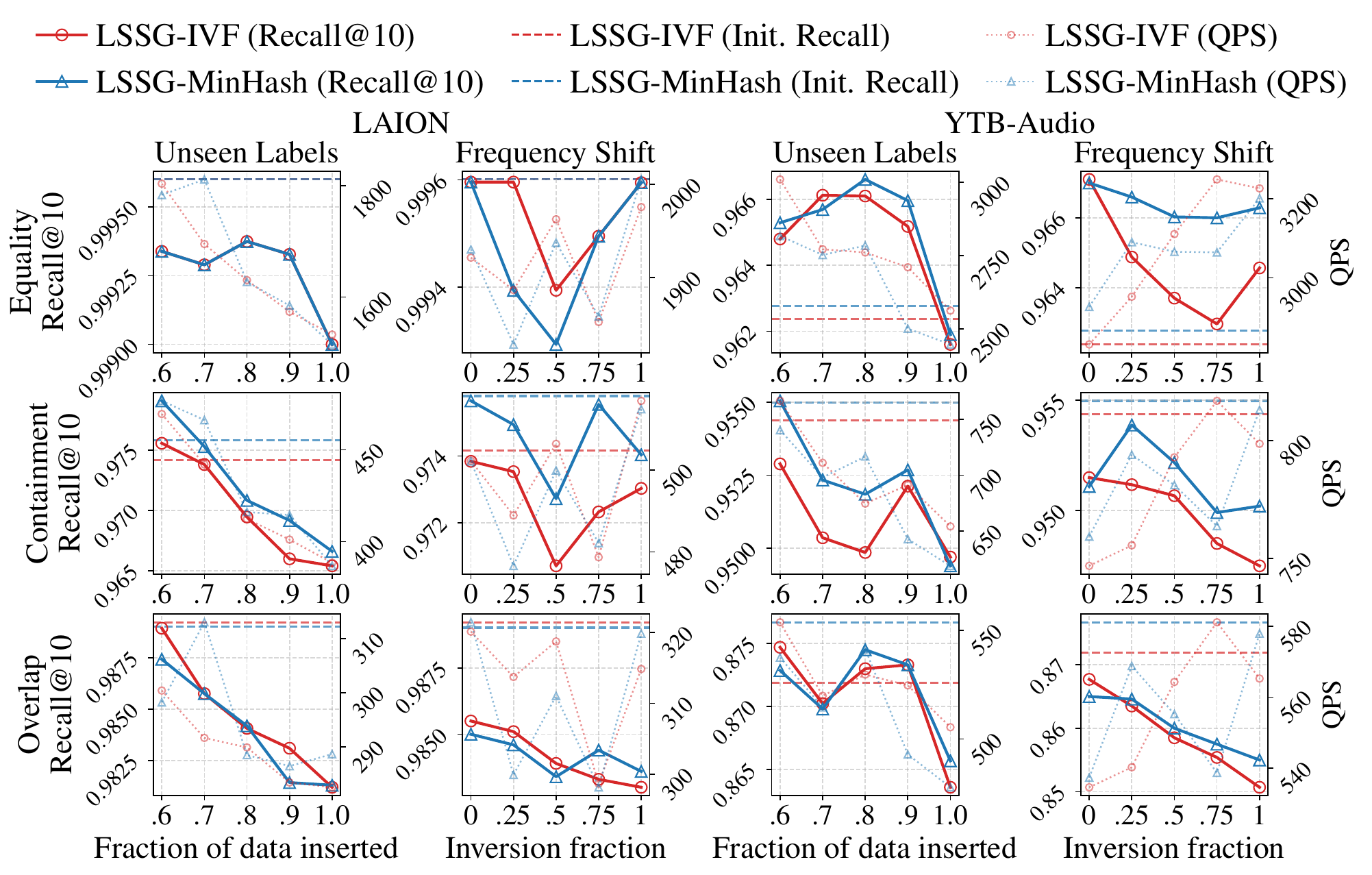}
    \caption{Insertion-drift evaluation on LAION and YTB-Audio with $\omega_s=500$.
    Dashed lines mark pre-insertion recall under the initial index. Dotted lines show QPS.}
    \Description{Insertion-drift evaluation on LAION and YTB-Audio.
    Dashed lines mark pre-insertion recall under the initial index. Dotted lines show QPS.}
    \label{fig:insertion-drift}
\end{figure}

\bparagraph{Comparison with ANNS libraries and vector databases.}
Fig.~\ref{fig:vdb-compare} compares LSSG with FAISS, VSAG, Milvus, and PGVector on GIST, LAION, YTB-Audio, and Wikipedia.
\begin{enumerate*}[left=0pt]
\item \emph{At the library level}, FAISS cannot reach 0.95 recall for equality or containment on any dataset, with maximum recalls of 0.002--0.79 and 0.13--0.90, respectively.
VSAG reaches high recall more often, but drops to 0.19 and 0.46~QPS at 0.95 recall for equality and containment on Wikipedia, and remains below 0.95 for equality on YTB-Audio.
For overlap, FAISS remains below 0.95 recall on Wikipedia and YTB-Audio, whereas VSAG reaches 0.95 recall at only 27--52~QPS.
\item \emph{At the database level}, the best Milvus and PGVector variants reach at least 0.95 recall for containment, but achieve only 16--41 and 0.6--13~QPS, respectively.
LSSG maintains high recall with substantially higher throughput across all three semantics.
\end{enumerate*}
Overall, integrating filtering into graph traversal yields a more consistent LFANNS query than library-level filtering or database-level query processing.
Results on the remaining datasets and indexing costs are in Appendix~\ref{app:all_exps}.

\bparagraph{Performance w.r.t. varying label scale.}
Fig.~\ref{fig:varying_labelnum} increases the label space ($|A|$) and unique label sets ($422{,}427\rightarrow 972{,}650$) of LAION.
LSSG is consistently superior, as tier-wise distance bounds enforce label constraints without enumerating label combinations.
ELI fails at $|A|=8{,}192$ due to rapidly growing index-sharing search space.
UNG degrades due to massive fragmented groups.

\bparagraph{Performance w.r.t. different label distributions.}
Following~\cite{ung}, we generate Uniform, Poisson, and Multinomial label distributions on LAION.
Fig.~\ref{fig:distribution} compares the strongest indices LSSG, UNG, and ELI.
\begin{enumerate*}[left=0pt]
    \item LSSG performs best, as tiered traversal depends on label set \emph{distance} rather than a particular frequency shape, maintaining stable navigation under all distributions.
    \item For UNG, equality degrades under Uniform/Multinomial, since more diverse label sets spread the edge budget over more groups and weaken within-group connectivity.
    \item ELI is competitive for containment/overlap but unstable for equality (0.69 under Zipf and below 0.1 under others), due to equality-incompatible design and distribution sensitivity.
\end{enumerate*}

\begin{figure}
    \centering
    \includegraphics[width=\linewidth]{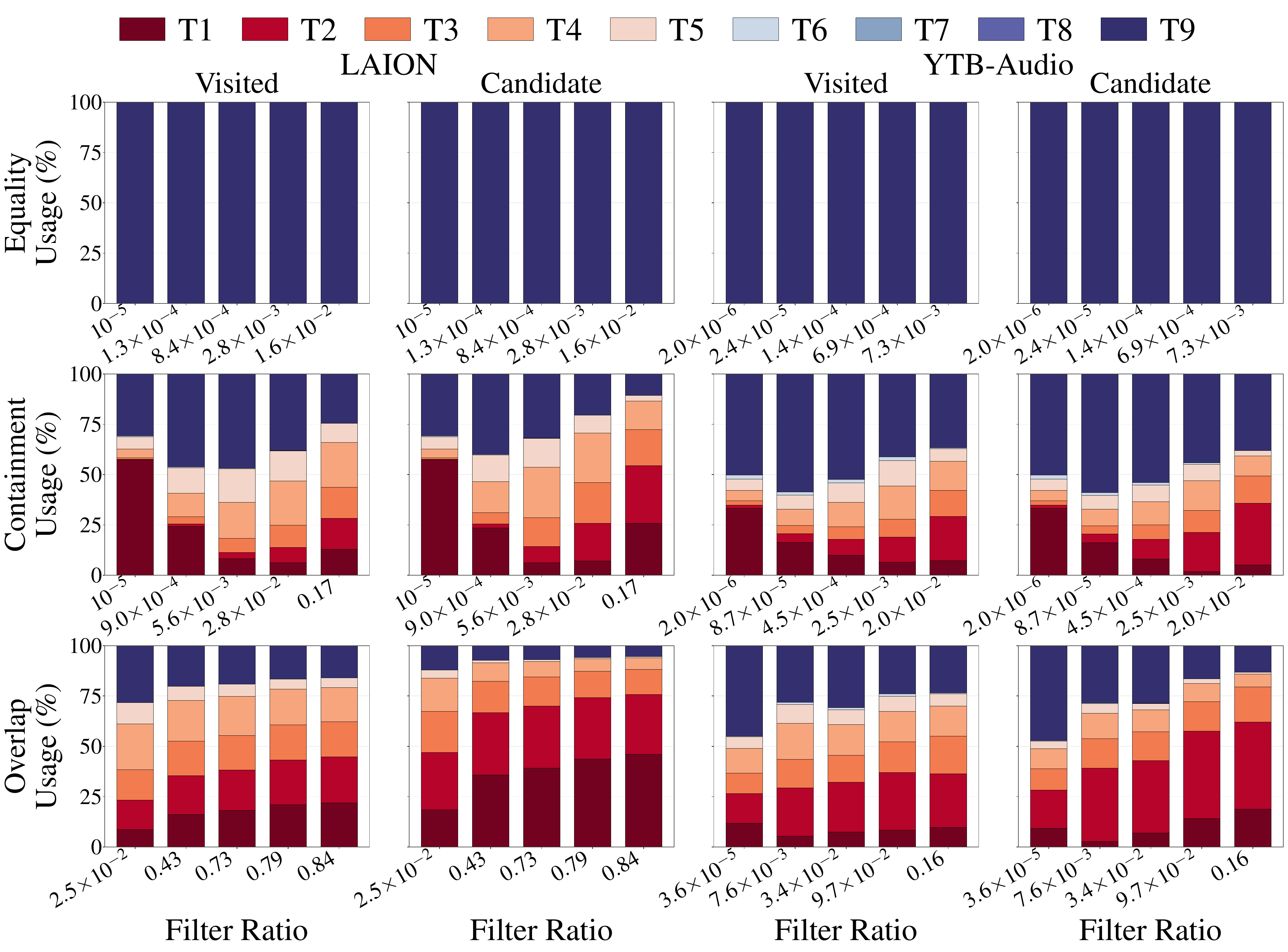}
    \caption{Tier usage on LAION. Each bar decomposes either visited vertices or final candidates, grouped by filter ratio.}
    \Description{Tier usage on LAION. Each bar decomposes either visited vertices or final candidates, grouped by filter ratio.}
    \label{fig:tier-use}
\end{figure}

\bparagraph{Performance w.r.t. label insertion drift.}
Fig.~\ref{fig:insertion-drift} evaluates unseen-label insertion and frequency inversion after indexing the first 50\% of each dataset.
Under unseen-label insertion, containment recall changes by at most 0.014 on LAION and 0.006 on YTB-Audio, while equality and overlap recalls decrease by at most 0.017.
Under frequency inversion, LSSG-MinHash remains stable, where LSSG-IVF drops by up to 0.017 on YTB-Audio due to frequency-based tier assignment.
This robustness comes from LSSG's tiered thresholds.
Each insertion follows the same label similarity constraints and updates only the affected local neighborhoods.
MinHash further improves stability through frequency-independent candidate selection.
Both variants remain nearly unchanged on LAION.
Insertion throughput reaches 3{,}010--5{,}668~vectors/s, with index size growing from 130--212~MB and 520--841~MB.

\subsection{RQ3: Micro Benchmarks}
%
We dissect LSSG's key components and quantify their impact on index size, time, and query efficiency: (1) tier-usage decomposition and parameter sensitivity, (2) filter-valid expansion ratio, (3) label set diversification (LabelPrune), (4) similar label set selection (IVF scanning vs.\ MinHash probing), and (5) scalability to large label space.
Appendix~\ref{app:all_exps} validates more label set distance sensitivity.

\begin{figure}[t]
    \centering
    \captionsetup[subfigure]{skip=0.5pt, belowskip=1.5pt}

    \begin{subfigure}[t]{\linewidth}
        \centering
        \includegraphics[width=\linewidth]{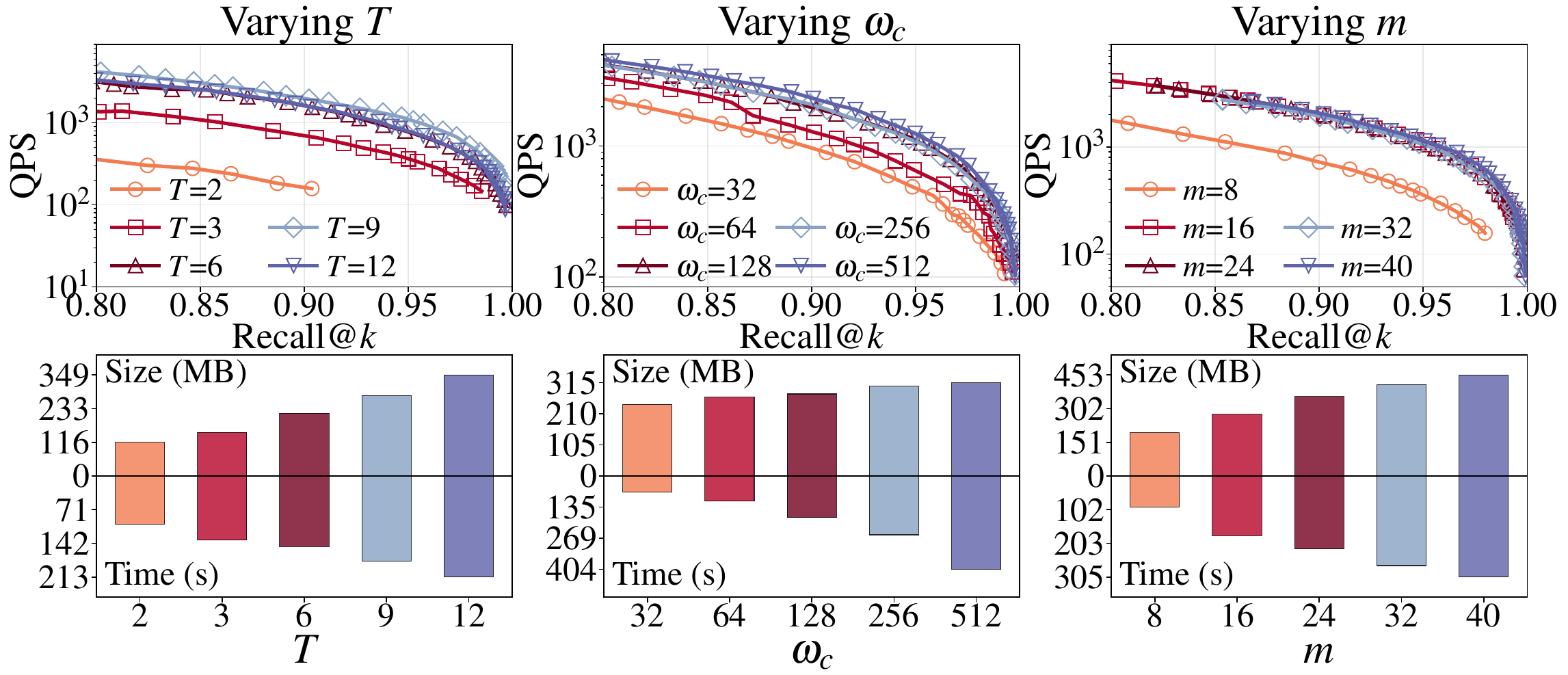}
        \caption{Index parameters}
    \end{subfigure}
    \par\vspace{-4pt}

    \begin{subfigure}[t]{\linewidth}
        \centering
        \includegraphics[width=\linewidth]{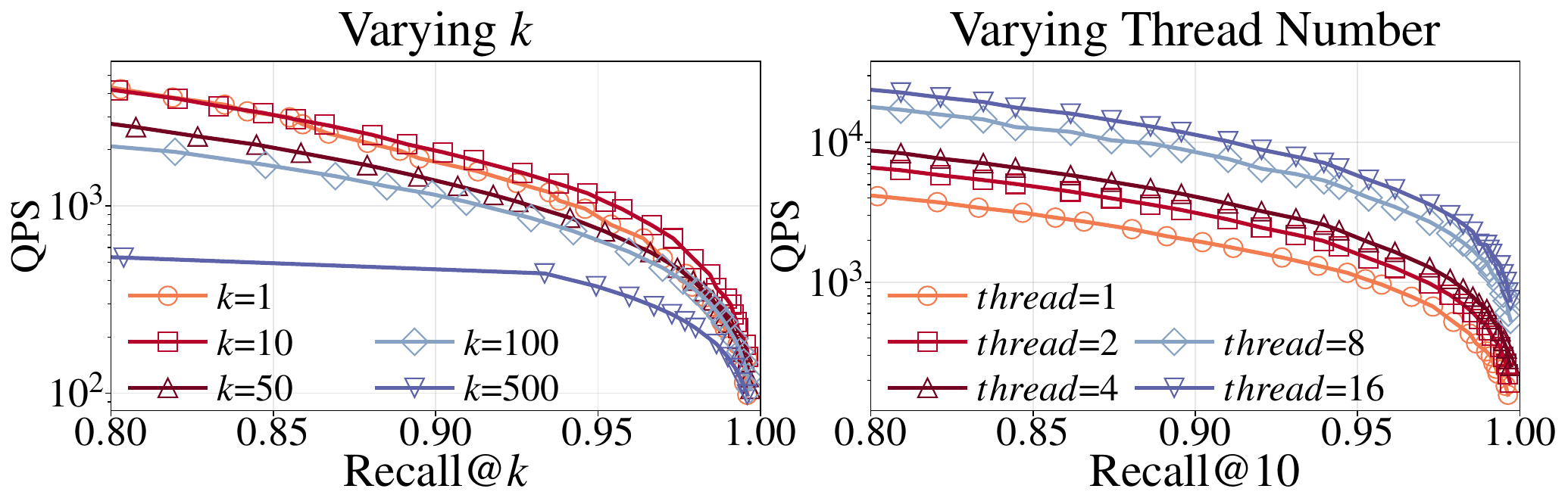}
        \caption{Output size $k$ and query threads}
    \end{subfigure}
    \vspace{-2pt}

    \caption{Parameter sensitivity on LAION}
    \Description{Parameter sensitivity on LAION}
    \label{fig:param-sensitivity}
\end{figure}

\begin{figure}[t]
\centering
\includegraphics[width=\linewidth]{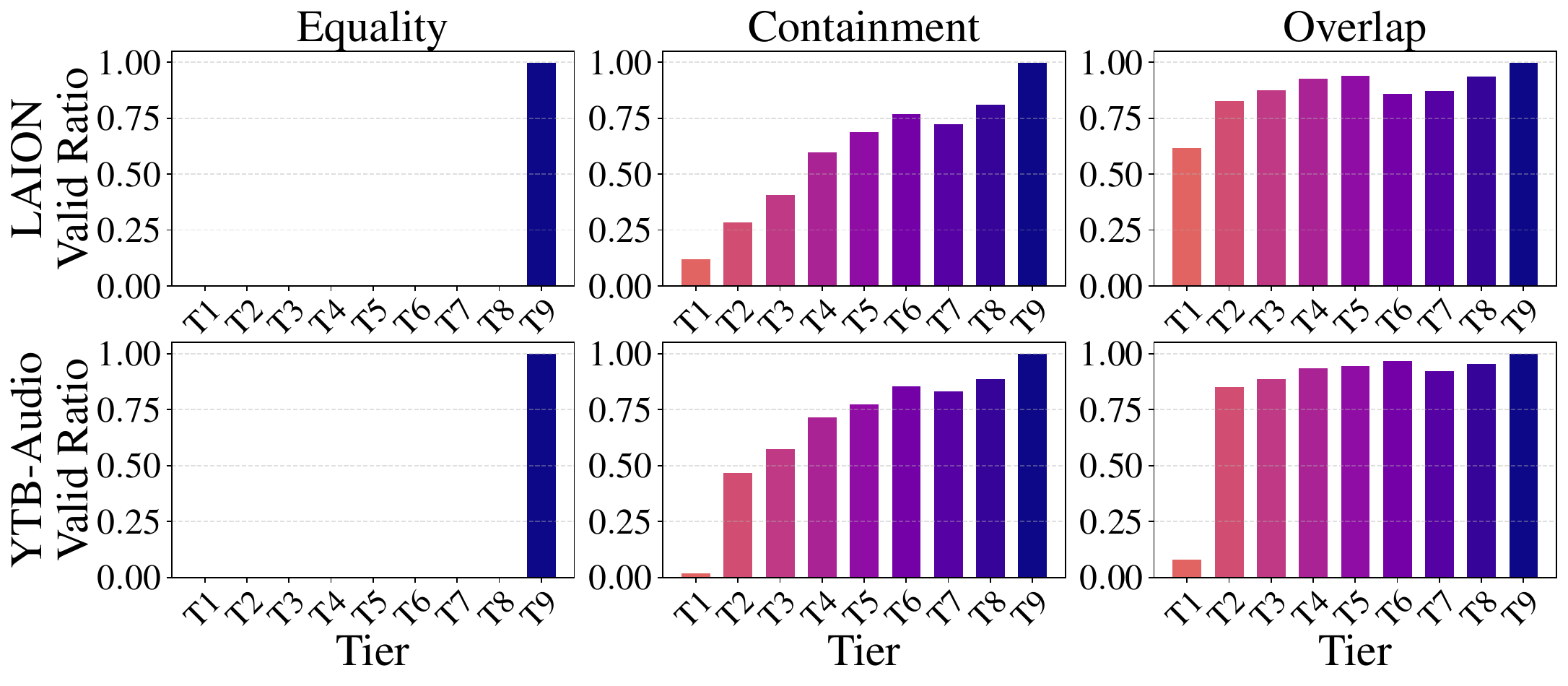}
\caption{Filter-valid expansion ratio on LAION.
Tier becomes stricter and higher ratio means fewer ineffective scans.}
\Description{Filter-valid expansion ratio on LAION.
Tier becomes stricter and higher ratio means fewer ineffective scans.}
\label{fig:expansion-ratio}
\end{figure}

\bparagraph{Tier use and sensitivity.}
Fig.~\ref{fig:tier-use} decomposes traversal by tier.
More selective filters shift visits and final candidates toward upper tiers, confirming the active usage of label-consistent edges.
For containment and overlap, fewer than ten queries per dataset use only the label-agnostic bottom tier.
Fig.~\ref{fig:param-sensitivity}(a) further shows that $T=2$ is insufficient at high recall, whereas intermediate tiers substantially improve accuracy.
Larger $m$ and $\omega_c$ improve recall but incur higher indexing costs with diminishing returns.

Fig.~\ref{fig:param-sensitivity}(b) confirms robust QPS--Recall@$k$ for $k\in\{1,10,50,100,500\}$ and varying threads.

\bparagraph{Filter-valid expansion ratio.}
Fig.~\ref{fig:expansion-ratio} shows the fraction of scanned tier neighbors satisfying the query filter, empirically reflecting the stepwise probability analyzed in Sec.~\ref{sec:analysis}.
The ratio generally increases with tier strictness, yet tiers~6--7 show minor drops since stricter tiers connect more similar label sets, which can still differ in a few labels and fail the exact filter.
From tier~1 to tier~9, the ratio rises on LAION/YTB-Audio from 0.12/0.02 to 1.0/1.0 for containment and from 0.62/0.08 to 1.0/1.0 for overlap.
Equality remains 1.0 since the search only contains exact matches.
Thus, lower tiers support broad exploration, while stricter tiers increase useful expansions.
\begin{figure}
    \centering
    \includegraphics[width=0.9\linewidth]{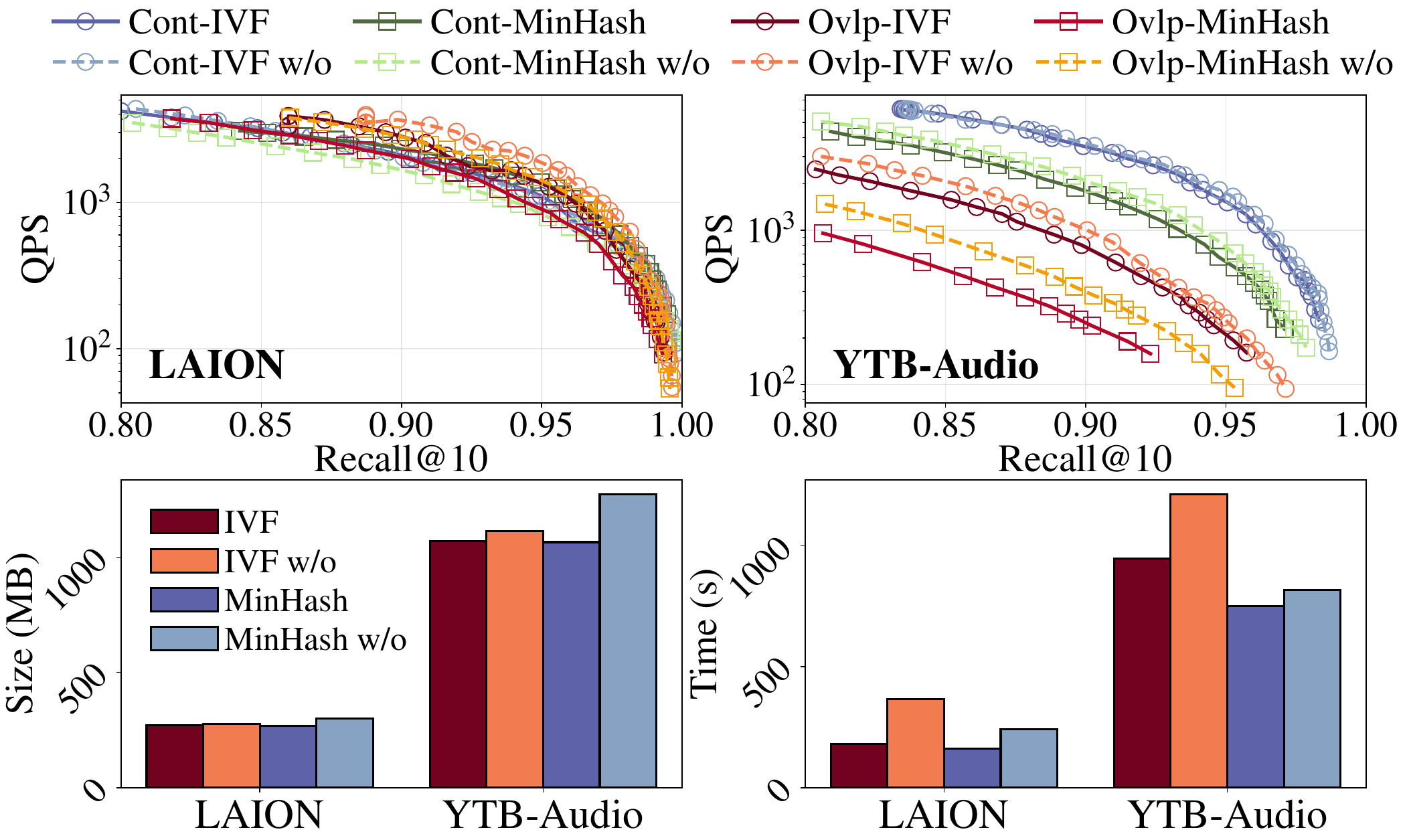}
    \caption{LSSG-IVF/MinHash performance w/o LabelPrune in the containment (Cont.) and overlap (Ovlp.) semantics}
    \Description{LSSG-IVF/MinHash performance w/o LabelPrune in the containment (Cont.) and overlap (Ovlp.) semantics}
    \label{fig:labelprune}
\end{figure}

\begin{figure}
    \centering
    \includegraphics[width=0.9\linewidth]{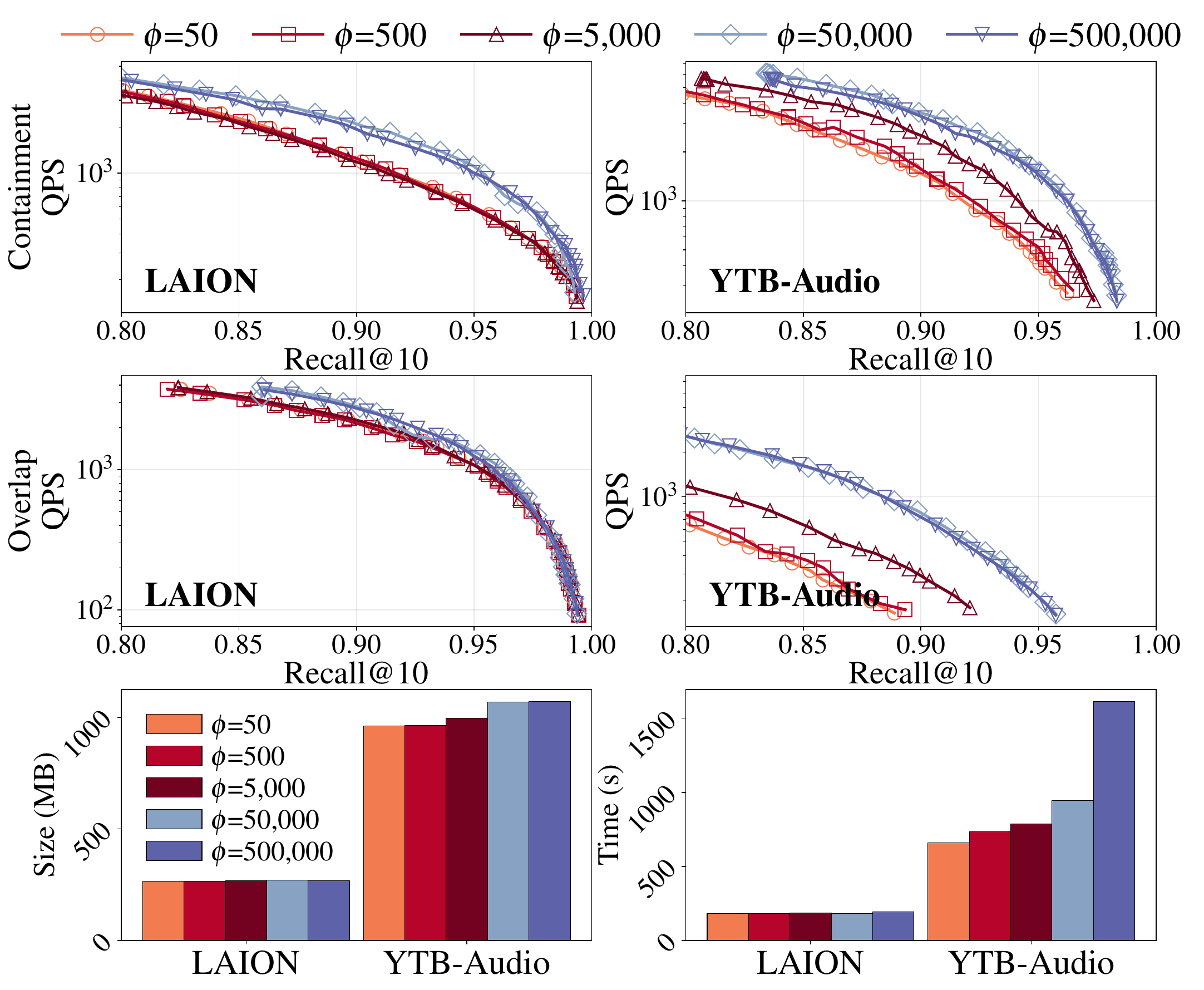}
    \caption{Performance w.r.t. varying IVF scanning budget $\phi$}
    \Description{Performance w.r.t. varying IVF scanning budget $\phi$}
    \label{fig:ivf_budget}
\end{figure}

\bparagraph{Effectiveness of label set diversification.}
LabelPrune removes redundant neighbors with highly similar label sets, reducing candidate processing and edge storage while preserving cross-label connectivity.
As Fig.~\ref{fig:labelprune} shows, the query impact is small.
At 0.9 recall, QPS drops by only $\sim$1\% on LAION and $<10\%$ on YTB-Audio, while building cost decreases substantially.
On LAION, LabelPrune reduces index size by 2\% (LSSG-IVF) and 10\% (LSSG-MinHash), and indexing time by 51\% and 33\%, respectively.
This aligns with Alg.~\ref{alg:insert}, where a sparser intermediate graph shrinks the candidate set $U$, boosting candidate selection that dominates insertion time.

\bparagraph{IVF scanning budget $\phi$.}
Fig.~\ref{fig:ivf_budget} varies $\phi$, the number of label bitmaps scanned to form $\mathcal{L}_q$ in Alg.~\ref{alg:beamsearch}.
For reference, $|\mathcal{L}|=62{,}066$ on LAION and $|\mathcal{L}|=498{,}445$ on YTB-Audio.
\begin{enumerate*}[left=0pt]
    \item Increasing $\phi$ improves QPS first, and then yields diminishing gains once IVF approaches a near-full scan.
    On LAION, at 0.9 recall, QPS increases from $1{,}098\rightarrow 2{,}008$ for containment and $2{,}049\rightarrow 2{,}770$ for overlap with $\phi=5{,}000\rightarrow 50{,}000$, while $\phi=500{,}000$ brings little further benefit (1{,}803 and 2{,}616).
    On YTB-Audio, QPS similarly saturates: $2{,}473\rightarrow 3{,}170 \rightarrow 3{,}239$ for containment and $268\rightarrow 802 \rightarrow 838$ for overlap with $\phi=5{,}000\rightarrow 50{,}000\rightarrow 500{,}000$.
    
    \item Larger $\phi$ increases bitmap unions and Jaccard computations during insertion and can add more out-neighbors.
    On YTB-Audio, indexing time grows as $\phi$ increases ($785\rightarrow 946\rightarrow1{,}613$\,s), while index size grows mildly ($997\rightarrow 1{,}069\rightarrow 1{,}071$\,MB).
    
    \item \emph{A practical knee point.}
    For large label scales, moderate $\phi$ can be near-optimal.
    On YTB-Audio, $\phi=50{,}000$ is close to $\phi=500{,}000$ in QPS at 0.9 recall (3{,}170 vs.\ 3{,}239; 802 vs.\ 838) but shortens indexing time by 41\% (946 vs.\ 1{,}613).
\end{enumerate*}

\begin{figure}
    \centering
    \includegraphics[width=0.95\linewidth]{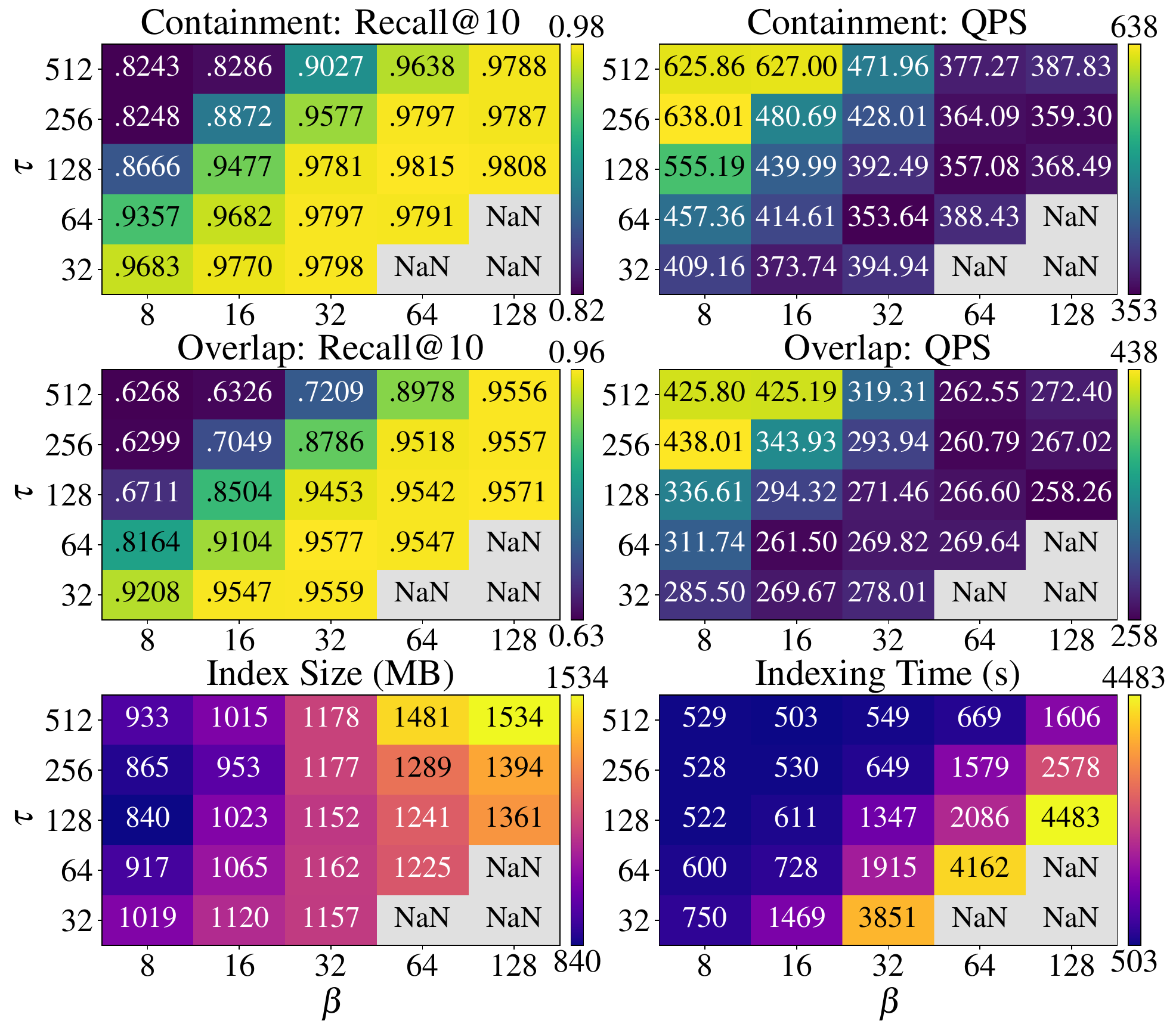}
    \caption{Performance w.r.t. varying MinHash $\tau,\beta$}
    \Description{Performance w.r.t. varying MinHash $\tau,\beta$}
    \label{fig:minhash_param}
\end{figure}

\begin{table}
    \centering
    \caption{Performance with large-scale labels}
    \label{tab:scalability}
    {\small
    \setlength{\tabcolsep}{3pt}
    \begin{tabular}{llccrrr}
        \toprule
        \multirow{2}{*}{\textbf{Datasets}} & \multirow{2}{*}{\textbf{Indices}} & \multicolumn{2}{c}{\textbf{Indexing}} & \multicolumn{3}{c}{\textbf{QPS-Recall@10}} \\
        \cmidrule(lr){3-4} \cmidrule(lr){5-7}
        & & \textbf{Size (MB)} & \textbf{Time (s)} & \textbf{0.85} & \textbf{0.90} & \textbf{0.95} \\
        \midrule
        \multirow{3}{*}{YFCC-1M}
        & Packing & 2,826 & \ \, 145 & 1,160 & 659 & 313 \\
        & LSSG-IVF & \ \, 239 & \ \, 379 & 8,788 & 4,120 & 880 \\
        & $\star$-MinHash & \ \, 303 & \ \, 111 & 8,393 & 3,346 & 766\\
        \midrule
        \multirow{3}{*}{YFCC-10M}
        & Packing & 2,424 & 4,409 & 12 & - & - \\
        & LSSG-IVF & 2,413 & 7,289 & 491 & 246 & 56 \\
        & $\star$-MinHash & 2,935 & 2,729 & 414 & 207 & 55\\
        \bottomrule
    \end{tabular}}
\end{table}

\bparagraph{MinHash probing analysis.}
Fig.~\ref{fig:minhash_param} shows the analysis on YTB-Audio for Sec.~\ref{sec:labelsel} by varying $(\tau,\beta)$ in Eq.~\eqref{eq:p_minhash_colli}.
\begin{enumerate*}[left=0pt]
    \item Reducing the probing selectivity by decreasing $\tau$ or increasing $\beta$ on the diagonal from upper left to bottom right to get higher collision probability increases recall but often lowers QPS due to longer traversals.
    
    \item Overly selective probing removes edges between moderately similar label sets, which are more important for overlap.
    For example, at $\beta=8$, increasing $\tau$ from 32 to 512 reduces recall by 0.14 for containment and 0.29 for overlap.
    
    \item Index size increases with $\beta$, and can be non-monotonic in $\tau$ when $\tau\le 16$ because larger $\tau$ sparsifies the graph (smaller) and lengthens signatures (larger).
    Indexing time increases with smaller $\tau$ and larger $\beta$ since more collisions make larger candidate pools and more $\delta_L$ computations.
\end{enumerate*}

\bparagraph{Index scalability.}
Table~\ref{tab:scalability} tests scalability on YFCC-10M (200{,}386 labels and $\sim$4M unique label sets) using released 10{,}000 containment queries.
\begin{enumerate*}[left=0pt]
    \item MinHash improves index scalability.
    LSSG-MinHash reduces indexing time by 62\% relative to LSSG-IVF, with a 21\% larger index size.
    Its QPS remains close (gap $<16\%$ at 0.9 recall).
   
    \item As dataset size increases, index size scales as $O(|\mathcal{V}|)$ and indexing time as $O(|\mathcal{V}|\log|\mathcal{V}|)$.
    Concretely, LSSG-MinHash shows a 9.6x size increase and a 24.5x indexing-time increase, consistent with Sec.~\ref{sec:complexity}.
    
    \item For other baselines, they hit scalability limits.
    UNG reaches only 0.6 peak recall at 0.93 QPS.
    ELI fails at this label scale due to exponential index-sharing enumeration.
    Packing builds in reasonable time but remains non-competitive in QPS.
\end{enumerate*}

\section{Conclusion}
We propose LSSG, a multi-tier proximity graph that addresses two key LFANNS challenges: consistency across query semantics and robustness to diverse label characteristics.
LSSG enforces label constraints via distribution-agnostic stratification by label set similarity, and scales to large label space with efficient similar label set selection and dual-space pruning.
We derive lower bounds on filter-valid next-hop probabilities under explicit label models, explaining how stricter tiers reduce ineffective expansions.
Experiments across datasets, label scales and distributions, selectivity regimes, and insertion-drift settings confirm fast query in different semantics and robust incremental construction.
Future work includes extending from label sets to sparse vectors for dense--sparse hybrid search and supporting richer filters such as regular expression.

\begin{acks}
We thank all reviewers for their valuable feedback.
This work was supported in part by the National Natural Science Foundation of China (No. 62676187) and the Postgraduate Research \& Practice Innovation Program of Jiangsu Province (No. 26CXJH0355).
\end{acks}

\bibliographystyle{ACM-Reference-Format}
\bibliography{main}

@inproceedings {nhq-nips23,
  author       = { Mengzhao Wang and Lingwei Lv and Xiaoliang Xu and Yuxiang Wang and Qiang Yue and Jiongkang Ni },
  title        = { An Efficient and Robust Framework for Approximate Nearest Neighbor Search with Attribute Constraint },
  booktitle    = { NeurIPS },
  year         = 2023,
  volume       = 36,
  pages        = { 15738--15751 },
  address      = { New Orleans, LA, USA },
  publisher    = { Curran Associates, Inc. }
}

@inproceedings{tfanns,
  author       = { Luo, Jiarui and Qiao, Miao and Zuo, Chaoji and Deng, Dong },
  title        = { Tag-Filtered Approximate Nearest Neighbor Search },
  booktitle    = { ICDE },
  year         = 2025,
  pages        = { 3642--3654 },
  address      = { Hong Kong, China },
  publisher    = { IEEE }
}

@inproceedings {filtered-diskann,
  author       = { Siddharth Gollapudi and Neel Karia and Varun Sivashankar and Ravishankar Krishnaswamy and Nikit Begwani and Swapnil Raz and Yiyong Lin and Yin Zhang and Neelam Mahapatro and Premkumar Srinivasan and Amit Singh and Harsha Vardhan Simhadri },
  title        = { Filtered-DiskANN: Graph Algorithms for Approximate Nearest Neighbor Search with Filters },
  booktitle    = { WWW },
  year         = 2023,
  pages        = { 3406--3416 },
  address      = { Austin, TX, USA },
  publisher    = { ACM }
}

@inproceedings{clip,
  author       = { Radford, Alec and Kim, Jong Wook and Hallacy, Chris and Ramesh, Aditya and Goh, Gabriel and Agarwal, Sandhini and Sastry, Girish and Askell, Amanda and Mishkin, Pamela and Clark, Jack and others },
  title        = { Learning Transferable Visual Models from Natural Language Supervision },
  booktitle    = { ICML },
  year         = 2021,
  pages        = { 8748--8763 },
  address      = { Virtual },
  publisher    = { PMLR }
}

@inproceedings{tripclick,
  author       = { Rekabsaz, Navid and Lesota, Oleg and Schedl, Markus and Brassey, Jon and Eickhoff, Carsten },
  title        = { TripClick: The Log Files of a Large Health Web Search Engine },
  booktitle    = { SIGIR },
  year         = 2021,
  pages        = { 2507--2513 },
  address      = { Virtual },
  publisher    = { ACM }
}

@inproceedings {milvus-vdb,
  author       = { Jianguo Wang and Xiaomeng Yi and Rentong Guo and Hai Jin and Peng Xu and Shengjun Li and Xiangyu Wang and Xiangzhou Guo and Chengming Li and Xiaohai Xu and Kun Yu and Yuxing Yuan and Yinghao Zou and Jiquan Long and Yudong Cai and Zhenxiang Li and Zhifeng Zhang and Yihua Mo and Jun Gu and Ruiyi Jiang and Yi Wei and Charles Xie },
  title        = { Milvus: A Purpose-Built Vector Data Management System },
  booktitle    = { SIGMOD },
  year         = 2021,
  pages        = { 2614--2627 },
  address      = { Xi'an, China },
  publisher    = { ACM }
}

@inproceedings{alayaDB,
  author       = { Deng, Yangshen and You, Zhengxin and Xiang, Long and Li, Qilong and Yuan, Peiqi and Hong, Zhaoyang and Zheng, Yitao and Li, Wanting and Li, Runzhong and Liu, Haotian and Mouratidis, Kyriakos and Yiu, Man Lung and Li, Huan and Shen, Qiaomu and Mao, Rui and Tang, Bo },
  title        = { AlayaDB: The Data Foundation for Efficient and Effective Long-context LLM Inference },
  booktitle    = { SIGMOD-Companion },
  year         = 2025,
  pages        = { 364–377 },
  numpages     = 14,
  address      = { Berlin, Germany },
  publisher    = { ACM }
}

@inproceedings {lsh,
  author       = { Yufei Tao and Ke Yi and Cheng Sheng and Panos Kalnis },
  title        = { Quality and Efficiency in High Dimensional Nearest Neighbor Search },
  booktitle    = { SIGMOD },
  year         = 2009,
  pages        = { 563--576 },
  address      = { Providence, RI, USA },
  publisher    = { ACM }
}

@inproceedings{lsh-lcc,
  author       = { Lei, Yifan and Huang, Qiang and Kankanhalli, Mohan and Tung, Anthony K. H. },
  title        = { Locality-Sensitive Hashing Scheme based on Longest Circular Co-Substring },
  booktitle    = { SIGMOD },
  year         = 2020,
  pages        = { 2589--2599 },
  address      = { Portland, OR, USA },
  publisher    = { ACM },
  numpages     = 11
}

@inproceedings {vbase-vdb,
  author       = { Qianxi Zhang and Shuotao Xu and Qi Chen and Guoxin Sui and Jiadong Xie and Zhizhen Cai and Yaoqi Chen and Yinxuan He and Yuqing Yang and Fan Yang and Mao Yang and Lidong Zhou },
  title        = { VBASE: Unifying Online Vector Similarity Search and Relational Queries via Relaxed Monotonicity },
  booktitle    = { OSDI },
  year         = 2023,
  pages        = { 377--395 },
  address      = { Boston, MA, USA },
  publisher    = { USENIX Association }
}

@inproceedings {pase-vdb,
  author       = { Wen Yang and Tao Li and Gai Fang and Hong Wei },
  title        = { PASE: PostgreSQL Ultra-High-Dimensional Approximate Nearest Neighbor Search Extension },
  booktitle    = { SIGMOD },
  year         = 2020,
  pages        = { 2241--2253 },
  address      = { Portland, OR, USA },
  publisher    = { ACM }
}

@inproceedings{glove,
  author       = { Pennington, Jeffrey  and Socher, Richard and Manning, Christopher },
  title        = { GloVe: Global Vectors for Word Representation },
  booktitle    = { EMNLP },
  year         = 2014,
  pages        = { 1532--1543 },
  address      = { Doha, Qatar },
  publisher    = { ACM },
}

@inproceedings{spann,
  author       = { Chen, Qi and Zhao, Bing and Wang, Haidong and Li, Mingqin and Liu, Chuanjie and Li, Zengzhong and Yang, Mao and Wang, Jingdong },
  title        = { SPANN: Highly-efficient Billion-scale Approximate Nearest Neighbor Search },
  booktitle    = { NeurIPS },
  year         = 2021,
  address      = { Red Hook, NY, USA },
  publisher    = { Curran Associates Inc. },
  numpages     = 14
}

@inproceedings{spfresh,
  author       = { Xu, Yuming and Liang, Hengyu and Li, Jin and Xu, Shuotao and Chen, Qi and Zhang, Qianxi and Li, Cheng and Yang, Ziyue and Yang, Fan and Yang, Yuqing and Cheng, Peng and Yang, Mao },
  title        = { SPFresh: Incremental In-Place Update for Billion-Scale Vector Search },
  booktitle    = { SOSP },
  year         = 2023,
  pages        = { 545–561 },
  address      = { Koblenz, Germany },
  publisher    = { ACM },
  numpages     = 17,
}

@inproceedings{jdvisual,
  author       = { Li, Jie and Liu, Haifeng and Gui, Chuanghua and Chen, Jianyu and Ni, Zhenyuan and Wang, Ning and Chen, Yuan },
  title        = { The Design and Implementation of a Real Time Visual Search System on JD E-commerce Platform },
  booktitle    = { Middleware },
  year         = 2018,
  pages        = { 9–16 },
  address      = { Rennes, France },
  publisher    = { ACM },
  numpages     = 8,
}

@inproceedings{minhash,
  author       = { Broder, A.Z. },
  title        = { On the resemblance and containment of documents },
  booktitle    = { SEQUENCES },
  year         = 1997,
  pages        = { 21--29 },
  address      = { Positano, Italy },
  publisher    = { IEEE },
  numpages     = 9,
}

@misc{fnv,
  author       = { Glenn Fowler and Landon Curt Noll and Kiem-Phong Vo and Donald E. Eastlake 3rd and Tony Hansen },
  title        = { {The FNV Non-Cryptographic Hash Algorithm} },
  year         = 2026,
  number       = 9923,
  url          = { https://www.rfc-editor.org/info/rfc9923 },
  howpublished = { RFC 9923 },
  pagetotal    = 111,
}

@misc{ungfilterbench,
  author        = { Shi, Jiayang and Cai, Yuzheng and Zheng, Weiguo },
  title         = { Filtered Approximate Nearest Neighbor Search: A Unified Benchmark and Systematic Experimental Study [Experiment, Analysis \& Benchmark] },
  year          = 2025,
  eprint        = { 2509.07789 },
  archiveprefix = { arXiv },
  primaryclass  = { cs.DB },
}

@misc{nipscomp2023,
  author        = { Simhadri, Harsha Vardhan and Aum{\"u}ller, Martin and Ingber, Amir and Douze, Matthijs and Williams, George and Manohar, Magdalen Dobson and Baranchuk, Dmitry and Liberty, Edo and Liu, Frank and Landrum, Ben and others },
  title         = { Results of the Big ANN: NeurIPS'23 Competition },
  year          = 2024,
  eprint        = { 2409.17424 },
  archiveprefix = { arXiv },
  primaryclass  = { cs.IR },
}

@misc{caps,
  author        = { Gupta, Gaurav and Yi, Jonah and Coleman, Benjamin and Luo, Chen and Lakshman, Vihan and Shrivastava, Anshumali },
  title         = { CAPS: A Practical Partition Index for Filtered Similarity Search },
  year          = 2023,
  eprint        = { 2308.15014 },
  archiveprefix = { arXiv },
  primaryclass  = { cs.IR },
}

@misc{eli,
  author        = { Yang, Mingyu and Xia, Wenxuan and Li, Wentao and Wong, Raymond Chi-Wing and Wang, Wei },
  title         = { Elastic Index Select for Label-Hybrid Search in Vector Database },
  year          = 2025,
  eprint        = { 2505.03212 },
  archiveprefix = { arXiv },
  primaryclass  = { cs.DB },
}

@misc{curator,
	author        = { Yicheng Jin and Yongji Wu and Wenjun Hu and Bruce M. Maggs and Jun Yang and Xiao Zhang and Danyang Zhuo },
	title         = { Curator: Efficient Vector Search with Low-Selectivity Filters },
	year          = 2026,
	eprint        = { 2601.01291 },
	archiveprefix = { arXiv },
	primaryclass  = { cs.DB }
}

@misc{siglip,
  author       = { Michael Tschannen and Alexey Gritsenko and Xiao Wang and Muhammad Ferjad Naeem and Ibrahim Alabdulmohsin and Nikhil Parthasarathy and Talfan Evans and Lucas Beyer and Ye Xia and Basil Mustafa and Olivier Hénaff and Jeremiah Harmsen and Andreas Steiner and Xiaohua Zhai },
  title        = { SigLIP 2: Multilingual Vision-Language Encoders with Improved Semantic Understanding, Localization, and Dense Features },
  year         = 2025,
  eprint       = { 2502.14786 },
  archiveprefix = { arXiv },
  primaryclass = { cs.CV }
}

@misc{ada-ivf,
  author       = { Jason Mohoney and Anil Pacaci and Shihabur Rahman Chowdhury and Umar Farooq Minhas and Jeffery Pound and Cedric Renggli and Nima Reyhani and Ihab F. Ilyas and Theodoros Rekatsinas and Shivaram Venkataraman },
  title        = { Incremental IVF Index Maintenance for Streaming Vector Search },
  year         = 2024,
  eprint       = { 2411.00970 },
  archiveprefix = { arXiv },
  primaryclass = { cs.DB }
}

@misc{rag-survey1,
  author       = { Yunfan Gao and Yun Xiong and Xinyu Gao and Kangxiang Jia and Jinliu Pan and Yuxi Bi and Yi Dai and Jiawei Sun and Meng Wang and Haofen Wang },
  title        = { Retrieval-Augmented Generation for Large Language Models: A Survey },
  year         = 2024,
  eprint       = { 2312.10997 },
  archiveprefix = { arXiv },
  primaryclass = { cs.CL }
}

@misc{rag-survey2,
  author       = { Wenqi Fan and Yujuan Ding and Liangbo Ning and Shijie Wang and Hengyun Li and Dawei Yin and Tat-Seng Chua and Qing Li },
  title        = { A Survey on RAG Meeting LLMs: Towards Retrieval-Augmented Large Language Models },
  year         = 2024,
  eprint       = { 2405.06211 },
  archiveprefix = { arXiv },
  primaryclass = { cs.CL }
}

@misc{dpr,
  author       = { Vladimir Karpukhin and Barlas Oğuz and Sewon Min and Patrick Lewis and Ledell Wu and Sergey Edunov and Danqi Chen and Wen-tau Yih },
  title        = { Dense Passage Retrieval for Open-Domain Question Answering },
  year         = 2020,
  eprint       = { 2004.04906 },
  archiveprefix = { arXiv },
  primaryclass = { cs.CL }
}

@misc{cng,
  author       = { Binhong Li and Xiao Yan and Shangqi Lu },
  title        = { Fast-Convergent Proximity Graphs for Approximate Nearest Neighbor Search },
  year         = 2026,
  eprint       = { 2510.05975 },
  archiveprefix = { arXiv },
  primaryclass = { cs.DS }
}

@misc{bench-agnostic,
  author       = { Patrick Iff and Paul Bruegger and Marcin Chrapek and David Kochergin and Maciej Besta and Torsten Hoefler },
  title        = { Benchmarking Filtered Approximate Nearest Neighbor Search Algorithms on Transformer-based Embedding Vectors },
  year         = 2026,
  eprint       = { 2507.21989 },
  archiveprefix = { arXiv },
  primaryclass = { cs.DB }
}

@misc{infinity,
  author       = { Mengzhao Wang and Boyu Tan and Yunjun Gao and Hai Jin and Yingfeng Zhang and Xiangyu Ke and Xiaoliang Xu and Yifan Zhu },
  title        = { Balancing the Blend: An Experimental Analysis of Trade-offs in Hybrid Search },
  year         = 2025,
  eprint       = { 2508.01405 },
  archiveprefix = { arXiv },
  primaryclass = { cs.DB }
}

@misc{minilm,
  author       = { Wenhui Wang and Furu Wei and Li Dong and Hangbo Bao and Nan Yang and Ming Zhou },
  title        = { MiniLM: Deep Self-Attention Distillation for Task-Agnostic Compression of Pre-Trained Transformers },
  year         = 2020,
  eprint       = { 2002.10957 },
  archiveprefix = { arXiv },
  primaryclass = { cs.CL }
}

@misc{siftgisturl,
  author       = { Laurent Amsaleg and Hervé Jégou },
  title        = { Evaluation of Approximate Nearest Neighbors: Large Datasets },
  year         = { 2026 },
  howpublished = { \url{http://corpus-texmex.irisa.fr/} }
}

@misc{tripclickurl,
  author       = { Navid Rekabsaz and Oleg Lesota and Markus Schedl and Jon Brassey and Carsten Eickhoff },
  title        = { TripClick: The Log Files of a Large Health Web Search Engine },
  year         = { 2026 },
  howpublished = { \url{https://tripdatabase.github.io/tripclick/} }
}

@misc{laion400murl,
  author       = { Christoph Schuhmann },
  title        = { LAION-400-Million Open Dataset },
  year         = { 2021 },
  howpublished = { \url{https://laion.ai/blog/laion-400-open-dataset/} }
}

@misc{youtube8murl,
  author       = { Google Research },
  title        = { YouTube-8M: A Large and Diverse Labeled Video Dataset for Video Understanding Research },
  year         = { 2026 },
  howpublished = {\url{https://research.google.com/youtube8m/index.html}}
}

@misc{wikieurl,
  author       = { Huggingface },
  title        = { Wikipedia-22-12-en Embeddings with all-MiniLM-L6-v2 },
  year         = { 2026 },
  howpublished = { \url{https://huggingface.co/datasets/maloyan/wikipedia-22-12-en-embeddings-all-MiniLM-L6-v2} }
}

@misc{yfccurl,
  author       = { BigANN Organizers },
  title        = { NeurIPS'23 Competition Track: Big-ANN },
  year         = { 2026 },
  howpublished = { \url{https://big-ann-benchmarks.com/neurips23.html} }
}

@misc{pgvector,
  author       = { Andrew Kane },
  title        = { pgvector: Open-source vector similarity search for Postgres },
  year         = { 2026 },
  howpublished = { \url{https://github.com/pgvector/pgvector} }
}

@article {nsg,
  author       = { Cong Fu and Chao Xiang and Changxu Wang and Deng Cai },
  title        = { Fast Approximate Nearest Neighbor Search with the Navigating Spreading-out Graph },
  journal      = { Proc. VLDB Endow. },
  year         = 2019,
  pages        = { 461--474 },
  volume       = 12,
  number       = 5
}

@article {nssg,
  author       = { Cong Fu and Changxu Wang and Deng Cai },
  title        = { High Dimensional Similarity Search With Satellite System Graph: Efficiency, Scalability, and Unindexed Query Compatibility },
  journal      = { IEEE Trans. Pattern Anal. Mach. Intell. },
  year         = 2022,
  pages        = { 4139--4150 },
  volume       = 44,
  number       = 8
}

@article{tau-mg,
  author       = { Yun Peng and Byron Choi and Tsz Nam Chan and Jianye Yang and Jianliang Xu },
  title        = { Efficient Approximate Nearest Neighbor Search in Multi-dimensional Databases },
  journal      = { Proc. ACM Manag. Data },
  year         = 2023,
  volume       = 1,
  number       = 1,
  numpages     = 27
}

@article {graph-survey,
  author       = { Mengzhao Wang and Xiaoliang Xu and Qiang Yue and Yuxiang Wang },
  title        = { A Comprehensive Survey and Experimental Comparison of Graph-Based Approximate Nearest Neighbor Search },
  journal      = { Proc. VLDB Endow. },
  year         = 2021,
  pages        = { 1964--1978 },
  volume       = 14,
  number       = 11
}

@article {acorn,
  author       = { Liana Patel and Peter Kraft and Carlos Guestrin and Matei Zaharia },
  title        = { ACORN: Performant and Predicate-Agnostic Search over Vector Embeddings and Structured Data },
  journal      = { Proc. ACM Manag. Data },
  year         = 2024,
  volume       = 2,
  number       = 3,
  numpages     = 27
}

@article {ung,
  author       = { Yuzheng Cai and Jiayang Shi and Yizhuo Chen and Weiguo Zheng },
  title        = { Navigating Labels and Vectors: A Unified Approach to Filtered Approximate Nearest Neighbor Search },
  journal      = { Proc. ACM Manag. Data },
  year         = 2024,
  volume       = 2,
  number       = 6,
  numpages     = 27
}

@article{unifilter,
  author       = { Xie, Jiadong and Yu, Jeffrey Xu and Teng, Siyi and Liu, Yingfan },
  title        = { Beyond Vector Search: Querying With and Without Predicates },
  journal      = { Proc. ACM Manag. Data },
  year         = 2025,
  volume       = 3,
  number       = 6,
  numpages     = 26
}

@article{rwalks,
  author       = { Ait Aomar, Anas and Echihabi, Karima and Arnaboldi, Marco and Alagiannis, Ioannis and Hilloulin, Damien and Cherkaoui, Manal },
  title        = { RWalks: Random Walks as Attribute Diffusers for Filtered Vector Search },
  journal      = { Proc. ACM Manag. Data },
  year         = 2025,
  volume       = 3,
  number       = 3,
  numpages     = 26
}

@article{vsag,
  author       = { Zhong, Xiaoyao and Li, Haotian and Jin, Jiabao and Yang, Mingyu and Chu, Deming and Wang, Xiangyu and Shen, Zhitao and Jia, Wei and Gu, George and Xie, Yi and Lin, Xuemin and Shen, Heng Tao and Song, Jingkuan and Cheng, Peng },
  title        = { VSAG: An Optimized Search Framework for Graph-Based Approximate Nearest Neighbor Search },
  journal      = { Proc. VLDB Endow. },
  year         = 2025,
  volume       = 18,
  number       = 12,
  numpages     = 14
}

@article{faiss,
  author       = { Johnson, Jeff and Douze, Matthijs and J\'{e}gou, Herv\'{e} },
  title        = { Billion-Scale Similarity Search with GPUs },
  journal      = { IEEE Trans. Big Data },
  year         = 2021,
  volume       = 7,
  number       = 3,
  pages        = { 535--547 }
}

@article{wow,
  author       = { Wang, Ziqi and Zhang, Jingzhe and Hu, Wei },
  title        = { WoW: A Window-to-Window Incremental Index for Range-Filtering Approximate Nearest Neighbor Search },
  journal      = { Proc. ACM Manag. Data },
  year         = 2025,
  volume       = 3,
  number       = 6,
  numpages     = 27
}

@article{filterbench_sigmod,
  author       = { Li, Mocheng and Yan, Xiao and Lu, Baotong and Zhang, Yue and Cheng, James and Ma, Chenhao },
  title        = { Attribute Filtering in Approximate Nearest Neighbor Search: An In-depth Experimental Study },
  journal      = { Proc. ACM Manag. Data },
  year         = 2025,
  volume       = 3,
  number       = 6,
  numpages     = 26
}

@article {hnsw,
  author       = { Yu A. Malkov and D. A. Yashunin },
  title        = { Efficient and Robust Approximate Nearest Neighbor Search Using Hierarchical Navigable Small World Graphs },
  journal      = { IEEE Trans. Pattern Anal. Mach. Intell. },
  year         = 2020,
  pages        = { 824--836 },
  volume       = 42,
  number       = 4
}

@article{hscg,
	author       = { Qiu, Runwen and Tang, Jing },
	title        = { Efficient Approximate Nearest Neighbor Search via Hemi-Sphere Centroids Graph },
	journal      = { Proc. ACM Manag. Data },
	year         = 2025,
	volume       = 3,
	number       = 6,
	numpages     = 26
}

@article {psp,
  author       ={ Chen, Tingyang and Fu, Cong and Wang, Kun and Ke, Xiangyu and Gao, Yunjun and Zhou, Wenchao and Ni, Yabo and Zeng, Anxiang },
  title        ={ Maximum Inner Product is Query-Scaled Nearest Neighbor },
  journal      = { Proc. VLDB Endow. },
  year         = 2025,
  pages        = { 1770--1783 },
  volume       = 18,
  number       = 6
}

@article{zipfreview,
  author       = { Piantadosi, Steven T. },
  title        = { Zipf's Word Frequency Law in Natural Language: A Critical Review and Future Directions },
  journal      = { Psychonomic Bulletin \& Review },
  year         = 2014,
  pages        = { 1112--1130 },
  volume       = 21,
  number       = 5
}

@article{db-lsh2,
  author       = { Tian, Yao and Zhao, Xi and Zhou, Xiaofang },
  title        = { DB-LSH 2.0: Locality-Sensitive Hashing With Query-Based Dynamic Bucketing },
  journal      = { IEEE Trans. Knowl. Data Eng. },
  year         = 2024,
  volume       = 36,
  number       = 3,
  pages        = { 1000--1015 }
}

@article{det-lsh,
  author       = { Wei, Jiuqi and Peng, Botao and Lee, Xiaodong and Palpanas, Themis },
  title        = { DET-LSH: A Locality-Sensitive Hashing Scheme with Dynamic Encoding Tree for Approximate Nearest Neighbor Search },
  journal      = { Proc. VLDB Endow. },
  year         = 2024,
  volume       = 17,
  number       = 9,
  numpages     = 14,
  pages        = { 2241--2254 }
}

@article {tree-based,
  author       = { Marius Muja and David G. Lowe },
  title        = { Scalable Nearest Neighbor Algorithms for High Dimensional Data },
  journal      = { IEEE Trans. Pattern Anal. Mach. Intell. },
  year         = 2014,
  pages        = { 2227--2240 },
  volume       = 36,
  number       = 11
}

@article{cracking-ivf,
  author       = { Mageirakos, Vasilis and Wu, Bowen and Alonso, Gustavo },
  title        = { Cracking Vector Search Indexes },
  journal      = { Proc. VLDB Endow. },
  year         = 2025,
  pages        = { 3951--3964 },
  volume       = 18,
  number       = 11
}

@article{pq,
  author       = { Jégou, Herve and Douze, Matthijs and Schmid, Cordelia },
  title        = { Product Quantization for Nearest Neighbor Search },
  journal      = { IEEE Trans. Pattern Anal. Mach. Intell. },
  year         = 2011,
  pages        = { 117--128 },
  volume       = 33,
  number       = 1
}

@article{almg,
  author       = { Xie, Jiadong and Yu, Jeffrey Xu and Liu, Yingfan },
  title        = { Graph Based K-Nearest Neighbor Search Revisited },
  journal      = { ACM Trans. Database Syst. },
  year         = 2025,
  numpages     = 30,
  volume       = 50,
  number       = 4
}

@article {hqi,
  author       = { Jason Mohoney and Anil Pacaci and Shihabur Rahman Chowdhury and Ali Mousavi and Ihab F. Ilyas and Umar Farooq Minhas and Jeffrey Pound and Theodoros Rekatsinas },
  title        = { High-Throughput Vector Similarity Search in Knowledge Graphs },
  journal      = { Proc. ACM Manag. Data },
  year         = 2023,
  volume       = 1,
  number       = 2,
  numpages     = 25
}

@article{sieve,
  author       = { Li, Zhaoheng and Huang, Silu and Ding, Wei and Park, Yongjoo and Chen, Jianjun },
  title        = { SIEVE: Effective Filtered Vector Search with Collection of Indexes },
  journal      = { Proc. VLDB Endow. },
  year         = 2025,
  volume       = 18,
  number       = 11,
  numpages     = 14,
  pages        = { 4723–-4736 }
}

@article {adbv-vdb,
  author       = { Chuangxian Wei and Bin Wu and Sheng Wang and Renjie Lou and Chaoqun Zhan and Feifei Li and Yuanzhe Cai },
  title        = { AnalyticDB-V: A Hybrid Analytical Engine Towards Query Fusion for Structured and Unstructured Data },
  journal      = { Proc. VLDB Endow. },
  year         = 2020,
  pages        = { 3152--3165 },
  volume       = 13,
  number       = 12
}

@article{navix,
  author       = { Sehgal, Gaurav and Saliho\u{g}lu, Semih },
  title        = { NaviX: A Native Vector Index Design for Graph DBMSs With Robust Predicate-Agnostic Search Performance },
  journal      = { Proc. VLDB Endow. },
  year         = 2025,
  volume       = 18,
  number       = 11,
  numpages     = 13,
  pages        = { 4438–4450 }
}

@article {serf,
  author       = { Chaoji Zuo and Miao Qiao and Wenchao Zhou and Feifei Li and Dong Deng },
  title        = { SeRF: Segment Graph for Range-Filtering Approximate Nearest Neighbor Search },
  journal      = { Proc. ACM Manag. Data },
  year         = 2024,
  volume       = 2,
  number       = 1,
  numpages     = 26
}

@article {irange,
  author       = { Yuexuan Xu and Jianyang Gao and Yutong Gou and Cheng Long and Christian S. Jensen },
  title        = { iRangeGraph: Improvising Range-dedicated Graphs for Range-filtering Nearest Neighbor Search },
  journal      = { Proc. ACM Manag. Data },
  year         = 2024,
  volume       = 2,
  number       = 6,
  numpages     = 26
}

@article {vdb-survey,
  author       = { James Jie Pan and Jianguo Wang and Guoliang Li },
  title        = { Survey of Vector Database Management Systems },
  journal      = { VLDB J. },
  year         = 2024,
  pages        = { 1591--1615 },
  volume       = 33
}

@article{exrbq,
  author       = { Gao, Jianyang and Gou, Yutong and Xu, Yuexuan and Yang, Yongyi and Long, Cheng and Wong, Raymond Chi-Wing },
  title        = { Practical and Asymptotically Optimal Quantization of High-Dimensional Vectors in Euclidean Space for Approximate Nearest Neighbor Search },
  journal      = { Proc. ACM Manag. Data },
  year         = 2025,
  volume       = 3,
  number       = 3,
  numpages     = 26
}

@article{rbq,
  author       = { Gao, Jianyang and Long, Cheng },
  title        = { RaBitQ: Quantizing High-Dimensional Vectors with a Theoretical Error Bound for Approximate Nearest Neighbor Search },
  journal      = { Proc. ACM Manag. Data },
  year         = 2024,
  volume       = 2,
  number       = 3,
  numpages     = 27
}

@article{dsg,
  author       = { Peng, Zhencan and Qiao, Miao and Zhou, Wenchao and Li, Feifei and Deng, Dong },
  title        = { Dynamic Range-Filtering Approximate Nearest Neighbor Search },
  journal      = { Proc. VLDB Endow. },
  year         = 2025,
  volume       = 18,
  number       = 10,
  numpages     = 13
}

@article{digra,
  author       = { Jiang, Mengxu and Yang, Zhi and Zhang, Fangyuan and Hou, Guanhao and Shi, Jieming and Zhou, Wenchao and Li, Feifei and Wang, Sibo },
  title        = { DIGRA: A Dynamic Graph Indexing for Approximate Nearest Neighbor Search with Range Filter },
  journal      = { Proc. ACM Manag. Data },
  year         = { 2025 },
  volume       = { 3 },
  number       = { 3 },
  numpages     = { 26 },
}

@article{eg-imgcontent,
  author       = { Akg{\"u}l, Ceyhun Burak and Rubin, Daniel L. and Napel, Sandy and Beaulieu, Christopher F. and Greenspan, Hayit and Acar, Burak },
  title        = { Content-Based Image Retrieval in Radiology: Current Status and Future Directions },
  journal      = { Journal of Digital Imaging },
  year         = 2011,
  volume       = 24,
  number       = 2,
  pages        = { 208--222 }
}

@article{eg-imgretrieval,
  author       = { Hwang, Kyung Hoon and Lee, Haejun and Koh, Geon and Willrett, Debra and Rubin, Daniel L. },
  title        = { Building and Querying {RDF}/{OWL} Database of Semantically Annotated Nuclear Medicine Images },
  journal      = { Journal of Digital Imaging },
  year         = 2017,
  volume       = 30,
  number       = 1,
  pages        = { 4--10 }
}

@book{zipfbook,
  author       = { Zipf, George Kingsley },
  title        = { Human Behavior and the Principle of Least Effort: An Introduction to Human Ecology },
  year         = 2016,
  publisher    = { Addison-Wesley Publishing },
  address      = { Boston, MA, USA }
}

\appendix


\section{Proof of Theorem~\ref{theorem:minhash}}
\label{app:proof_th31}
\begin{proof}
    By the MinHash property, for any random function $h_i$, $\Pr\big(h_i^{\min}(L_1)=h_i^{\min}(L_2)\big)=J(L_1,L_2)\ge1-\theta_t$.
    Since the hash functions are independent, the probability that $L_1,L_2$ match in all $\frac{\tau}{\beta}$ hash values of a single band is at least $(1-\theta_t)^{\frac{\tau}{\beta}}$, and the prabability that they do not match in any of the $\beta$ band is at most $\big(1-(1-\theta_t)^\frac{\tau}{\beta}\big)^\beta$.
    Taking the complementary event gives us the probability that $L_1,L_2$ match in at least one band in Eq.~\eqref{eq:p_minhash_colli}.
\end{proof}

\section{Proofs for Containment Lower Bound}
\label{app:proof_containment}

We first derive the containment kernel for fixed $L_1$ and $L_2$.
Without the Jaccard similarity constraint, the containment probability conditioned on fixed $L_1$ and $L_2$ is given by Lemma~\ref{lemma:iid_prod}.

\begin{lemma}\label{lemma:iid_prod}
For any fixed non-empty $L_1,L_2\in\mathcal{L}$,
\begin{align}
    \Pr(L_q\subseteq L_2\,|\, L_q\subseteq L_1,L_1,L_2)
    =
    \prod_{a_i\in L_1\backslash L_2}(1-p_i)
    :=
    K(L_1,L_2).
\end{align}
\end{lemma}

\begin{proof}
By definition of conditional probability:
    \begin{align} \label{eq:lqsubl2}
        \Pr(L_q\subseteq L_2 \,|\, L_q \subseteq L_1, L_1, L_2)=\frac{\Pr(L_q \subseteq L_1\cap L_2 \,|\, L_1,L_2)}{\Pr(L_q \subseteq L_1 \,|\, L_1,L_2)}.
    \end{align}
    
    As for two arbitrary sets $S$ and $T$, $\Pr(S \subseteq T)=\prod_{a_i\notin T}\Pr(a_i\notin S)=\prod_{a_i\notin T}(1-p_i)$, we have
    \begin{align}
        \Pr(L_q\subseteq L_2 \,|\, L_q \subseteq L_1, L_1, L_2)=\frac{\prod_{a_i\notin (L_1\cap L_2)}(1-p_i)}{\prod_{a_i \notin L_1}(1-p_i)}.
    \end{align}
    
    With relation between sets $L_1 \backslash (L_1\cap L_2)=L_1\backslash L_2$, $\Pr(L_q\subseteq L_2 \,|\, L_q \subseteq L_1, L_1, L_2)=\prod_{a_i\in L_1\backslash L_2}(1-p_i)$, which is defined as $K(L_1,L_2)$.
\end{proof}

Lemma~\ref{lemma:iid_prod} shows that, given any label sets $L_1$ and $L_2$, the containment target fails only when $L_q$ contains a label in $L_1$ that does not belong to $L_2$.
Thus, every label in $L_1\backslash L_2$ must be absent from $L_q$, giving the product term $K(L_1,L_2)$.

The term $K(L_1,L_2)$ does not yet incorporate the similarity constraint between $L_1$ and $L_2$.
Under $J(L_1,L_2)\ge\epsilon$, the size of the difference set $L_1\backslash L_2$ is controlled by Lemma~\ref{lemma:jac_cond}.

\begin{lemma}\label{lemma:jac_cond}
If $J(L_1,L_2)\ge\epsilon$, then $|L_1 \backslash L_2|\le(1-\epsilon)|L_1|$.
\end{lemma}

\begin{proof}
As $|L_1\cup L_2|\ge |L_1|$, $J(L_1,L_2)=\frac{|L_1\cap L_2|}{|L_1\cup L_2|}\ge \epsilon$, we have $|L_1\cap L_2| \ge \epsilon|L_1|$, and $|L_1\backslash L_2|=|L_1|-|L_1\cap L_2|\le (1-\epsilon)|L_1|$.
\end{proof}

Lemma~\ref{lemma:jac_cond} formalizes the basic effect of label set similarity: a larger Jaccard threshold leaves fewer labels in $L_1$ that may be absent from $L_2$.
This is the combinatorial reason why stricter tiers improve containment consistency.

\subsection{Proof of Theorem~\ref{theorem:pc_gen}}
\label{app:proof_th41}

\begin{proof}
Combining Lemma~\ref{lemma:iid_prod} and Lemma~\ref{lemma:jac_cond}, we have
\begin{align} \label{eq:Hlower}
    K(L_1,L_2)\geq \prod_{a_i\in L_1\backslash L_2}(1-p_{\max}) \geq (1-p_{\max})^{(1-\epsilon)|L_1|}.
\end{align}

Since the point-wise lower bound in Eq.~\eqref{eq:Hlower} holds for all pairs $L_1,L_2$, we can take expectation w.r.t. the conditional distribution of $L_1,L_2$ given $L_q\subseteq L_1$ and $J(L_1,L_2) \ge \epsilon$ as follows:
\begin{align}
    P_c =\mathbb{E}_{L_1,L_2}\big[K(L_1,L_2) \,|\, J(L_1,L_2)\ge \epsilon\big].
\end{align}

Replace $|L_1\backslash L_2|$ with Lemma~\ref{lemma:jac_cond} and use the maximum probability among all labels, taking the fact that $|L_1|\le |A|$:
\begin{align}
    P_c\ge \mathbb{E}\left[(1-p_{\max})^{(1-\epsilon)|L_1|} \,|\, J(L_1,L_2)\ge \epsilon \right] \ge (1-p_{\max})^{(1-\epsilon)|A|}.
\end{align}

By Chebyshev's inequality, $\Pr\left(|L_1|\le \mu+\frac{\sigma}{\sqrt{\xi}}\right)\ge 1-\xi$, replacing $|L_1|$ in Eq.~\eqref{eq:Hlower} gives us a high-confidence bound in Eq.~\eqref{eq:prob_pc}, where $\mu \ll |A|$ in real-world datasets.
\end{proof}

We next characterize the skewed setting.
Let $\mathcal{J}: J(L_1,L_2)\ge\epsilon$ and $D=L_1\backslash L_2$.
The following lemma bounds the probability that a specific label appears in the difference set under the high-Jaccard constraint.

\begin{lemma} \label{lemma:highinD}
    Let event $\mathcal{J}: J(L_1,L_2) \ge \epsilon$, $D=L_1\backslash L_2$. $\forall a_i \in A, \Pr(a_i\in D \,|\,\mathcal{J})\le (1-\epsilon)(1-p_i)p_i$.
\end{lemma}

\begin{proof}
$\Pr(a_i\in D\,|\, \mathcal{J})=\Pr(a_i\in L_1 \,|\, \mathcal{J})\Pr(a_i \notin L_2 \,|\, a_i \in L_1,\mathcal{J})$.
Given $a_i \in L_1$ and $\mathcal{J}$ with Lemma~\ref{lemma:jac_cond}, $|L_1\backslash L_2| \le (1-\epsilon) |L_1|$ means that at most $(1-\epsilon)$ fraction of elements in $L_1$ can be missing from $L_2$ with probability $1-p_i$, i.e. $\Pr(a_i \notin L_2 \,|\, a_i \in L_1,\mathcal{J})\le (1-\epsilon)(1-p_i)$:
\begin{align}
\begin{split}
    \Pr(a_i\in D\,|\, \mathcal{J})&=\Pr(a_i \in L_1\,|\, \mathcal{J})\Pr(a_i\notin L_2 \,|\, a_i \in L_1, \mathcal{J}) \\
    &\le (1-\epsilon)(1-p_i)p_i.
\end{split}
\end{align}
\end{proof}

Lemma~\ref{lemma:highinD} explains the intuition behind the Zipf-skewed bound.
Although head labels have large inclusion probabilities, they are less likely to dominate $D$ under a high-Jaccard constraint.
Intuitively, if a frequent label appears in $L_1$, omitting it from $L_2$ consumes part of the limited missing-label budget imposed by $J(L_1,L_2)\ge\epsilon$.
Therefore, under skewed marginals, the difference set is more likely to be formed by lower-probability labels, making $K(L_1,L_2)=\prod_{a_i\in D}(1-p_i)$ larger than what the worst-case $p_{\max}$ bound suggests.

\subsection{Proof of Theorem~\ref{theorem:pc_skew}}
\label{app:proof_th42}

\begin{proof}
Let $p_i := \frac{\mu i^{-s}}{H_{|A|,s}}, H_{|A|,s} := \sum_{j=1}^{|A|} j^{-s}, s \ge 0$ and define $r := \max(2, \arg\min_i \ \frac{H_{i,s}}{H_{|A|,s}} \ge \rho )$.
Hence, the prefix $\{a_1,\ldots,a_r\}$ contains at least a $\rho$-fraction of the Zipf mass.  
As in the theorem statement, let $\mathcal{J} := \{J(L_1,L_2)\ge \epsilon\}, \mathcal{A} := \{D \subseteq \{a_r,\ldots,a_{|A|}\}\}$.

\bparagraph{(1) Bound the expected tail mass.} For fixed $q\ge 0$ and $\forall x\ge 0$,
\begin{align}
\Psi_q(x):=
\begin{cases}
\dfrac{(r+x-1)^{1-q}-(r-1)^{1-q}}{1-q}, & q\neq 1,\\[1.2ex]
\ln\dfrac{r+x-1}{r-1}, & q=1.
\end{cases}
\end{align}

Since $t\mapsto t^{-q}$ is decreasing on $[1,\infty)$, for any integer $y\ge 1$,
\begin{align}
\sum_{i=r}^{r+y-1} i^{-q}
\le \int_{r-1}^{r+y-1} t^{-q}\,dt
= \Psi_q(y),
\end{align}
where the equality follows by direct evaluation of the integral (with the logarithmic form when $q=1$). Therefore, on $\mathcal{A}$,
\begin{align}
\sum_{a_i\in D} p_i
\le \frac{\mu}{H_{|A|,s}} \Psi_s(|D|),
\quad
\sum_{a_i\in D} p_i^2
\le \frac{\mu^2}{H_{|A|,s}^2} \Psi_{2s}(|D|).
\end{align}

Using
$|D|\le |L_1|+|L_2|-2|L_1\cap L_2|\le \frac{1-\epsilon}{1+\epsilon}\bigl(|L_1|+|L_2|\bigr)$, we get
\begin{align}
\mathbb{E}[|D|\,|\, \mathcal{J},\mathcal{A}]
\le \mathbb{E}[|D|\,|\, \mathcal{J}]
\le \frac{2\mu(1-\epsilon)}{1+\epsilon}
:= B.
\end{align}

Moreover, $\Psi_q$ is concave on $x\ge 0$ for $q\ge 0$ since $\Psi_q''(x)=-q(r+x-1)^{-q-1}\le 0$. Hence, by Jensen's inequality,
\begin{align}
\begin{split}
\mathbb{E}\left[\sum_{a_i\in D} p_i \,\middle|\, \mathcal{J},\mathcal{A}\right]
&\le \frac{\mu}{H_{|A|,s}}\,\Psi_s(B), \\
\mathbb{E}\left[\sum_{a_i\in D} p_i^2 \,\middle|\, \mathcal{J},\mathcal{A}\right]
&\le \frac{\mu^2}{H_{|A|,s}^2}\,\Psi_{2s}(B).
\end{split}
\end{align}

\bparagraph{(2) Bound the event kernel and control the complement.}
\begin{align}
    K(L_1,L_2)=\prod_{a_i\in D}(1-p_i) = \exp\left(\sum_{a_i\in D}\ln(1-p_i)\right).
\end{align}

Under the standard long-tail condition on $\mathcal{A}$ (i.e., assuming $p_i\le 0.6$ for all $i\ge r$), we use
$\ln(1-x)\ge -x-x^2$ for $x\in[0,0.6]$ and obtain
\begin{align}
\mathbb{E}[K(L_1,L_2)\,|\, \mathcal{J},\mathcal{A}]
\ge
\exp\left(
-\frac{\mu}{H_{|A|,s}}\Psi_s(B)
-\frac{\mu^2}{H_{|A|,s}^2}\Psi_{2s}(B)
\right).
\end{align}

For $\mathcal{A}^c$, by Lemma~\ref{lemma:highinD},
\begin{align}
\begin{split}
\Pr(\mathcal{A}^c\,|\, \mathcal{J})
&= \Pr\left(\bigcup_{i<r}\{a_i\in D\}\,\middle|\,\mathcal{J}\right) \\
&\le (1-\epsilon)\sum_{i<r} p_i(1-p_i) \le (1-\epsilon)\sum_{i<r} p_i \le \mu\rho(1-\epsilon),
\end{split}
\end{align}
thus $\Pr(\mathcal{A}\,|\, \mathcal{J}) \ge 1-\mu\rho(1-\epsilon)$.

\bparagraph{(3) Combine bounds by the Bayes' theorem.}
\begin{align}
\begin{split}
P_c
&= \mathbb{E}[K(L_1,L_2)\,|\, \mathcal{J}] \\
&\ge \mathbb{E}[K(L_1,L_2)\mathbf{1}(\mathcal{A})\,|\, \mathcal{J}] \ge \mathbb{E}[K(L_1,L_2)\,|\, \mathcal{J},\mathcal{A}] \cdot \Pr(\mathcal{A}\,|\, \mathcal{J}) \\
&\ge
\exp\left(
-\frac{\mu}{H_{|A|,s}}\Psi_s(B)
-\frac{\mu^2}{H_{|A|,s}^2}\Psi_{2s}(B)
\right)\bigl(1-\mu\rho(1-\epsilon)\bigr),
\end{split}
\end{align}
where $B=\frac{2\mu(1-\epsilon)}{1+\epsilon}$.
This establishes the claimed lower bound.
\end{proof}

\section{Proofs for Overlap Lower Bound}
\label{app:proof_overlap}

For overlap queries, we analyze the complementary event that a query overlaps $L_1$ but does not overlap $L_2$.
The following lemma shows that this non-overlap probability is maximized when the query label set is as small as possible.

\begin{lemma} \label{lemma:max_po_comp}
    $\Pr(L_q\cap L_2=\emptyset \,|\, L_q \cap L_1 \neq \emptyset)$ is maximized at $|L_q|=1$.
\end{lemma}

\begin{proof}
Define $L_A=L_1\cap L_2, L_B=L_1\backslash L_2, L_C=L_2\backslash L_1,L_D=(L_1\cup L_2)^c$, and  $a=|L_A|, b=|L_B|, c=|L_C|, d=|L_D|$.
Then, $|L_1|=a+b, |L_2|=a+c, |L_1\cup L_2|=|a+b+c|$.
Define $\kappa=|L_q|$, 
\begin{align}
\resizebox{\columnwidth}{!}{$
    f(\kappa)=\Pr(L_q\cap L_2=\emptyset \,|\, L_q \cap L_1 \neq \emptyset)=\frac{\Pr(L_q\cap L_2=\emptyset \,\wedge\, L_q\cap L_1 \neq \emptyset)}{\Pr(L_q\cap L_1\neq \emptyset)}.
$}
\end{align}

The numerator of $f(\kappa)$ is the probability of a joint event.
$L_q \cap L_2=\emptyset$ iff. $L_q$ cannot contain any elements from $L_2=L_A\cup L_C$, and thus $L_q$ must be contained in $(L_A\cup L_C)^c=L_B\cup L_D$.
Given $L_q\subseteq L_B\cup L_D$, since $L_D\cap L_1 =\emptyset$, $L_q\cap L_1 \neq \emptyset$ means that there is at least one element from $B$.
Thus, $(L_q\subseteq L_B\cup L_D) \,\wedge\, (L_q\cap L_B\neq \emptyset)$.
The number of $\kappa$-element subsets satisfying the above condition is $\binom{b+d}{\kappa}-\binom{d}{\kappa}$.
The dominator of $f(\kappa)$ equals to $1-\Pr(L_q\cap L_1=\emptyset)$, and for $L_q\cap L_1=\emptyset$, $L_q$ must be contained by $(L_A\cup L_B)^c=L_C\cup L_D$.
$f(\kappa)=\frac{\binom{b+d}{\kappa}-\binom{d}{\kappa}}{\binom{|A|}{\kappa}-\binom{c+d}{\kappa}}$, which is non-increasing in $\kappa$ because $\frac{f(\kappa+1)}{f(\kappa)} <1$.
Therefore, $f(\kappa) \le f(1)=\frac{b}{a+b}$.
\end{proof}

Lemma~\ref{lemma:max_po_comp} reduces the worst case of overlap preservation to singleton query labels.
Thus, to lower-bound $P_o$, it suffices to minimize the probability that a singleton label overlapping $L_1$ also belongs to $L_2$ under the Jaccard constraint.

\subsection{Proof of Theorem~\ref{theorem:po}}
\label{app:proof_th43}

\begin{proof}
We consider the probability from the complement without the Jaccard constraint:
\begin{align}
    P_o' = 1-\Pr(L_q\cap L_2 = \emptyset \,|\, L_q\cap L_1 \neq \emptyset).
\end{align}

By Lemma~\ref{lemma:max_po_comp}, $P_o'$ is minimized at $|L_q|=1$, thus the lower bound of $P_o$ is to minimize $P_o'$ subject to $\frac{|L_1 \cap L_2|}{|L_1\cup L_2|} \ge \epsilon$.
Using the same notation in the proof of Lemma~\ref{lemma:max_po_comp}, the object is to find the lower bound of $\frac{a}{a+b}$ subject to $\frac{a}{a+b+c} \ge \epsilon$.
$\frac{a}{a+b+c}\ge \epsilon \Leftrightarrow a\ge \frac{\epsilon}{1-\epsilon}(b+c)$:
\begin{align}
    \frac{a}{a+b}\ge\frac{\frac{\epsilon}{1-\epsilon}(b+c)}{\frac{\epsilon}{1-\epsilon}(b+c)+b}=\frac{\epsilon(b+c)}{b+\epsilon c}=:f(c).
\end{align}

Taking the derivative, $\frac{df}{dc}=\frac{\epsilon b(1-\epsilon)}{(b+\epsilon c)^2} \ge0$.
So, the worst case lower bound is taken at $c=0$, i.e. $L_2 \subseteq L_1$, which gives us $\frac{a}{a+b}\ge \epsilon$.
\end{proof}

\section{Dataset Details}
\label{app:data_detail}
\begin{table}
\centering
\caption{Query-set workload statistics}
\label{tab:workload-query}
\small
\setlength{\tabcolsep}{3pt}
\begin{tabular}{llrrccr}
\toprule
\multirow{2}{*}{\textbf{Datasets}} &
\multirow{2}{*}{\textbf{Semantics}} &
\multirow{2}{*}{\textbf{Labels}} &
\multirow{2}{*}{\textbf{Sets}} &
\multicolumn{3}{c}{\textbf{Label set size}} \\
\cmidrule{5-7} & & & & \textbf{Avg} & \textbf{Med} & \textbf{P90} \\
\midrule
\multirow{3}{*}{SIFT}
  & Equality    & 12 & 588 & 2.3 & 2 & 4 \\
  & Containment & 12 & 611 & 2.3 & 2 & 4 \\
  & Overlap     & 12 & 622 & 2.3 & 2 & 4 \\

\multirow{3}{*}{GIST}
  & Equality    & 12 & 589 & 2.3 & 2 & 4 \\
  & Containment & 12 & 611 & 2.3 & 2 & 4 \\
  & Overlap     & 12 & 622 & 2.3 & 2 & 4 \\

\multirow{3}{*}{TripClick}
  & Equality    & 29 & 249 & 2.0 & 2 & 3 \\
  & Containment & 29 & 305 & 2.1 & 2 & 3 \\
  & Overlap     & 29 & 503 & 2.7 & 3 & 4 \\

\multirow{3}{*}{LAION}
  & Equality    & 30 & 491 & 2.5 & 2 & 4 \\
  & Containment & 30 & 504 & 2.7 & 3 & 4 \\
  & Overlap     & 30 & 524 & 2.8 & 3 & 5 \\

\multirow{3}{*}{YTB-Video}
  & Equality    & 135 & 120 & 2.2 & 2 & 4 \\
  & Containment & 177 & 143 & 2.2 & 2 & 4 \\
  & Overlap     & 279 & 174 & 3.1 & 3 & 5 \\

\multirow{3}{*}{YTB-Audio}
  & Equality    & 675 & 652 & 2.7 & 2 & 5 \\
  & Containment & 731 & 680 & 2.7 & 2 & 5 \\
  & Overlap     & 820 & 708 & 3.1 & 3 & 5 \\

\multirow{3}{*}{Wikipedia}
  & Equality    & 32,780 & 9,673 & 5.8 & 6 & 7 \\
  & Containment & 1,520 & 2,689 & 1.7 & 2 & 3 \\
  & Overlap     & 5,960 & 9,756 & 5.9 & 6 & 9 \\

\multirow{3}{*}{YFCC-1M}
  & Equality    & 186 & 450 & 2.9 & 3 & 3 \\
  & Containment & 261 & 502 & 3.1 & 3 & 3 \\
  & Overlap     & 4753 & 939 & 9.0 & 6 & 16 \\
\bottomrule
\end{tabular}
\end{table}

\bparagraph{Label--vector correlation.}
We compute Spearman's $\rho$ between label Jaccard distance and vector L2 distance for all datasets.
Base vectors are ranked by their label set size $|L|$ and split into four equal-count quartile buckets: Q1 contains the smallest 25\% of label sets and Q4 the largest 25\%.
Fig.~\ref{fig:corr-vec} shows negligible correlations on seven datasets and a modest positive correlation on Wikipedia.
Thus, the evaluated workloads generally lack strong global label--vector correlation, which LSSG does not assume.
Instead, label stratification improves local in-filtering navigation by increasing the likelihood of label-qualified next-hop candidates.

\bparagraph{Query workload statistics.}
Table~\ref{tab:workload-query} reports query-side effective labels, distinct query label sets, and query label set size distributions for all filter semantics.
Most equality and containment queries are compact, with median size 2--3 and P90 at most 7.
Overlap queries are broader.
YFCC-1M overlap uses 4,753 effective query labels with average/median/P90 size 9.0/6/16, while Wikipedia overlap uses 5,960 effective query labels with size 5.9/6/9.
These statistics complement Table~\ref{tab:data} and show that the workloads cover both compact query labels and large, diverse label spaces.

\begin{figure*}
    \centering
    \includegraphics[width=\linewidth]{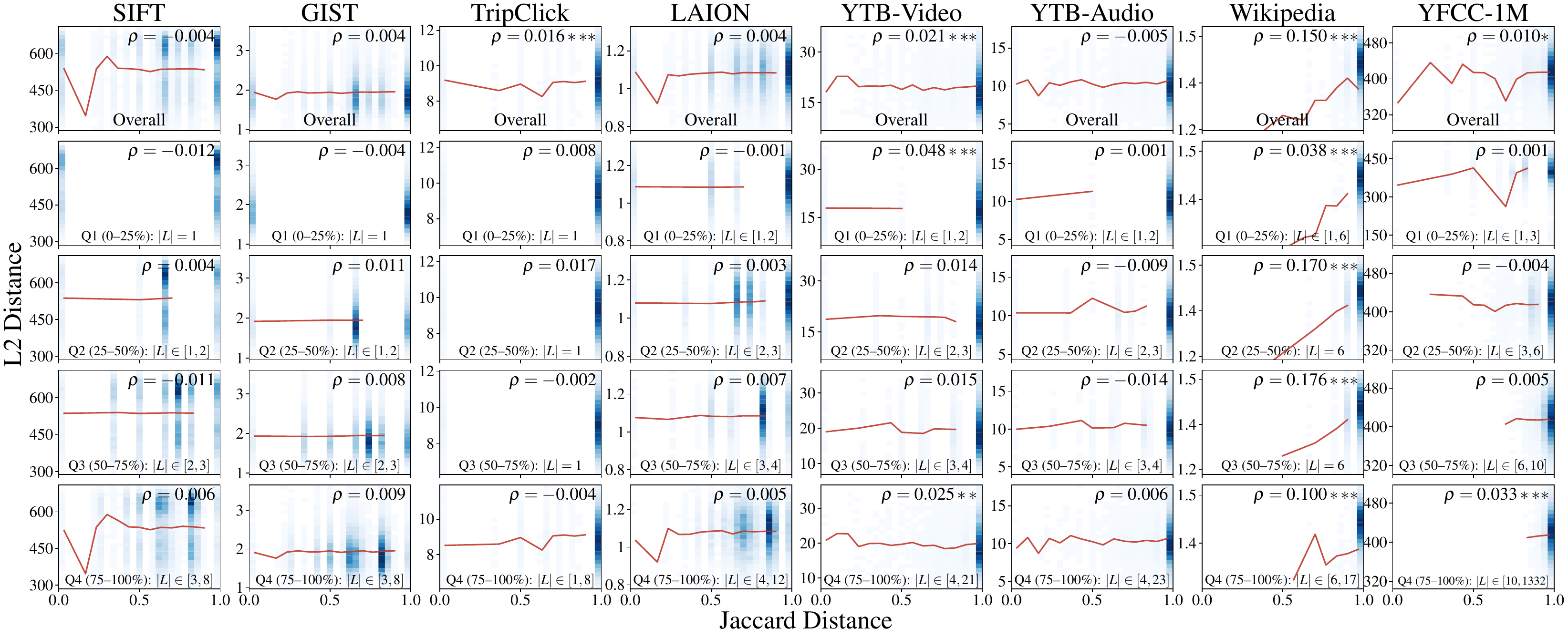}
    \caption{Correlation between pairwise Jaccard distance and L2 distance. Darker regions indicate more object pairs. The red curve shows $\mathbb{E}[\delta_v(v_1,v_2)\,|\,\delta_L(l(v_1),l(v_2))]$, and Spearman's coefficient $\rho$ is marked at the top right corner.}
    \Description{Correlation study of pair-wise Jaccard distance and L2 distance}
    \label{fig:corr-vec}
\end{figure*}

\section{Additional Experiments}
\label{app:all_exps}

\bparagraph{Approximation to the oracle on other datasets.}
Fig.~\ref{fig:other_dc_recall} compares distance computations (DC) against the oracle HNSW on the remaining four datasets.
On SIFT and TripClick, the label space is small and containment filters are highly selective, so all methods operate on relatively small filtered subsets and exhibit similar DC behavior close to the oracle.
In contrast, the gap widens on YTB-Video and YFCC-1M, whose larger label spaces make cross-label navigation more difficult.
LSSG is the closest approximation among LFANNS indices, indicating better navigation efficiency under large-scale label constraints.

\bparagraph{Selectivity breakdown on other datasets.}
Fig.~\ref{fig:other_qps_sel} reports QPS across selectivity percentiles on the remaining four datasets.
On SIFT and TripClick where the label universe is small, ELI, UNG, and LSSG are broadly comparable across percentiles.
On YFCC-1M, ELI fails to construct and UNG cannot reach 0.9 recall across all selectivity percentiles in the overlap semantic.
ACORN performs strongly for overlap on low-selectivity percentiles, consistent with its connectivity advantage when a large fraction of candidates pass the filter.

\begin{figure}[!h]
    \centering
    \includegraphics[width=\linewidth]{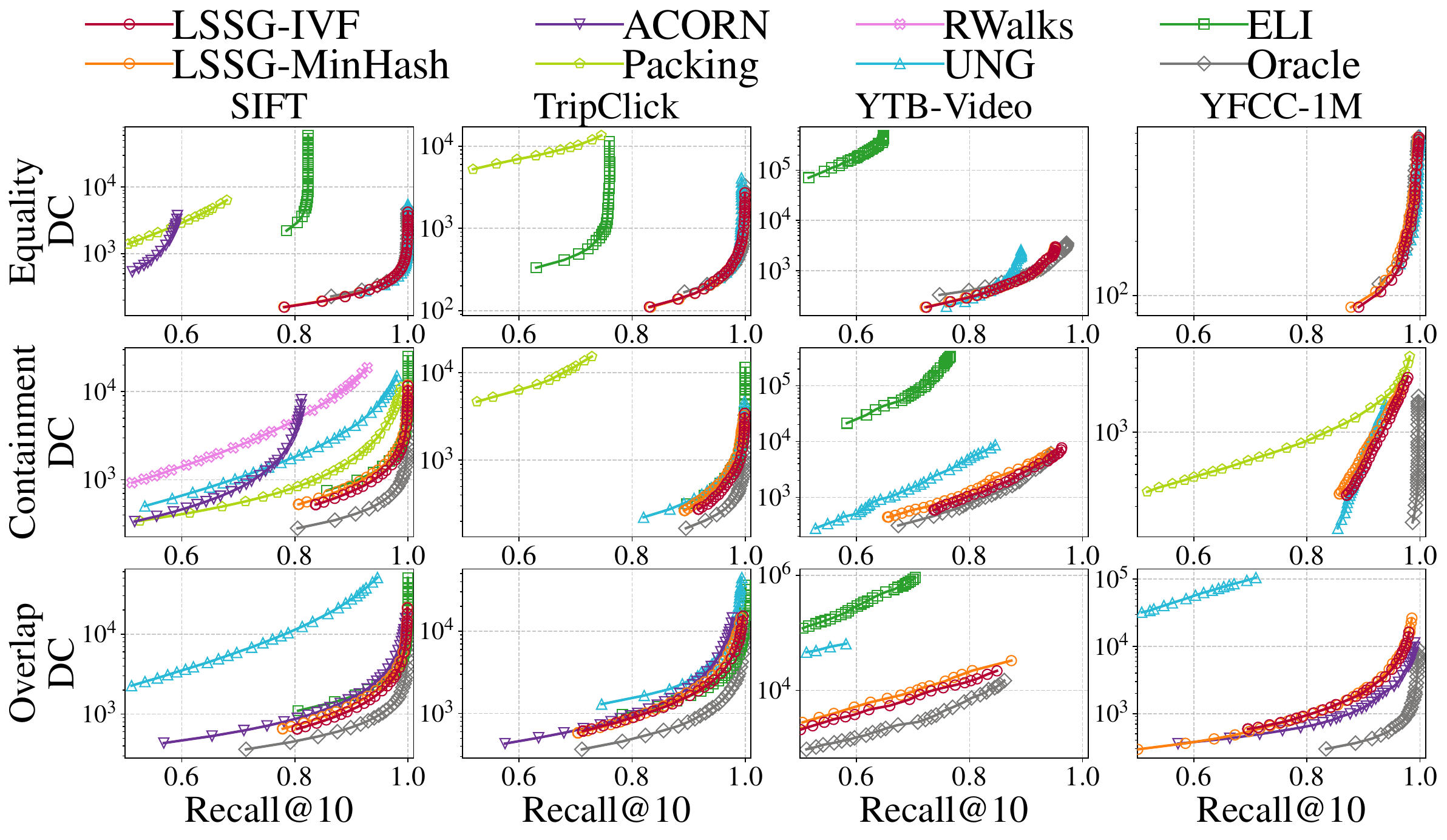}
    \caption{DC of LFANNS indices on other datasets}
    \Description{DC of LFANNS indices on other datasets}
    \label{fig:other_dc_recall}
\end{figure}

\begin{figure}
    \centering
    \includegraphics[width=\linewidth]{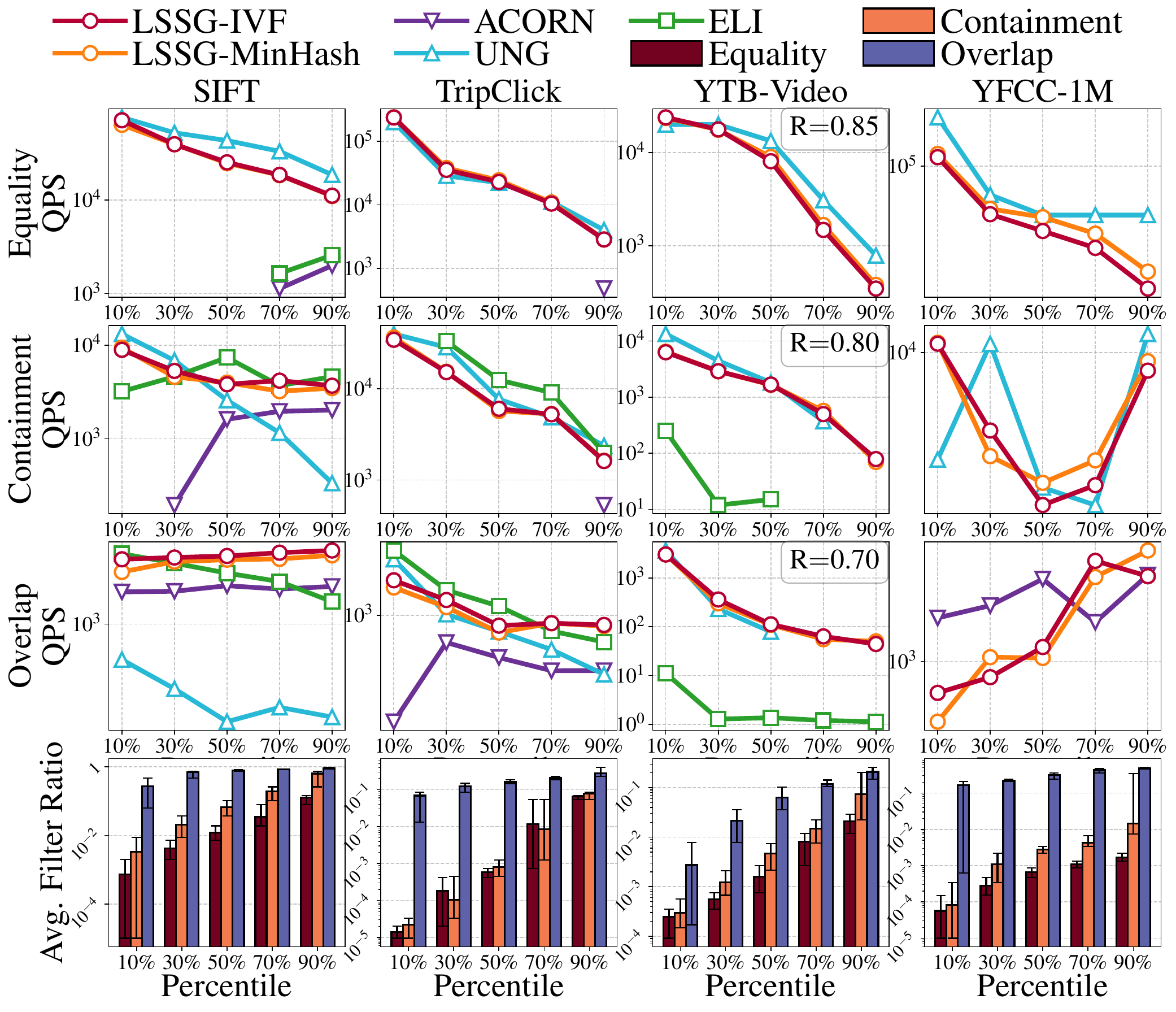}
    \caption{Selectivity breakdown on other datasets}
    \Description{Selectivity breakdown on SIFT, TripClick, and YFCC-1M}
    \label{fig:other_qps_sel}
\end{figure}

\begin{table}[t!]
\centering
\caption{Indexing performance of ANNS systems}
\label{tab:system-indexing}
\setlength{\tabcolsep}{2.5pt}
\resizebox{\columnwidth}{!}{%
\begin{tabular}{lrrrrrrrr}
\toprule
\multicolumn{9}{c}{\textbf{Index size (MB)}} \\
\midrule
\textbf{Indices} & SIFT & GIST & Trip & LAION & Video & Audio & Wiki & YFCC \\
\midrule
FAISS-IVF
& 13 & 41 & 37 & 26 & 47 & 27 & 53 & 12 \\
FAISS-HNSW
& 138 & 134 & 145 & 136 & 672 & 672 & 688 & 135 \\
VSAG-IVF
& 207 & 137 & 117 & 155 & 148 & 236 & 170 & 138 \\
VSAG-HNSW
& 177 & 203 & 233 & 246 & 209 & 867 & 810 & 186 \\
\midrule
Milvus-IVF
& 8 & 11 & 11 & 10 & 12 & 39 & 40 & 8 \\
Milvus-HNSW
& 488 & 488 & 516 & 488 & 488 & 2{,}441 & 2{,}441 & 488 \\
PGVector-IVF
& 37 & 304 & 1{,}037 & 656 & 3{,}914 & 167 & 493 & 58 \\
PGVector-HNSW
& 478 & 4{,}016 & 1{,}024 & 638 & 3{,}905 & 2{,}453 & 2{,}793 & 505 \\
\midrule
LSSG-IVF
& 311 & 204 & 199 & 271 & 136 & 1{,}069 & 602 & 239 \\
LSSG-MinHash
& 296 & 192 & 200 & 277 & 142 & 1{,}065 & 606 & 303 \\
\toprule
\multicolumn{9}{c}{\textbf{Indexing time (s)}} \\
\midrule
\textbf{Indices} & SIFT & GIST & Trip & LAION & Video & Audio & Wiki & YFCC \\
\midrule
FAISS-IVF
& 1{,}546 & 9{,}107 & 8{,}043 & 4{,}880
& 9{,}994 & 3{,}507 & 10{,}920 & 1{,}851 \\
FAISS-HNSW
& 30 & 149 & 129 & 71 & 74 & 192 & 400 & 42 \\
VSAG-IVF
& 94 & 473 & 350 & 253 & 595 & 250 & 457 & 119 \\
VSAG-HNSW
& 31 & 82 & 79 & 46 & 59 & 163 & 311 & 29 \\
\midrule
Milvus-IVF
& 84 & 438 & 387 & 269 & 440 & 446 & 968 & 103 \\
Milvus-HNSW
& 87 & 678 & 350 & 279 & 410 & 483 & 1{,}257 & 132 \\
PGVector-IVF
& 90 & 631 & 583 & 356 & 750 & 447 & 1{,}237 & 180 \\
PGVector-HNSW
& 238 & 1{,}092 & 923 & 544
& 1{,}128 & 1{,}136 & 2{,}841 & 276 \\
\midrule
LSSG-IVF
& 133 & 241 & 270 & 181 & 287 & 946 & 487 & 379 \\
LSSG-MinHash
& 122 & 224 & 233 & 179 & 243 & 721 & 426 & 111 \\
\bottomrule
\end{tabular}%
}
\end{table}

\bparagraph{Indexing performance of ANNS systems.}
Table~\ref{tab:system-indexing} illustrates the index size and indexing time of the systems compared in Fig.~\ref{fig:vdb-compare}.
Across eight datasets, LSSG incurs moderate construction costs, with median index size and indexing time of 287~MB and 229~s, respectively.
FAISS and VSAG HNSW variants build quickly with moderate space, whereas IVF variants are compact but substantially slower to construct (median time: 302--6{,}462~s).
Milvus and PGVector provide integrated filtering but often incur larger index size and comparable or higher construction time (median: 380--1{,}008~s).
The two LSSG variants have similar costs, and offer a balanced construction trade-off while achieving stronger LFANNS performance.

\begin{figure}
    \centering
    \includegraphics[width=\linewidth]{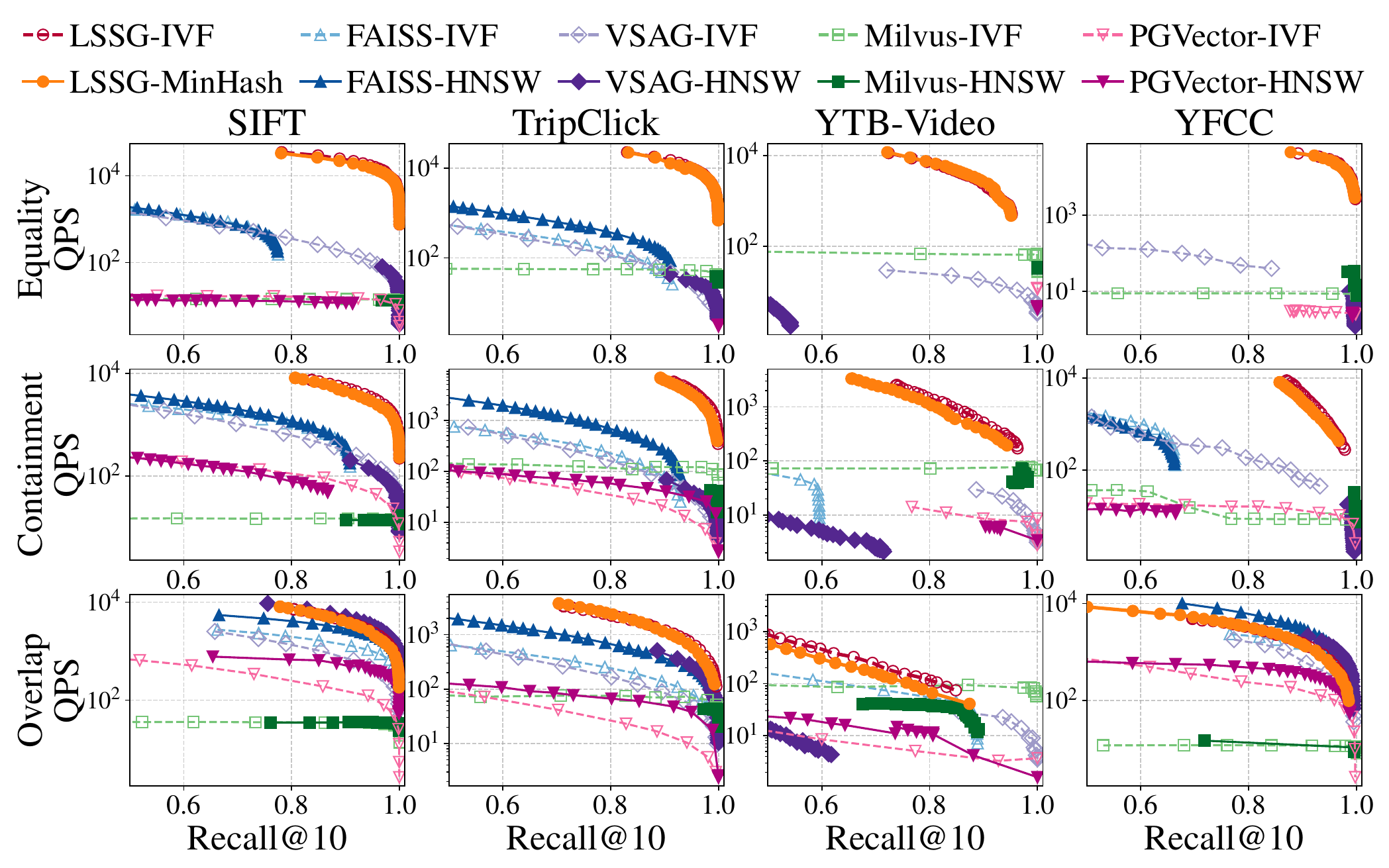}
    \caption{Query performance of ANNS systems on other datasets}
    \Description{Query performance of ANNS systems on other datasets}
    \label{fig:vdb-compare-other}
\end{figure}

\bparagraph{Query performance of ANNS systems on other datasets.}
Fig.~\ref{fig:vdb-compare-other} shows the system results on SIFT, TripClick, YTB-Video, and YFCC-1M.
FAISS remains recall-limited, attaining at most 0.13--0.91 for equality and 0.60--0.93 for containment across these datasets.
On YTB-Video overlap, it reaches only 0.89 recall.
At 0.95 recall for containment, the best VSAG, Milvus, and PGVector variants achieve 18--136, 17--122, and 9--44~QPS, respectively, while LSSG remains on the upper QPS--recall frontier.
Although library implementations can be competitive for overlap on SIFT, TripClick, and YFCC-1M, their performance varies substantially across workloads and degrades on YTB-Video.
Overall, LSSG provides more consistent high-recall performance across datasets and filter semantics.

\begin{figure}
    \centering
    \includegraphics[width=\linewidth]{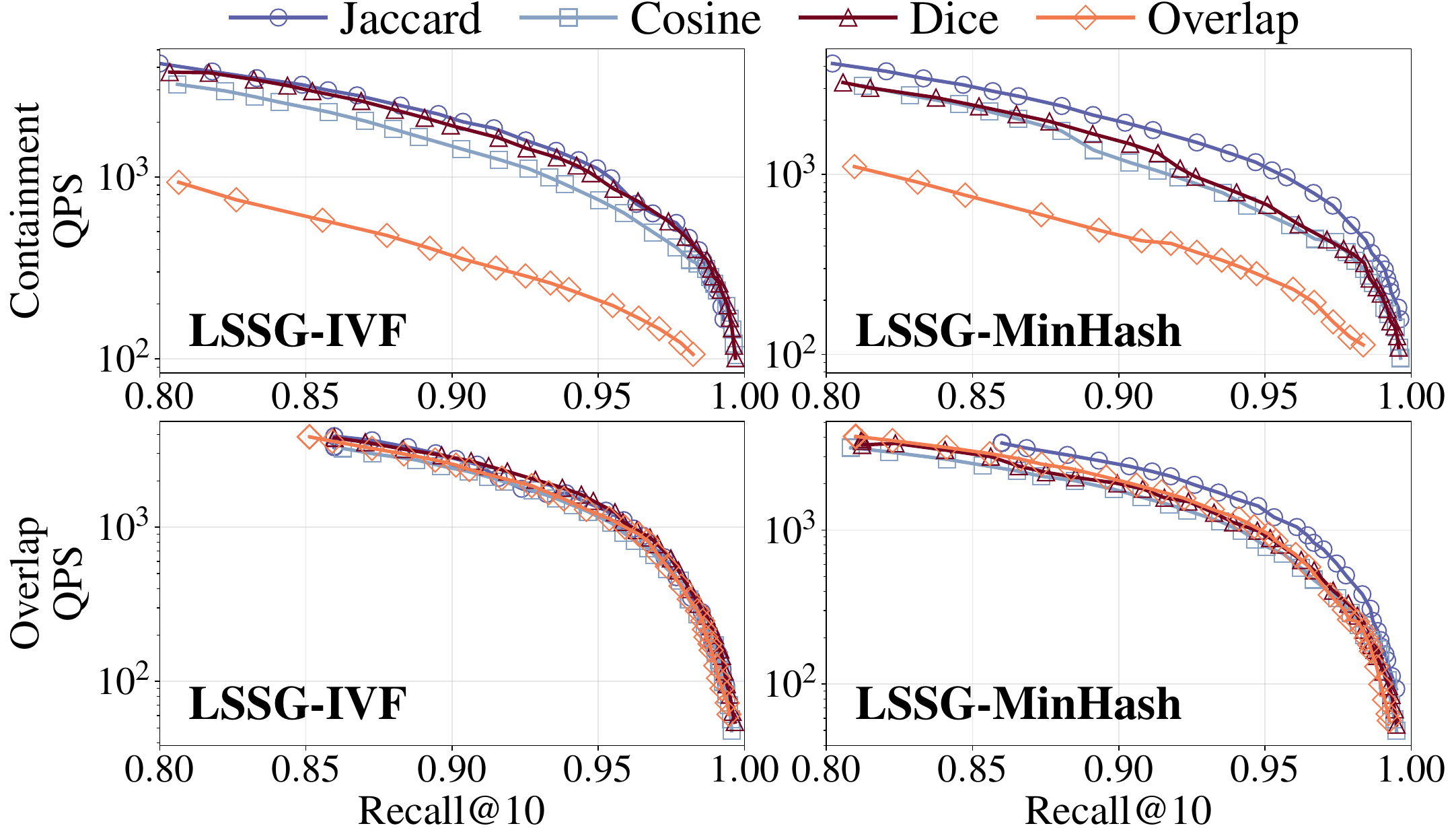}
    \caption{Comparison of different label set distances}
    \Description{Comparison of different label set distances}
    \label{fig:setdistance}
\end{figure}

\bparagraph{Impact of different label set distances.}
Fig.~\ref{fig:setdistance} compares tiering on LAION using distances induced by alternative label set similarities, including cosine ($\frac{|L_1 \cap L_2|}{\sqrt{|L_1|\times|L_2|}}$), Dice ($\frac{2|L_1\cap L_2|}{|L_1|+|L_2|}$), and overlap ($\frac{|L_1\cap L_2|}{\min(|L_1|,|L_2|)}$).
\begin{enumerate*}[left=0pt]
    \item Jaccard is the best fit for tiering.
    Stratification by overlap distance degrades containment because it provides weaker granularity for separating tiers, underutilizing higher tiers and weakening progressive navigation.
    
    \item MinHash benefits from metric alignment.
    Non-Jaccard distances hurt LSSG-MinHash more than LSSG-IVF, since MinHash is an unbiased estimator for Jaccard similarity but not for the other distance measures.
\end{enumerate*}

\begin{figure}
    \centering
    \includegraphics[width=0.96\linewidth]{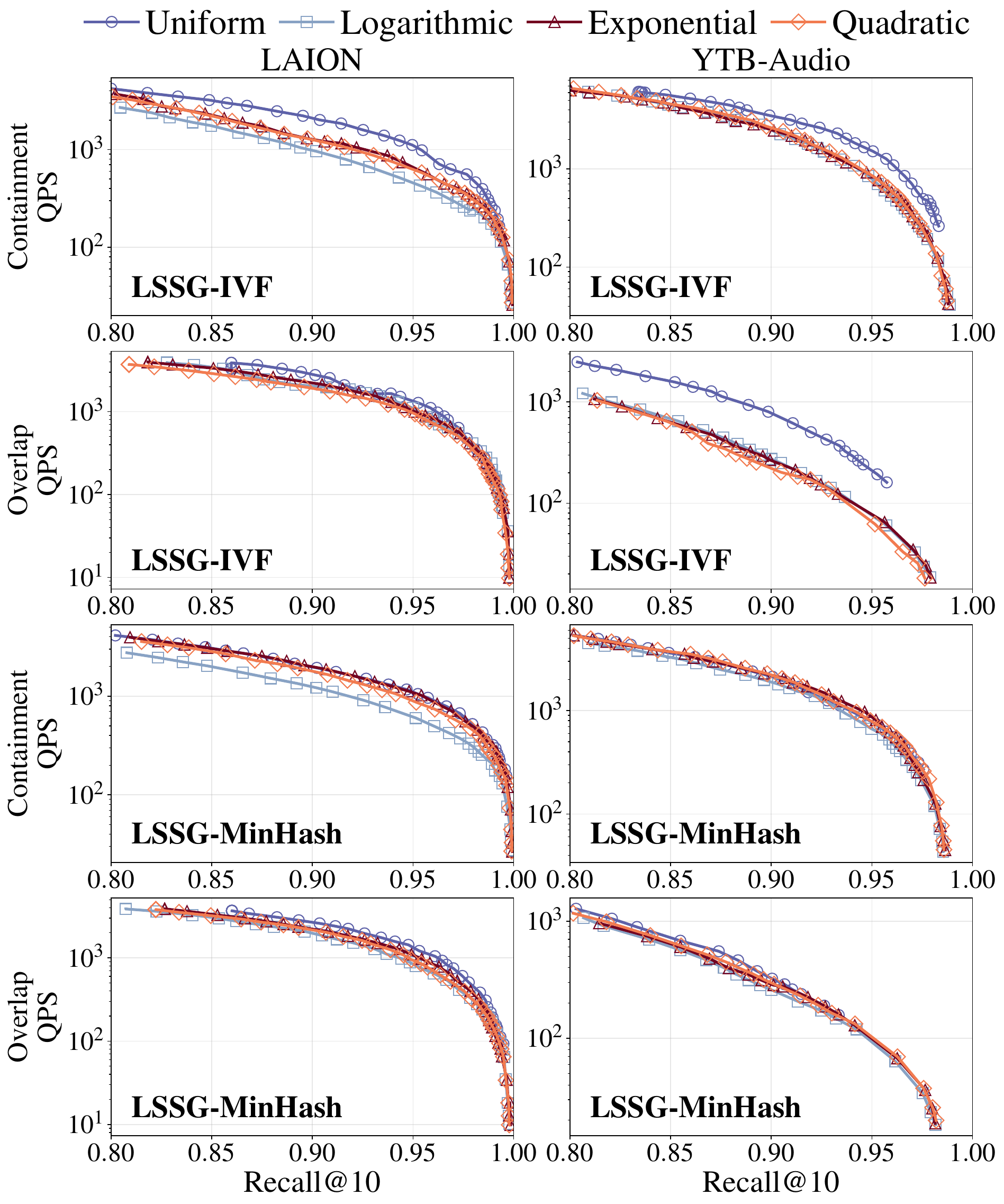}
    \caption{Impact of Jaccard distance threshold distribution}
    \Description{Impact of Jaccard distance threshold distribution}
    \label{fig:jc_dist}
\end{figure}

\bparagraph{Impact of the Jaccard threshold distribution.}
In Def.~\ref{def:lssg}, we use uniformly spaced Jaccard-distance thresholds across tiers.
We additionally evaluate logarithmic ($\theta_t=1-\frac{\log(t)}{\log(T)}$), exponential ($\theta_t=1-\frac{\exp(\frac{t-1}{T-1})-1}{e-1}$), and quadratic ($\theta_t=1-(\frac{t-1}{T-1})^2$) schedules.
As shown in Fig.~\ref{fig:jc_dist}, LSSG-MinHash is robust across all threshold schedules, suggesting that its probing step tolerates moderate changes in tier granularity.
LSSG-IVF is more sensitive.
Non-uniform schedules can over-emphasize a narrow similarity range, while the uniform schedule covers a broader spectrum of label set similarity and yields better containment and overlap performance in our experiments.

\begin{figure}
\centering
\begin{subfigure}[t]{\linewidth}
\centering
\includegraphics[width=\linewidth]
{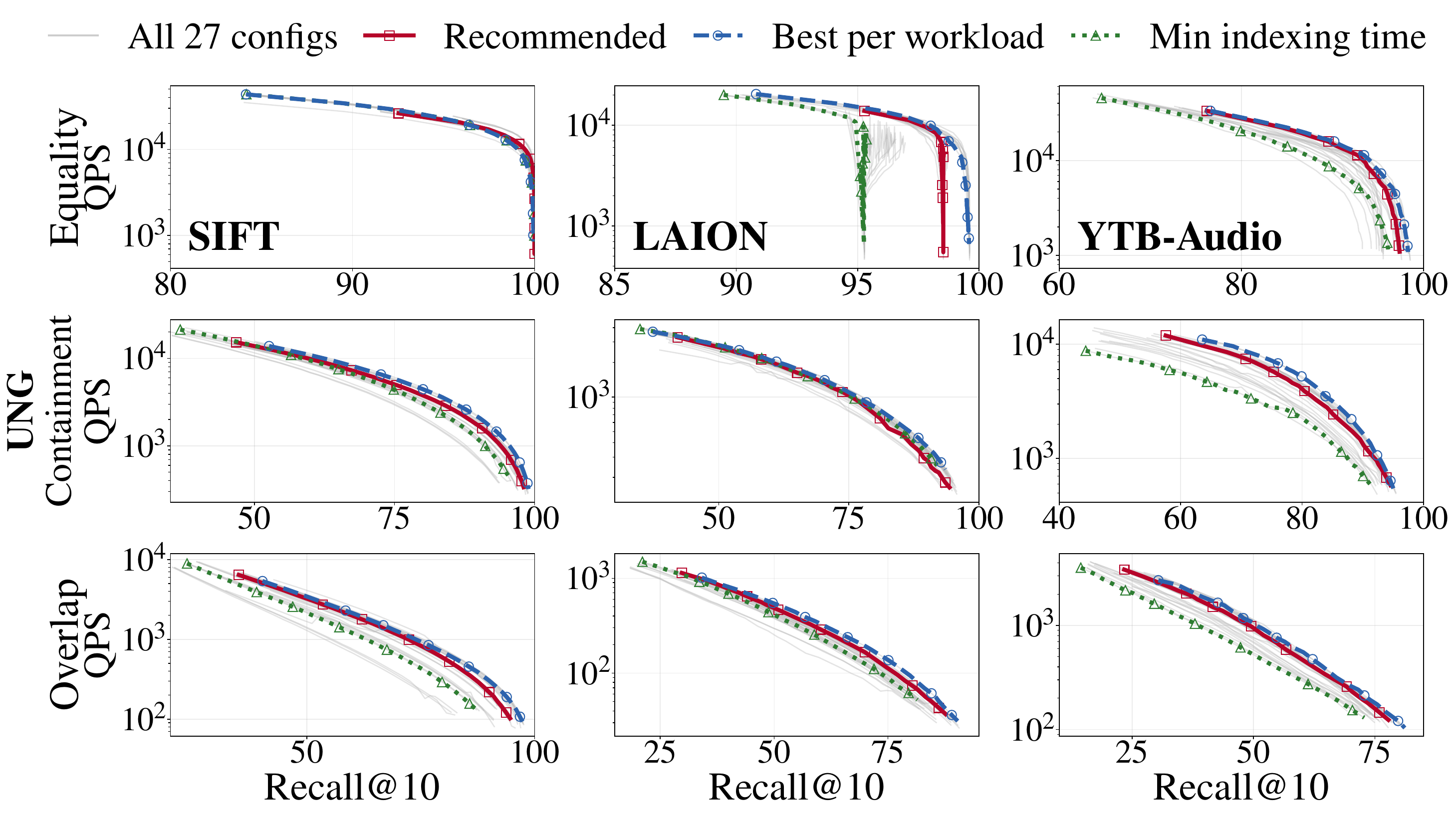}
\label{fig:ung_param_audit}
\end{subfigure}

\par\vspace{-4pt}

\begin{subfigure}[t]{\linewidth}
\centering
\includegraphics[width=\linewidth]
{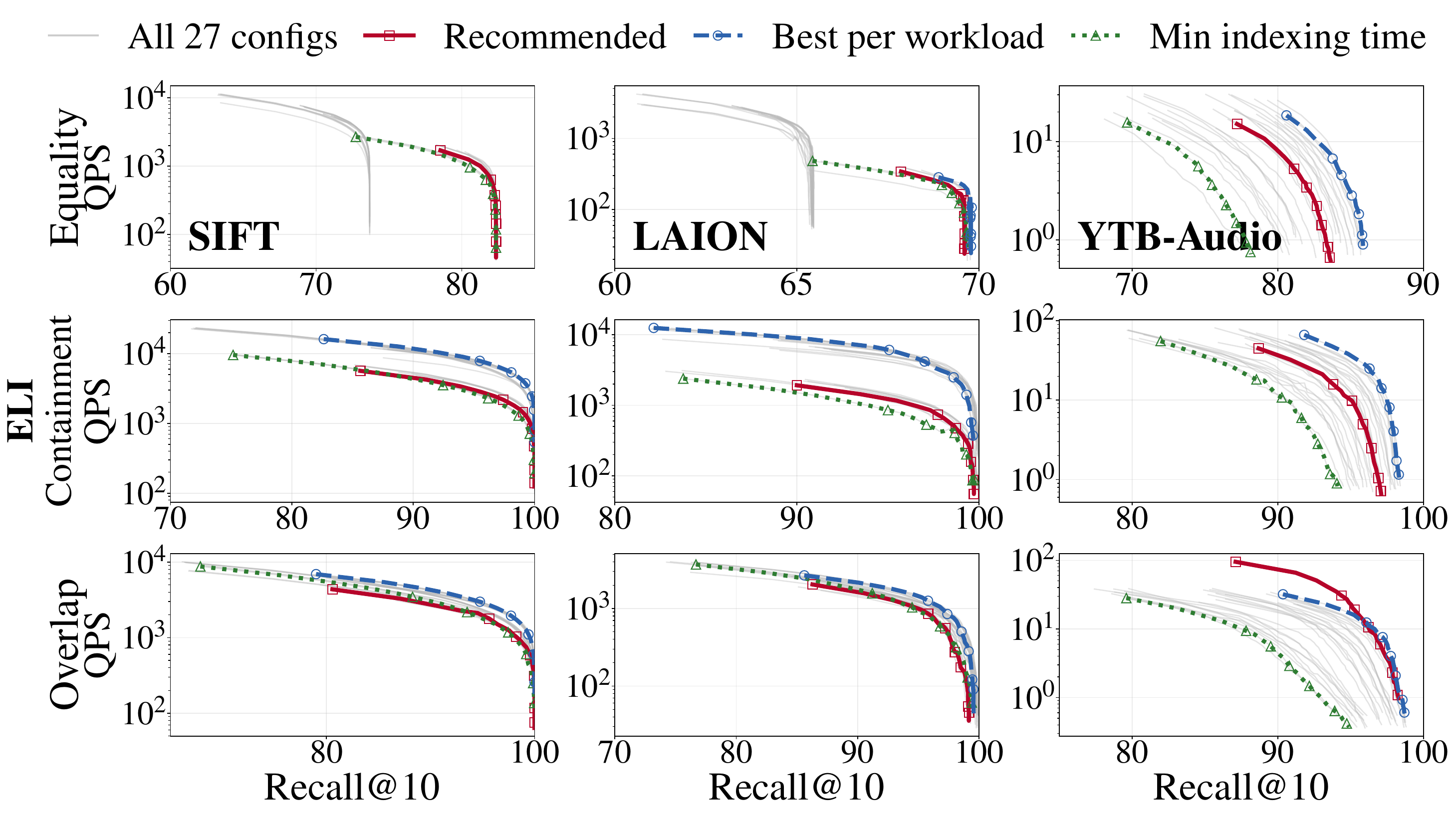}
\label{fig:eli_param_audit}
\end{subfigure}
\vspace{-4pt}

\caption{Query-side parameter audits}
\Description{Query-side parameter audits}
\label{fig:baseline_param_audit}
\end{figure}

\bparagraph{Baseline parameter audit.}
To ensure fair baseline tuning, we evaluate all 27 construction configurations of UNG and ELI on SIFT, LAION, and YTB-Audio.
For UNG, we sweep $M\in\{16,32,64\}$, $L_{\mathrm{build}}\in\{50,100,200\}$, and $\delta\in\{3,6,12\}$. 
For ELI, we sweep $M\in\{8,16,32\}$, $\mathrm{efc}\in\{128,200,256\}$, and $e\in\{0.2,0.4,0.6\}$.
Fig.~\ref{fig:baseline_param_audit} shows all swept results and highlights the recommended, query-best, and minimum indexing time settings.
Table~\ref{tab:baseline_build_param_audit} reports the construction costs of the recommended and query-best settings.

For containment, tuning UNG improves QPS by 1.23--1.33x, but the gains are accompanied by construction or cross-semantic penalties.
The query-best settings require 1.12--2.25x indexing time of the recommended setting.
Index size increases by 83\% on SIFT and 37\% on YTB-Audio.
They also degrade other semantics: for equality, QPS drops by 13\% on SIFT and 14\% on YTB-Audio, while on LAION the maximum recall decreases from 0.877 to 0.832.
For overlap, the best settings show the same trade-off.
On SIFT, the 1.40x QPS gain reduces that for equality and containment by 21\% and 23\%, while increasing index size and indexing time by 55\% and 18\%.
On LAION and YTB-Audio, the maximum recall for overlap increases from 0.877/0.778 to 0.903/0.812, but QPS for equality decreases by 13\%/14\% and index size increases by 28\%/63\%; YTB-Audio also requires 30\% longer indexing.
Thus, UNG's query-best settings improve the target workloads by trading off construction cost or performance under other filter semantics, rather than providing a uniformly stronger replacement.

ELI exhibits a clearer dependence on both dataset and filter semantics.
The recommended $e=0.2$ gives the best accuracy for equality on SIFT and LAION, whereas containment favors $e=0.6$ and overlap favors $e=0.4$ on these datasets.
Tuning improves QPS by 2.08--3.95x for containment and by at most 1.64x for overlap, but reduces maximum recall for equality by 10.6\% on SIFT and 6.3\% on LAION.
These settings are therefore not uniformly stronger replacements.
Even $(16,128,0.2)$, which retains $e=0.2$ and improves several SIFT/LAION workloads, loses on YTB-Audio overlap, where the recommended setting achieves the best result.
The best setting for containment on YTB-Audio instead increases indexing time by 23\%.
Thus, no alternative among the 27 configurations consistently improves the recommended setting across all datasets and semantics.

In summary, using the query-best audited configurations changes some point-wise winners but not the overall cross-semantics comparison.
LSSG-MinHash remains 1.33--5.24x faster than query-best UNG on all three containment workloads, while no audited UNG configuration reaches 0.95 recall for overlap on LAION or YTB-Audio.
Query-best ELI is 2.63x and 4.49x faster than LSSG-MinHash on SIFT and LAION for containment and matches LSSG-MinHash on SIFT for overlap.
Conversely, LSSG-MinHash is 1.20x faster on LAION for overlap and 32.0x/10.1x faster on YTB-Audio for containment/overlap.
Thus, tuning changes individual winners but yields no baseline with consistent dominance across the audited datasets and semantics compared.
LSSG retains the advantage of consistently strong performance across all semantics with one index.

\begin{table}
    \centering
    \caption{Parameter audits of UNG and ELI.
    ``Rec.'' is the index size and indexing time on the recommended settings. 
    The recommend settings are $(32,100,6)$ for UNG and $(16,200,0.2)$ for ELI.
    ``Query-best'' is the setting with the best query performance per workload (blue curves in Fig.~\ref{fig:baseline_param_audit}).}
    \label{tab:baseline_build_param_audit}
    \small
    \setlength{\tabcolsep}{2.5pt}
    \resizebox{\columnwidth}{!}{%
    \begin{tabular}{lllccc}
        \toprule
        \textbf{Indices} & \textbf{Datasets} & \textbf{Sem.}
        & \textbf{\makecell{Rec.\\(MB / s)}}
        & \textbf{\makecell{Query-best\\setting}}
        & \textbf{\makecell{Query-best\\(MB / s)}} \\
        \midrule
        \multirow{9}{*}{\textbf{UNG}}
        & \multirow{3}{*}{SIFT}
        & Eq.   & 228 / 76   & $(16,50,3)$   & 169.6 / 33.7 \\
        & & Cont. & 228 / 76   & $(64,200,6)$  & 417.6 / 170.8 \\
        & & Ovlp. & 228 / 76   & $(64,100,3)$  & 353.1 / 89.5 \\
        \cmidrule(l){2-6}
        & \multirow{3}{*}{LAION}
        & Eq.   & 271 / 189  & $(16,50,12)$  & 234.7 / 87.9 \\
        & & Cont. & 271 / 189  & $(16,200,6)$  & 193.3 / 211.2 \\
        & & Ovlp. & 271 / 189  & $(64,50,12)$  & 346.3 / 126.5 \\
        \cmidrule(l){2-6}
        & \multirow{3}{*}{YTB-Audio}
        & Eq.   & 954 / 1,422 & $(32,100,12)$ & 1,220.5 / 1,716.4 \\
        & & Cont. & 954 / 1,422 & $(64,200,6)$  & 1,303.5 / 1,870.6 \\
        & & Ovlp. & 954 / 1,422 & $(64,200,12)$ & 1,552.8 / 1,843.1 \\
        \midrule
        \multirow{9}{*}{\textbf{ELI}}
        & \multirow{3}{*}{SIFT}
        & Eq.   & 419 / 98   & $(16,200,0.2)$ & 419 / 98 \\
        & & Cont. & 419 / 98   & $(16,128,0.6)$ & 615.0 / 88.8 \\
        & & Ovlp. & 419 / 98   & $(16,128,0.4)$ & 615.0 / 88.8 \\
        \cmidrule(l){2-6}
        & \multirow{3}{*}{LAION}
        & Eq.   & 745 / 257  & $(32,256,0.2)$ & 1,331.2 / 399.1 \\
        & & Cont. & 745 / 257  & $(8,128,0.6)$  & 597.0 / 182.7 \\
        & & Ovlp. & 745 / 257  & $(16,128,0.4)$ & 968.0 / 245.0 \\
        \cmidrule(l){2-6}
        & \multirow{3}{*}{YTB-Audio}
        & Eq.   & 92,693 / 807 & $(32,256,0.4)$ & 94,141.4 / 912.1 \\
        & & Cont. & 92,693 / 807 & $(32,256,0.6)$ & 94,915.3 / 989.3 \\
        & & Ovlp. & 92,693 / 807 & $(32,256,0.6)$ & 94,915.3 / 989.3 \\
        \bottomrule
    \end{tabular}
    }
\end{table}

\end{document}